\pdfoutput=1

\documentclass[%
    11pt,
    article,
    showpacs,preprintnumbers,
    showkeys,
    nofootinbib,
    amsmath,amssymb,
    aps,
    prb,
    floatfix,
    longbibliography,
]{revtex4-2}

\usepackage[english]{babel}
\usepackage[utf8]{inputenc}
\usepackage{lmodern}
\usepackage{graphicx}
\usepackage{bm}
\usepackage{hyperref}
\usepackage{xcolor}
\usepackage{epsfig}
\usepackage{amsmath}
\usepackage{amsfonts}
\usepackage{enumitem} 
\usepackage{amssymb}
\usepackage{amsthm}
\usepackage{tikz}
\usetikzlibrary{fit, positioning, shapes.geometric, arrows.meta, calc, backgrounds}
\usepackage{dsfont}
\usepackage[all]{nowidow}
\usepackage{pifont}
\usepackage{algorithm}
\usepackage{algpseudocode}
\usepackage[capitalise]{cleveref}
\usepackage[export]{adjustbox}
\usepackage{xspace}
\usepackage{nicefrac}

\usepackage{doi}
\usepackage{lipsum}
\usepackage{setspace}
\usepackage{titlesec}
\usepackage[normalem]{ulem}
\renewcommand{\thesubsection}{\arabic{subsection}}

\usepackage[letterspace=130]{microtype}

\definecolor{refColor}{HTML}{087ECC}
\definecolor{figColor}{HTML}{CC082D}
\definecolor{urlColor}{HTML}{08CC91}

\hypersetup{
    hypertexnames=false,           
    unicode=false,                 
    pdftoolbar=true,               
    pdfmenubar=true,               
    pdffitwindow=false,            
    pdfstartview={FitH},           
    pdftitle={},                   
    pdfauthor={Nicolai Lang},      
    pdfsubject={Quantum physics},  
    pdfcreator={Nicolai Lang},     
    pdfproducer={Nicolai Lang},    
    pdfkeywords={},                
    pdfnewwindow=true,             
    colorlinks=true,               
    linkcolor=figColor,            
    citecolor=refColor,            
    filecolor=magenta,             
    urlcolor=urlColor              
}

\newcommand{\bra}[1]{\mathinner{\langle{#1}|}}
\newcommand{\ket}[1]{\mathinner{|{#1}\rangle}}
\newcommand{\braket}[1]{\mathinner{\langle{#1}\rangle}}

\renewcommand{\vec}[1]{\mathbf{#1}}

\renewcommand{\vec}[1]{{\boldsymbol{#1}}}

\newcommand{\Tr}[2]{\operatorname{Tr}_{#2}\left[#1\right]}
\newcommand{\E}{\mathbb{E}}

\newcommand{\com}[2]{\left[#1,#2\right]}
\newcommand{\comm}[2]{\left[#1,#2\right]}

\renewcommand{\d}{\mathrm{d}}
\renewcommand{\H}{\mathcal{H}}
\newcommand{\NL}{\mathfrak{S}_L}
\newcommand{\N}{\mathbf{N}}

\newcommand{\B}{\mathbb{B}}
\renewcommand{\S}{\mathbb{S}}
\renewcommand{\L}{\mathbb{L}}
\renewcommand{\P}{\mathbb{P}}
\renewcommand{\d}{d}
\renewcommand{\E}{\mathbb{E}}
\newcommand{\V}{\mathbb{V}}

\newcommand{\GG}{\mathfrak{G}}
\newcommand{\GB}{\mathfrak{B}}
\newcommand{\GH}{\mathfrak{H}}

\newcommand{\vsigma}{\vec{\sigma}}
\newcommand{\cL}{\textsf{const}_L}
\newcommand{\G}{\mathcal{G}}
\newcommand{\barNL}{\mathfrak{S}^*_{2L}}
\newcommand{\barN}{\mathbf{N}^*}
\newcommand{\barn}{{n^*}}
\newcommand{\barB}{\mathbb{B}^*}
\newcommand{\s}{{\pmb{1}}}

\newcommand{\nO}{{n_0}}

\renewcommand{\O}{\Omega}
\newcommand{\OSTA}{\O_0}
\newcommand{\OGGS}{\O_2}
\newcommand{\OOPA}{\O_1}
\newcommand{\OSPB}{\O_3}
\newcommand{\OOLA}{\O_1'}
\newcommand{\OOLB}{\O_2'}
\newcommand{\OOLD}{\O_4'}
\newcommand{\OSLE}{\O_5'}
\newcommand{\OORC}{\O_3''}

\newcommand{\LOTA}{L_0}
\newcommand{\LGGS}{L_2}
\newcommand{\LOPA}{L_1}
\newcommand{\LOPB}{L_3}
\newcommand{\LOLA}{L_1'}
\newcommand{\LOLB}{L_2'}
\newcommand{\LOLD}{L_4'}
\newcommand{\LOLE}{L_5'}
\newcommand{\LORC}{L_3''}

\newcommand{\DE}{\Delta E}
\newcommand{\lb}{\mathrm{bd}}
\newcommand{\DElb}{\Delta E_\lb}

\newcommand{\Id}{\mathds{1}}
\newcommand{\spn}[1]{\operatorname{span}\left\{\,#1\,\right\}}
\newcommand{\p}{{\hat p}}
\newcommand{\hP}{{P}}
\newcommand{\hlf}{\tfrac{1}{2}}
\DeclareRobustCommand{\etal}{\textit{et~al.}\@}

\newcommand{\Deltamax}{\Delta_\mathrm{max}}
\newcommand{\Deltamin}{\Delta_\mathrm{min}}
\newcommand{\group}[1]{\mathrm{#1}}
\newcommand{\vep}{\varepsilon}

\newtheorem{theorem}{Theorem}
\newtheorem{lemma}{Lemma}
\newtheorem{applem}{Lemma}

\newtheorem{definition}{Definition}

\newtheorem{proposition}{Proposition}

\crefname{notation}{Notation}{Notations}
\Crefname{notation}{Notation}{Notations}
\crefname{applem}{Lemma}{Lemmas}
\Crefname{applem}{Lemma}{Lemmas}
\crefrangeformat{applem}{Lemmas~#3#1#4 to~#5#2#6}
\Crefrangeformat{applem}{Lemmas~#3#1#4 to~#5#2#6}
\crefrangemultiformat{applem}%
    {Lemmas~#3#1#4 to~#5#2#6}%
    {, #3#1#4 to~#5#2#6}%
    {, #3#1#4 to~#5#2#6}%
    {, and~#3#1#4 to~#5#2#6}
\Crefrangemultiformat{applem}%
    {Lemmas~#3#1#4 to~#5#2#6}%
    {, #3#1#4 to~#5#2#6}%
    {, #3#1#4 to~#5#2#6}%
    {, and~#3#1#4 to~#5#2#6}

\graphicspath{{./fig/}{./data/}}
\makeatletter\renewcommand{\bibfont}{\footnotesize\@clubpenalty\clubpenalty}\makeatother

\begin{document}

\title{Excitation gap of a blockade structure with \texorpdfstring{$\mathbb{Z}_2$}{Z2} topological order}

\author{Simon Fell}
\email{simon.fell@itp3.uni-stuttgart.de}
\author{Tobias F. Maier}
\author{Hans Peter Büchler}
\author{Nicolai Lang}

\affiliation{%
    Institute for Theoretical Physics III 
    and Center for Integrated Quantum Science and Technology,\\
    University of Stuttgart, 70550 Stuttgart, Germany
}

\date{\today}


\begin{spacing}{1.0}
\begin{abstract}
    Mathematically rigorous statements on the spectral gap of quantum many-body
    systems in the thermodynamic limit are notoriously difficult to prove --
    yet they are of fundamental importance for classifying quantum phases of
    matter. Here we prove the existence of a finite excitation gap for a
    particular Hamiltonian which was proposed in
    \href{https://doi.org/10.1103/dtlf-2q82}{T.~F.~Maier~\etal, PRX Quantum \textbf{6}, 030340 (2025)}
    and is motivated by the Rydberg platform. The Hamiltonian exhibits only
    two-body blockade interactions between two-level systems and has a topologically
    ordered ground state in the toric code phase. We show that our result also
    applies to a broader class of blockade Hamiltonians which realize
    non-Abelian quantum double phases and were proposed in
    \href{https://doi.org/10.1103/99sy-kggw}{H.~P.~Büchler~\etal, Phys.~Rev.~B \textbf{114}, 065113 (2026)}.
    The proof builds on known gap stability results and exploits the local
    symmetry of the studied models.
\end{abstract}
\end{spacing}

\maketitle

\onehalfspacing

\section{Introduction}
\label{sec:introduction}

Proving the existence of spectral gaps for quantum many-body systems in the
thermodynamic limit is notoriously difficult~\cite{spectralgap1, spectralgap2,
spectralgap3, yarotsky2006ground, gosset2016local, lemm2019spectral,
lemm2020finite, koma1997spectral, bachmann2015product}. Yet the presence of a
gap is of fundamental importance for the classification of quantum phases of
matter and crucial for their stability. In particular, quantum phases with
topological order have recently sparked a lot of interest~\cite{klich2010stability, 
bravyi2010topological, bravyi2011short, michalakis2013, nachtergaele2022},
since their classification goes beyond Landau's symmetry-based paradigm and
they exhibit fundamentally new properties, such as robust ground state
degeneracies and fractional excitations. In two spatial dimensions, bosonic
topological orders are characterized by the fusion algebra and statistics of
their localized excitations (so-called \emph{anyons})
-- a structure that is captured by unitary modular tensor categories (UMTCs)~\cite{wang2010topological, 
rowell2009classification, kitaev2006anyons}.
Notably, a spectral gap in the thermodynamic limit is crucial for well-defined
localized excitations, and thereby for the applicability of this powerful
machinery. In this paper, we prove the existence of a gap for a particular
class of Hamiltonians that realize Abelian and non-Abelian topological orders
with physically realistic two-body interactions.

Exactly solvable commuting-projector Hamiltonians with many-body interactions
have proven useful for studying the renormalization fixpoint wave functions of
(non-chiral) bosonic topological phases. The paradigmatic example is the toric
code Hamiltonian proposed by Kitaev~\cite{kitaev2003fault}, which
realizes the topological order of the quantum double
$\mathcal{D}(\mathbb{Z}_2)$ in its ground state. This particular
commuting-projector Hamiltonian is exactly solvable, gapped, and relies on
nearest-neighbor multi-body interactions -- which makes it hard to realize.
By contrast, interacting many-body Hamiltonians that require only
\emph{two-body} interactions are physically more realistic -- but typically not
exactly solvable; in particular, their ground state is usually not a
renormalization fixpoint of the topological phase~\cite{bravyi2005commutative,
kitaev2006anyons, bombin2009quantum}. The study of such models is therefore
often restricted to approximate methods based on perturbation theory and
numerics.
Recently, some of us bridged this gap by engineering a Hamiltonian with only
two-body blockade interactions between two-level systems that, for weak
transverse fields, \emph{provably} realizes the toric code topological order
beyond the fixpoint~\cite{maier2025}. This construction was later generalized
to the complete family of quantum double phases~\cite{buchler2026quantum}. At
the fixpoint, the ground states of this family are equal-weight superpositions
(condensates) of extended objects that are characterized by local constraints;
for the toric code, the extended objects are simply \emph{loops}.
The construction of blockade-based Hamiltonians that realize these condensates
away from the fixpoint is then based on two pillars. First, the local
constraints are enforced energetically by two-body blockade interactions and
local ``self energies''.  
These interactions are constructed such that the resulting Hamiltonian features
\emph{local symmetries} acting on the faces of the underlying lattice. Second,
quantum fluctuations are introduced by a weak, uniform transverse field. The
symmetries then enforce the perturbative generation of off-diagonal
\emph{plaquette operators} which stabilize the equal-weight superpositions of
the excitation patterns that satisfy the local constraints.
A notable advantage of this class of Hamiltonians is that, despite not being
exactly solvable, their local symmetries provide enough leverage to prove
topological order in the many-body ground state~\cite{maier2025}. This result
rests on a theorem by Michalakis and Zwolak~\cite{michalakis2013}, which
establishes the stability of a gap in the thermodynamic limit under certain
conditions and for sufficiently small (but finite) perturbations. For the
proposed blockade Hamiltonians, this stability also implies the existence of a
gap, and an exponentially suppressed ground state splitting, \emph{within the
sectors} of the Hilbert space defined by the local symmetries.  Whether these
blockade Hamiltonians actually stabilize a \emph{gapped} quantum phase
therefore depends on whether or not the lowest-energy states in different
symmetry sectors are separated by a finite energy gap in the thermodynamic
limit. While such a gap is expected on the basis of a heuristic
Schrieffer--Wolff transformation~\cite{maier2025, buchler2026quantum}, it has
not yet been rigorously established.

Here we close this gap (pun intended) by proving the existence of a finite
excitation gap in the thermodynamic limit for the particular blockade structure
constructed in Ref.~\cite{maier2025}. Combined with the results presented
there, this finally proves that the blockade Hamiltonian realizes a gapped
topological phase with $\mathcal{D}(\mathbb{Z}_2)$ toric code order for weak
but finite transverse fields. In particular, it shows that the engineered local
symmetries are \emph{not} crucial for the existence of the topological order --
they are a convenient feature for analytical access. The proof is conceptually
split into two parts.
First, we employ the gap stability theorem \emph{within} each symmetry sector
to show that the lowest-energy states are separated by finite excitation gaps
from the excited states, with an exponentially suppressed ground state
splitting. In the second part, we focus on the symmetric sector (the trivial
representation of the local symmetries) and show that its dimensionality
naturally exceeds that of all other symmetry sectors. Combined with general
features of blockade Hamiltonians, we show that this implies a finite energy
shift separating the lowest-energy state in the symmetric sector from all other
symmetry sectors.  We then show that this shift provides a lower bound on the
excitation gap for small enough transverse fields. The main technical part of
the proof is dedicated to finding a lower bound on the energy shift that is
both positive and independent of the system size.

\begin{figure}[tbp]
    \scalebox{0.8}{\input{tree.tex}}
    \caption{%
    \textbf{Overview.} 
    Conceptual flow to establish the main result (\cref{the:ExcitationGap})
    which rests on
    \cref{pro:SpectrumWithinSS,pro:GlobalGS,pro:GapBetweenSS}.
    \cref{pro:SpectrumWithinSS,pro:GlobalGS} are corollaries of the
    more general \cref{lem:GapStability_cpy,lem:GroundState_cpy}, respectively.
    \cref{pro:GapBetweenSS} builds on
    \cref{lem:AsymmetricSectors_cpy,lem:LowerBound_cpy,lem:Induction_cpy}.
    The exposition in this paper is structured from top to bottom and left to
    right. The arrows indicate the dependence between the lemmas, propositions
    and the theorem. 
    Note that for ${\textsf{i}},{\textsf{j}} \in
    \{\textsf{\ref{lem:GapStability_cpy}}, \textsf{\ref{lem:GroundState_cpy}},
    \textsf{\ref{lem:LowerBound_cpy}}, \textsf{\ref{lem:Induction_cpy}}\}$ with
    ${\textsf{i}} < {\textsf{j}}$, Lemma {\textsf{i}} is used to prove Lemma
    {\textsf{j}} (arrows on the bottom). 
    The technical
    \cref{lem:CommutationsRelations,lem:OperatorNorms,lem:Bounds,lem:Stabilizer,lem:ConditionsForGapStability,lem:SpectralFlow}
    from the \ref{sec:Appendix} are not shown.
    }
    \label{fig:structure}
\end{figure}

The remainder of the paper is organized as follows. In \cref{sec:system} we
define the Hamiltonian under study and review key concepts that were introduced
in Refs.~\cite{maier2025, buchler2026quantum}. There we also fix the notation
that is used throughout the paper. In \cref{sec:structure} we lay out the main
goal of the paper and sketch the overall structure and strategy of the proof.
The subsequent sections follow the ``top-down'' approach illustrated in
\cref{fig:structure}, starting with a high-level formulation of the main
theorem in \cref{sec:MainResult}, with its proof based on three propositions. In
\cref{sec:Propositions} we present the proofs of these propositions, which, in
turn, rest on five lemmas. These lemmas are then proven in the technical
\cref{sec:TechnicalPart}. In \cref{sec:nonabelian} we generalize the proof to
both Abelian and non-Abelian quantum double models. We close in
\cref{sec:conclusion} with a discussion of the scope of the proof, potential
generalizations, and open problems. The \ref{sec:Appendix} provides derivations
of commutation relations and (in)equalities, and lengthy (but straightforward)
checks of assumptions used in the main text.

\section{System and definitions} 
\label{sec:system}

We start by defining the Hamiltonian we study throughout the paper. For a
motivation and explicit construction, we refer the reader to
Ref.~\cite{maier2025}. In this section we also introduce concepts and notation
used throughout the paper. 

\subsection{Blockade graphs, Hilbert space, and Hamiltonian} 
\label{subsec:hamiltonian}

Consider a finite, vertex-weighted, undirected graph $\GG \equiv (\V, \E, W)$
with vertex set $\V$, edge set $\E$, and positive vertex weights
$W:\V\to\mathbb{R}^+$ with $W:\V\ni i\mapsto\Delta_i$. We associate a Hilbert
space and a Hamiltonian to this graph as follows: 

We identify each vertex $i\in\V$ with a two-level system described by the
Hilbert space $\mathbb{C}_i^2$. We refer to these two-level systems as
\emph{qubits} throughout the paper, but emphasize that this term is a
placeholder for generic two-level systems which may be realized by electronic
states, spin states, mode occupancies etc. The states of the qubit $i$ are
denoted by $\ket{n}_i$ with $n\in\{0,1\}$; we refer to $\ket{0}_i$ as the
\emph{ground state} and $\ket{1}_i$ as the \emph{excited state}.  The full
Hilbert space is then given by $\H := \bigotimes_{i\in\V}\mathbb{C}^2_i$ with
natural basis $\ket{\vec{n}} := \ket{n_1, n_2, \ldots}$. 
The Pauli matrices acting on the qubit $i$ are denoted by $\sigma_i^\alpha$ for
$\alpha\in\{x,y,z\}$. In particular,
$\sigma_i^x=\ket{1}\bra{0}_i+\ket{0}\bra{1}_i$ and
$\sigma_i^z=\ket{0}\bra{0}_i-\ket{1}\bra{1}_i$. The number operator is given by
$\hat{n}_i := \tfrac{1}{2}(\Id-\sigma_i^z)=\ket{1}\bra{1}_i$ so that
$\hat{n}_i\ket{\vec{n}}=n_i\ket{\vec{n}}$.  
We then define the \emph{(classical) blockade Hamiltonian} acting on $\H$ as
\begin{align}
    H_0
    :=
    \overbrace{U\!\!\sum_{\{i,j\} \in \E}\!\!\hat{n}_i\hat{n}_j}^{=:\,H_U}
    -\overbrace{\sum_{i\in\V}\Delta_i\,\hat{n}_i}^{=:\,H_\Delta}\,,
    \label{eq:H0}
\end{align}
where the interaction strength $U\in\mathbb{R}^+$ is a fixed parameter chosen
such that $U\gg\max_{i \in \V}\Delta_i$. Physically, the edges $\{i,j\} \in \E$
are identified with pairwise ``blockade interactions'' between the qubits, and
the vertex weights $\Delta_i$ correspond to ``detunings'' of the qubits. 
In this context, $\GG$ is referred to as a \emph{blockade graph}. The
\emph{(quantum) blockade Hamiltonian} is then defined as the sum of the
classical blockade Hamiltonian together with a uniform \emph{transverse field}
$\Omega\ge 0$ which induces quantum fluctuations:
\begin{align}
    H(\Omega)
    :=
    H_0+\overbrace{\Omega\sum_{i\in\V}\sigma_i^x}^{=:\,H_\Omega}\,.
    \label{eq:HOmega}
\end{align}
In the following, we often omit the dependence on $\Omega$ and $\GG$, simply
write $H$, and refer to this Hamiltonian as a \emph{blockade Hamiltonian}.
Hamiltonians of the form \eqref{eq:HOmega} can be realized by artificial
structures of neutral atoms that are trapped by optical tweezers and excited
into Rydberg states by lasers~\cite{browaeys2020manybody, kaufman2021tweezer};
in this context, the blockade interaction is called \emph{Rydberg blockade},
which is caused by strong van der Waals interactions between the highly
polarizable Rydberg states~\cite{jaksch2000fast, saffman2010rydberg}.

Because $\Delta_i>0$ and $U\gg\max_{i \in \V}\Delta_i$, it is easy to see
that the ground state manifold of the classical blockade Hamiltonian $H_0$ is
spanned by states $\ket{\vec{n}}$ with patterns $\vec{n}$ of excited qubits
that are in one-to-one correspondence with the \emph{maximum-weight independent
sets} (MWISs) of $\GG$, i.e., the subsets of pairwise disconnected vertices
that maximize the total weight. Note that the ground state energy of $H_0$ is
always negative.
The complexity and versatility of blockade Hamiltonians~\eqref{eq:HOmega}
derives from the freedom to design the blockade graph $\GG$. In this paper, we
are interested in a particular family of translationally invariant blockade
graphs that was constructed in Ref.~\cite{maier2025} to realize the toric code
topological order as the ground state of $H(\Omega)$ for finite but small
$\Omega$.  We turn to the description of these blockade graphs next.

\subsection{Blockade graph and honeycomb grid}
\label{subsec:SpatialStructure}

\begin{figure}[tbp]
    \includegraphics[width=0.8\linewidth]{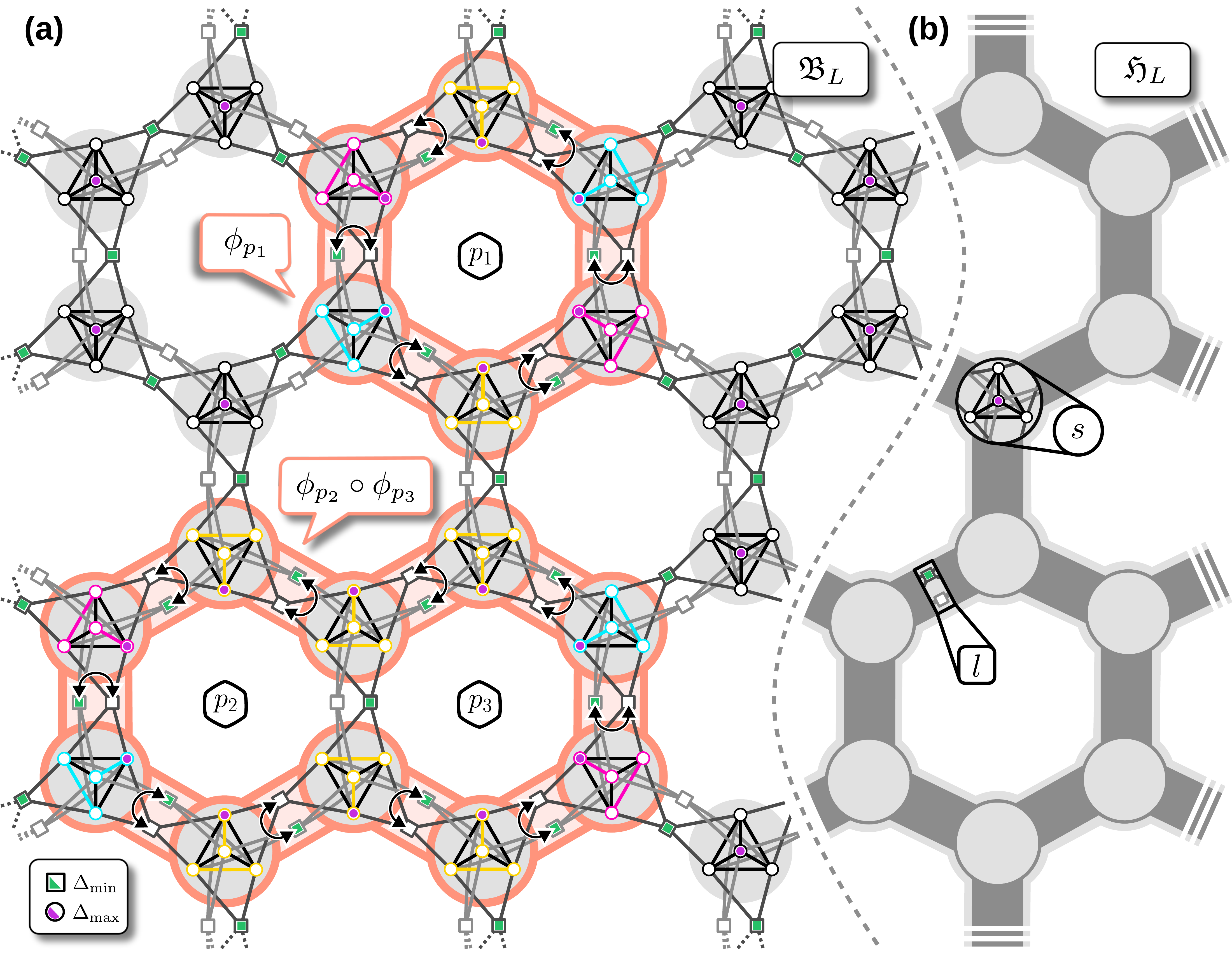}
    \caption{%
        \textbf{Blockade graphs and honeycomb grid.}
        \textbf{(a)}~We are interested in the vertex-weighted blockade graphs
        $\GB_L = (\V, \E, W)$ where we identify the vertices $i \in \V$ with
        qubits (squares and circles), the edges $\{i, j\} \in \E$ with the
        blockades (gray and black lines), and the vertex weights
        $\Delta_i=W(i)$ with the detunings (squares $\Deltamin$ and circles
        $\Deltamax \equiv 2\Deltamin$). The green/purple coloring of the
        vertices illustrates an exemplary MWIS, which can be identified with a
        ground state of $H_0$. 
        Along the boundary (orange) of the plaquettes $p_1$, $p_2$ and $p_3$,
        the light-gray qubits on the bounding links are excited; this
        corresponds to a ``loop excitation.'' It can be obtained from the
        loop-free MWIS by applying the graph automorphisms $\phi_{p_1}$ and
        $\phi_{p_2} \circ \phi_{p_3}$.
        \textbf{(b)}~After coarse-graining, we identify the sites $s \in \S$
        with four qubits from $\V$ each (gray circles), and the links $l \in
        \L$ with two qubits (thick gray lines). The constructed graph $\GH_L =
        (\S, \L)$ is a \emph{honeycomb tiling graph}.
    }
    \label{fig:CoarseGraining}
\end{figure}

In this paper, we focus on the family of translationally invariant blockade
graphs $\GB_L \equiv (\V, \E, W)$ proposed in Ref.~\cite{maier2025} and
depicted in \cref{fig:CoarseGraining}~(a) for periodic boundary conditions
(PBCs). Here, the parameter $\mathbb{N}\ni L \sim \sqrt{|\V|}$ denotes the
linear size of the system (see below) and we suppress the $L$-dependence of
$\V$, $\E$, and $W$. The associated Hilbert space is denoted by $\H_L$ and
the derived Hamiltonian \eqref{eq:HOmega} is denoted by $H_L$ (we often omit
the subscript $L$ to streamline the notation). A quantity $X$ that is
independent of $L$ is denoted by $X=\cL$. 

We also define a coarse-graining of $\GB_L$ by identifying subsets of
$\V$ with the vertices and edges of the \emph{honeycomb tiling graph}
$\GH_L \equiv (\S, \L)$, see \cref{fig:CoarseGraining}~(b). The
vertices $s \in \S$ are called \emph{sites} and the edges $l \in \L$ are called
\emph{links} (of the honeycomb grid). Then each site $s =
\{i_1,i_2,i_3,i_4\}\subset\V$ corresponds to four qubits from $\V$ and each
link $l=\{i_1,i_2\}\subset\V$ to two qubits. The \emph{faces} (or
\emph{plaquettes}) of the honeycomb graph are denoted by $p \in \P$ and
correspond to sets of $36=6\times(2+4)$ qubits from $\V$ that bound $p$. With
this interpretation, we can write $s\subset p$ ($l\subset p$) if site $s$ (link
$l$) bounds plaquette $p$. 
Note that the vertices (= qubits) in $\GB_L$ have only two distinct weights
$\Delta_i$ (= detunings) which we denote by $\Deltamin$ on the links and
$\Deltamax \equiv 2\Deltamin$ on the sites [squares and circles in
\cref{fig:CoarseGraining}~(a), respectively].
Finally, since a honeycomb grid contains two sites for every unit cell, we can
define the linear system size as $L := \sqrt{|\S|/2}$, i.e., $\S \simeq
\mathbb{Z}_L^2 \times \{A, B\}$ with two sublattices $A$ and $B$.

\subsection{Local symmetries and symmetry sectors}
\label{subsec:LocalSymmetries}

It was shown in Ref.~\cite{maier2025} that the automorphism group
$\group{Aut}(\GB_L)$ has an extensive (in $L$) number of local generators
$\phi_p$ associated to plaquettes $p\in\P$ of the underlying honeycomb grid.
Recall that the automorphisms $\phi$ of a vertex-weighted graph are vertex
permutations $\phi:\V\xrightarrow{\sim}\V$ that leave the edges and detunings
invariant: $\{i,j\}\in\E\,\Leftrightarrow\,\{\phi(i),\phi(j)\}\in\E$ and
$\Delta_i=\Delta_{\phi(i)}$. One can convince oneself [see
\cref{fig:CoarseGraining}~(a)] that for each plaquette $p\in\P$ there is a
permutation $\phi_p$ of vertices bounding the plaquette with this property.
The permutation $\phi_p$ is composed of vertex transpositions and acts locally in the sense
that $\phi_p(s) = s$ and $\phi_p(l) = l$ for sites $s \subset p$ and links $l
\subset p$ bounding the plaquette $p \in \P$; for vertices $i \in \V\backslash p$
outside the plaquette $\phi_p(i) = i$, i.e., $\phi_p$ acts as the
identity. These plaquette automorphisms can be concatenated to yield \emph{loop
automorphisms} (since adjacent plaquette automorphisms transpose vertices on
shared links twice); an example is shown in \cref{fig:CoarseGraining}~(a).

The graph automorphisms naturally define an action on the associated Hilbert
space via $U_p\ket{\vec
n}:=\ket{n_{\phi_p^{-1}(1)},n_{\phi_p^{-1}(2)},\ldots}\equiv\ket{\phi_p\cdot\vec
n}$; the inverse is needed for $\phi\mapsto U_\phi$ to be a left action
(it is immaterial here since $\phi_p^{-1}=\phi_p$).
This representation is unitary $U_pU_p^\dag=\Id$ and Hermitian $U_p^\dag=U_p$;
furthermore, it is easy to show that
\begin{align}
    \comm{H}{U_p}=0
    \qquad\text{and}\qquad
    \comm{U_p}{U_q}=0
    \qquad\text{and}\qquad
    U_p^2=\Id
    \label{eq:symmetries}
\end{align}
for all plaquettes $p,q\in\P$; it follows that all $U_p$ have eigenvalues $\pm
1$. Since all $U_p$ commute with the blockade Hamiltonian (even for
$\Omega>0$), they are referred to as \emph{plaquette symmetries}.
Note that for periodic boundaries one has the additional constraint
\begin{align}
    \prod_{p\in\P}U_p=\Id\,.
    \label{eq:prodUp}
\end{align}
Indeed, the two qubits on every link are exchanged twice, thus leaving them
invariant, and the permutations on the sites are fully determined by the
permutations on the adjacent links [see \cref{fig:CoarseGraining}~(a)]. Because
of \cref{eq:symmetries}, it makes sense to define the \emph{symmetry sectors}
\begin{align}
    \H^\vsigma
    := \spn{%
        \ket{\psi} \in \H_L
        \; \vert \; \forall_{p \in \P}: \; U_p\ket{\psi} = \sigma_p\ket{\psi} 
    }
    \quad\text{with}\quad
    \H_L=\bigoplus_\vsigma\H^\vsigma
    \label{eq:SymmetrySectors}
\end{align}
and corresponding projectors $P^\vsigma$. 
Here, $\vsigma \in \{\pm 1\}^{|\P|}$ fixes an eigenvalue for each plaquette
operator $U_p$ with $p \in \P$. Due to \cref{eq:prodUp}, we can restrict
ourselves in the following to choices $\vsigma$ that fulfill $\prod_{p \in \P}
\sigma_p=1$ (otherwise $\H^\vsigma = \{0\}$ is the trivial Hilbert space). The
\emph{symmetric sector} $\H^\s$ with $\sigma_p = +1$ for all $p \in \P$ is
labeled by $\vsigma = \s$ with projector $P^\s$. Symmetry sectors $\H^\vsigma$
that are \emph{not} the symmetric sector are called \emph{asymmetric sectors}.

\section{Goals and preliminaries}
\label{sec:structure}

The goal of this paper is to prove that the spectrum of the sequence of
Hamiltonians $H_L$, defined via \cref{eq:H0,eq:HOmega} and the family of
blockade graphs $\GB_L$ in \cref{fig:CoarseGraining}~(a), has a finite gap
above the (potentially degenerate) ground state space in the thermodynamic
limit $L\to\infty$. To define these concepts rigorously, we introduce the
notation $E_L = \inf\,\mathrm{spec}(H_L)$ to denote the ground state energy
(GSE) of $H_L$; instead of $L\to\infty$ we simply write $\lim_L$. We can then
define what it means for a sequence of Hamiltonians to be gapped:

\begin{definition}\label{def:ExcitationGap}
    Let $\gamma > 0$ and $L^* \in \mathbb{N}$. A sequence of Hamiltonians
    $(H_L)_{L \ge L^*}$ is called \emph{gapped} with \emph{asymptotic gap}
    $\gamma$, if there exists a nonnegative sequence $(\delta_L)_{L \ge L^*}$
    with $\lim_L\delta_L = 0$, and there exists a strictly larger sequence
    $(\gamma_L)_{L \ge L^*}$ with $\gamma_L > \delta_L$ and $\lim_L\gamma_L =
    \gamma$, such that
    \begin{align}
        \forall_{L \ge L^*}:\; \mathrm{spec}(H_L) \subset [E_L, E_L + \delta_L] \cup [E_L + \gamma_L, \infty)\,.
    \end{align}
    If the sequence $(H_L)_{L \ge L^*}$ is gapped, for $L \ge L^*$ we say $H_L$
    is gapped by at least $\gamma_L-\delta_L$ with ground state splitting at
    most $\delta_L$.
\end{definition}

This definition decomposes the spectrum of $H_L$ for $L \ge L^*$ into two
nonempty disjoint sets. The eigenstates whose eigenenergies lie within $[E_L,
E_L + \delta_L]$ span the \emph{ground state space} (GSS) $\G_L$. Its dimension
is called the \emph{ground state degeneracy} (GSD). The eigenstates with
eigenenergy exactly $E_L$ are called the \emph{ground states} (GSs)
$\ket{\psi_L}$ (this is often a single state in $\G_L$). The remaining spectrum
within $[E_L + \gamma_L, \infty)$ is called the \emph{high-energy spectrum}.
With this nomenclature, the goal of the paper is to prove that, for the family
of Hamiltonians $H_L$ defined by $\GB_L$, there exists a parameter $L^* \in
\mathbb{N}$ such that $(H_L)_{L \ge L^*}$ is gapped. In the following, we
suppress the explicit parametric dependence on $L$ for these quantities.

\subsection{Proof strategy}
\label{subsec:SymmetrySectors}

The full proof of this statement is quite technical, see \cref{fig:structure}.
The strategy of the proof is to split the Hilbert space $\H_L$ into the
symmetry sectors that were introduced in \cref{eq:SymmetrySectors}.
The full Hamiltonian $H$ commutes with $U_p$ [\cref{eq:symmetries}] and
therefore with the projectors $P^\vsigma$ (\cref{lem:CommutationsRelations}).
Hence the Hamiltonian is block-diagonal in the symmetry sectors. It is
therefore useful to define the projection of the Hamiltonian onto a symmetry
sector by $H^\vsigma := P^\vsigma H P^\vsigma = HP^\vsigma$, with GSs
$\ket{\psi^{\vsigma}}$ of GSE $E^{\vsigma}$ and GSS $\mathcal{G}^{\vsigma}$ (if
$H^\vsigma$ is gapped). Note that $E^\vsigma$ is always negative
(\cref{lem:Stabilizer}). In particular for the symmetric sector, we denote the
Hamiltonian by $H^\s$, with GSs $\ket{\psi^\s}$ of GSE $E^\s$ and GSS
$\mathcal{G}^\s$ (if $H^\s$ is gapped). 

\noindent The proof of the main result (\cref{the:ExcitationGap}) in
\cref{sec:MainResult} is then split into three steps:
\begin{itemize}

\item The first step is to show that, for sufficiently large linear system size
    $L$, $H^\s$ is gapped within the symmetric sector
    (\cref{pro:SpectrumWithinSS}).

\item Secondly, we show that the ground state of $H_L$ is unique and in the
    symmetric sector $\H^\s$ (\cref{pro:GlobalGS}).

\item Finally, we show that the lowest eigenenergies in the \emph{asymmetric}
    sectors $\H^{\vsigma\neq\s}$ are separated from $E^\s$ by a constant
    independent of the system size, i.e., that $E^\vsigma - E^\s \ge \cL > 0$
    for all $\vsigma \neq \s$ (\cref{pro:GapBetweenSS}). 

\end{itemize}
\noindent \cref{pro:SpectrumWithinSS,pro:GlobalGS,pro:GapBetweenSS} are
explained and proven in
\cref{subsec:SpectrumWithinSS,subsec:GlobalGS,subsec:GapBetweenSS},
respectively. These proofs are based in turn on
\cref{lem:GapStability_cpy,lem:GroundState_cpy,lem:AsymmetricSectors_cpy,lem:LowerBound_cpy,lem:Induction_cpy}
which are proven in the technical
\cref{subsec:GapStability,subsec:GroundState,subsec:AsymmetricSectors,subsec:LowerBound,subsec:Induction},
respectively. These dependencies are illustrated in \cref{fig:structure}.

\subsection{Constrained Hilbert space and Hamiltonian}
\label{subsec:RestrictedHilbertSpace}

The proof of the gap \emph{between} symmetry sectors (\cref{pro:GapBetweenSS})
is the most complicated part of the paper and based on the insight that every
state with a completely de-excited (``empty'') plaquette $\hat p$ -- i.e., on
which all qubits are in state $\ket{0}_i$ -- is necessarily symmetric under
$U_\p$ so that $\sigma_\p = +1$.  We then perform an induction in the number of
excitations on plaquette $\p$ to lower-bound the eigenenergies for partially
occupied $\p$. This motivates the definition of a family of auxiliary
Hamiltonians $H_k^q$ for which the off-diagonal elements of $\Omega \sum_{i \in
\p} \sigma_i^x$ on a fixed plaquette $\p \in \P$ are constrained to a specific
number of excitations:

\begin{definition}[Constrained Hilbert space and Hamiltonian]  
    \label{def:restH}

    Consider a fixed plaquette $\p\in\P$ of the honeycomb graph. Let
    $\mathbb{Z}_2 = \{0,1\}$, let $k \le q$ be two nonnegative integers, and
    let $[\mathbb{Z}_2^{|\V|}]_k^q := \{\vec{n} \in \mathbb{Z}_2^{|\V|}\,\vert\,k
    \le \textstyle\sum_{i \in \p} n_i \le q\}$ be the set of all excitation
    patterns $\vec n$ with at least $k$ and at most $q$ excitations on $\hat
    p$. We can then define the corresponding constrained Hilbert space 
    \begin{align}
        \H_k^q 
        := \spn{\ket{\vec{n}} \in \H\,\Big|\,\vec{n} \in [\mathbb{Z}_2^{|\V|}]_k^q}
    \end{align}
    with associated projector $P^q_k$. The constrained Hamiltonian is defined as
    \begin{align}
        H_k^q 
        :=
        \overbrace{%
        H_0 
        + \Omega 
        \sum_{i \notin \p} \sigma_i^x 
        }^{=:\,H_\p}
        + \Omega\,P_k^q \overbrace{\sum_{i \in \p} \sigma_i^x}^{=: \sigma_\p^x} P_k^q 
        \,,
        \label{eq:GeneralizedHamiltonian}
    \end{align}
    where the transitions on the fixed plaquette $\p$ are constrained to $\H_k^q$.

\end{definition}

For the constrained Hamiltonian $H_k^q$, we denote its GSs by $\ket{\psi_k^q}$
with GSE $E_k^q$, and the GSS by $\mathcal{G}_k^q$ (if $H_k^q$ is gapped).
Crucially, $H_k^q$ and $P^\vsigma$ commute (\cref{lem:CommutationsRelations}).
This makes the constrained Hamiltonian \eqref{eq:GeneralizedHamiltonian}
compatible with the symmetry sectors \eqref{eq:SymmetrySectors}. Hence it will
be useful to introduce the Hamiltonian $[H_k^q]^\vsigma := P^\vsigma H_k^q
P^\vsigma = H_k^q P^\vsigma$; we denote its GSs by $\ket{[\psi_k^q]^\vsigma}$
with GSE $[E_k^q]^\vsigma$, and the GSS by $[\G_k^q]^\vsigma$ (if
$[H_k^q]^\vsigma$ is gapped).

The Hamiltonian $H_\p$ is the part of $H^q_k$ that is diagonal on $\p$; note
that $H_k^k = H_\p$ (independent of $k$). For $q = |\p|$ we can write
symbolically $q=\infty$ (indicating the maximum number of excitations per
plaquette). For $(k, q) = (0, \infty)$, we have $H = H_0^\infty$ and we recover
the full Hamiltonian~\eqref{eq:HOmega}. 
For $\Omega = 0$ we recover the classical Hamiltonian $H_0 = H_k^q(0)$
(independent of $k$ and $q$). We write $\ket{\psi_0}$ for its GSs and $E_0$ for
the GSE. Since this classical Hamiltonian only enforces local loop constraints~\cite{maier2025}, 
it features an extensively degenerate GSS $\G_0$ with
vanishing state splitting (namely the space of all closed-loop patterns).

For the plaquette $\p$, we define $\hat{n}_\p := \sum_{i \in \p} \hat{n}_i$ and
label its eigenvalue $n_\p := \sum_{i \in \p} n_i$, which counts the number of
excitations on the plaquette. This number is constant across all
\emph{classical} ground states by construction; we denote this number by $n_0
:= \langle \hat n_\p \rangle_{\G_0}$.  For the particular family of blockade graphs
$\GB_L$ in \cref{fig:CoarseGraining}~(a) then $n_0 = 12$ -- a single
excitation per site and link bounding $\p$.

\section{Main result}
\label{sec:MainResult}

\noindent The main result of this paper is the following theorem:
\begin{theorem}[Excitation Gap] 
    \label{the:ExcitationGap}
    There exist $\OSTA > 0$ and $\LOTA \in \mathbb{N}$ such that for $\Omega
    \in (0, \OSTA)$ and $L \ge \LOTA$ the Hamiltonian $H$ is gapped.  Its
    ground state space $\G$ lies within the symmetric sector $\H^\s$ and has a
    $4$-fold (topological) degeneracy.
\end{theorem}
\noindent This statement can be conceptually split into three propositions that are proven
separately in the remainder of the paper.
The first proposition, which is proven in \cref{subsec:SpectrumWithinSS},
concerns the spectrum \emph{within} each symmetry sector $\H^\vsigma$ and
asserts that the symmetry-projected Hamiltonian $H^\vsigma$ is gapped:
\begin{proposition}[Spectrum within symmetry sectors]
    \label{pro:SpectrumWithinSS}
    There exist $\OOPA > 0$ and $\LOPA \in \mathbb{N}$ such that for $\Omega
    \in [0, \OOPA]$ and $L \ge \LOPA$ the following holds: Within each symmetry
    sector $\H^\vsigma$, the Hamiltonian $H^\vsigma$ is gapped by at least
    $\Deltamin/2$ and the ground state space $\G^\vsigma$ is $4$-fold
    (topologically) degenerate. 
\end{proposition}

\noindent The second proposition, which is proven in \cref{subsec:GlobalGS},
asserts the uniqueness and symmetry of the global ground state of $H$:
\begin{proposition}[Global ground state]
    \label{pro:GlobalGS}
    There exist $\OGGS > 0$ and $\LGGS \in \mathbb{N}$ such that for $\Omega
    \in (0, \OGGS]$ and $L \ge \LGGS$ the following holds: The ground state
    $\ket{\psi}$ of the full Hamiltonian $H$ is non-degenerate and lies within
    the symmetric sector $\H^\s$, i.e., $\ket{\psi} = \ket{\psi^\s}$ is also
    the unique ground state of $H^\s$.
\end{proposition}

\noindent The third proposition, which is proven in \cref{subsec:GapBetweenSS}, concerns
the gap \emph{between} different symmetry sectors; it asserts that the ground
state energies $E^\vsigma$ of the \emph{asymmetric} sectors are separated by a
gap from the ground state energy $E^\s$ of the \emph{symmetric} sector --
independent of the system size:
\begin{proposition}[Gap between symmetry sectors]
    \label{pro:GapBetweenSS}
    There exist $\OSPB > 0$ and $\LOPB \in \mathbb{N}$ such that for $\Omega
    \in (0, \OSPB)$ the following holds: There exists a constant $\DElb > 0$
    that, for every asymmetric sector $\H^\vsigma$ and for all $L \ge \LOPB$,
    lower-bounds the gap $E^\vsigma - E^\s$ between the ground state energy in
    $\H^\vsigma$ and in the symmetric sector $\H^\s$, i.e., $E^\vsigma - E^\s \ge
    \DElb$.
\end{proposition}

\noindent\cref{the:ExcitationGap} then follows by combining the three propositions:
\begin{proof}[Proof of \cref{the:ExcitationGap} (Excitation Gap, see \cref{fig:Spectrum})]
    To combine \cref{pro:SpectrumWithinSS,pro:GlobalGS,pro:GapBetweenSS}, we choose
    $\OSTA := \min\{\OOPA,\OGGS,\OSPB\}$ and require $\LOTA\geq\max\{\LOPA,\LGGS,
    \LOPB\}$ (the exact choice of $\LOTA$ follows below).
    From \cref{pro:GlobalGS} we then know that for $\Omega\in(0,\Omega_0)$
    and $L\geq L_0$ the ground state $\ket{\psi}$ of $H$ is unique and lies in
    the symmetric sector $\H^\s$.  From \cref{pro:SpectrumWithinSS} we know
    that $H^\s$ is gapped by at least $\Deltamin/2$ in the symmetric sector with a ground state space
    $\G^\s$.
    From \cref{pro:GapBetweenSS} we know that, for every asymmetric sector
    $\H^\vsigma$, we have $E^\vsigma - E^\s \ge \DElb$ where $\DElb > 0$ is a
    constant independent of $\vsigma$ and $L$ (it can depend on
    $\Omega$); hence there exists a gap between the ground states
    $\ket{\psi^\vsigma}$ and the ground state $\ket{\psi^\s}$ in the symmetric
    sector.
    Additionally, for $L \ge \LOPA$, the ground state splitting of $\G^\s$ is
    bounded by some $\delta_L^\s$ with $\lim_L \delta_L^\s = 0$. In particular,
    there exists a minimal size $L^* \in \mathbb{N}$ such that for all $L \ge
    L^*$ one has $\delta_L^\s < \min\{\DElb, \Deltamin/2\}$; so we choose $\LOTA :=
    \max\{\LOPA,\LGGS,\LOPB,L^*\}$.
    We can then set $\gamma_L \equiv \gamma := \min\{\DElb, \Deltamin/2\}$, so
    that for $L \ge \LOTA$ the full Hamiltonian $H$ is gapped by at least
    $\gamma - \delta_L^\s > 0$. The ground state space is $\G = \G^\s$ in the
    symmetric sector $\H^\s$ with $4$-fold (topological) degeneracy. The ground
    state splitting is bounded by $\delta_L := \delta_L^\s$.
\end{proof}

\begin{figure*}[tb]
    \begin{minipage}[c]{0.33\textwidth}
        \includegraphics[width=\linewidth]{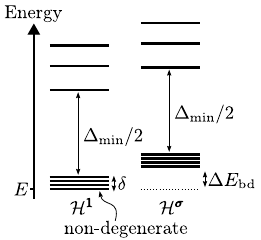}
    \end{minipage}\hfill
    \begin{minipage}[c]{0.63\textwidth}
    \caption{%
    \textbf{Spectrum in different symmetry sectors.} 
    As explained in \cref{subsec:LocalSymmetries}, the spectrum of the
    Hamiltonian $H$ can be split into different symmetry sectors. Here,
    $\H^\s$ denotes the symmetric sector and $\H^\vsigma$ some asymmetric
    sector. From \cref{pro:SpectrumWithinSS} we know that
    every symmetry sector has a $4$-fold degenerate ground state space
    with an excitation gap of at least $\Deltamin/2$; the ground state
    splitting $\delta$ is bounded by a function decaying with the
    system size $L$.  According to \cref{pro:GlobalGS}, the ground state of
    the Hamiltonian is non-degenerate and lives in the symmetric sector.
    \cref{pro:GapBetweenSS} asserts that the energy gap \emph{between} the
    symmetry sectors is lower-bounded by a constant $\DElb > 0$ that is
    independent of the system size $L$ and the symmetry sector $\H^\vsigma$.
    }
    \label{fig:Spectrum}
    \end{minipage}
\end{figure*}

\section{Proofs of the propositions}
\label{sec:Propositions}

In this section we discuss and derive the three propositions using several lemmas:
In \cref{subsec:SpectrumWithinSS} we establish \cref{pro:SpectrumWithinSS} as
a corollary of the more general \cref{lem:GapStability_cpy}, the proof of
which follows in \cref{subsec:GapStability}.
In \cref{subsec:GlobalGS} we establish \cref{pro:GlobalGS} as a
corollary of the more general \cref{lem:GroundState_cpy}, the proof of which is
given in \cref{subsec:GroundState}.
Finally, in \cref{subsec:GapBetweenSS} we prove \cref{pro:GapBetweenSS} by
combining
\cref{lem:AsymmetricSectors_cpy,lem:LowerBound_cpy,lem:Induction_cpy}, the
proofs of which are given in
\cref{subsec:AsymmetricSectors,subsec:LowerBound,subsec:Induction},
respectively.

\subsection{Spectrum within symmetry sectors (\texorpdfstring{\cref{pro:SpectrumWithinSS}}{Proposition 1})} 
\label{subsec:SpectrumWithinSS}

\noindent In \cref{subsec:GapStability} we prove the following lemma for the
constrained Hamiltonian from \cref{def:restH}:
\begin{lemma}[Gap stability]
    \label{lem:GapStability_cpy}
    There exist $\OOLA > 0$, $\LOLA \in \mathbb{N}$, such that for all $\Omega
    \in [0,\OOLA]$, $L \ge \LOLA$, and for all nonnegative integers $k$, $q$
    with $k \le q$, the following holds: 
    Within each symmetry sector $\H^\vsigma$, the constrained Hamiltonian
    $[H_k^q]^\vsigma$ is gapped by at least $\Deltamin/2$ and the ground state
    space $[\mathcal{G}_k^q]^{\vsigma}$ is $4$-fold (topologically) degenerate.
\end{lemma}
The proof is quite technical; the crucial insight is that the \emph{classical}
Hamiltonian $H^\vsigma(\Omega=0)$ within each symmetry sector $\H^\vsigma$ (for
an arbitrary $\vsigma \in \{\pm 1\}^{|\P|}$) is gapped, frustration-free, and
has a ground state manifold with local topological quantum order (LTQO). This allows
us to apply a result by Michalakis and Zwolak~\cite{michalakis2013}
which asserts stability of the gap under arbitrary local perturbations -- in
particular under $H_\Omega$ for weak-enough transverse fields $\Omega>0$.
The general form of \cref{lem:GapStability_cpy} is required for the proof of
\cref{lem:LowerBound_cpy,lem:Induction_cpy};
\cref{pro:SpectrumWithinSS} follows as a corollary for the special case $(k, q)
= (0, \infty)$:
\begin{proof}[Proof of \cref{pro:SpectrumWithinSS}]
    Consider the unconstrained special case $(k,q) = (0, \infty)$. Recall that in this case
    $[H_0^\infty]^\vsigma=H^\vsigma$, $[\G_0^\infty]^\vsigma= \G^\vsigma$
    and $\ket{[\psi_0^\infty]^\vsigma}=\ket{\psi^\vsigma}$. Hence
    \cref{lem:GapStability_cpy} reduces to \cref{pro:SpectrumWithinSS}.
\end{proof}

\subsection{Global ground state (\texorpdfstring{\cref{pro:GlobalGS}}{Proposition 2})} 
\label{subsec:GlobalGS}

\noindent In \cref{subsec:GroundState} we prove the following lemma, again for
the constrained Hamiltonian from \cref{def:restH}:

\begin{lemma}[Ground state]
    \label{lem:GroundState_cpy}
    There exist $\OOLB > 0$, $\LOLB \in \mathbb{N}$, such that for all $\Omega
    \in (0, \OOLB]$, $L \ge \LOLB$, and for all nonnegative integers $k$, $q$
    with $k \le \nO \le q$, the following holds:
    \begin{itemize}
    
    \item For every symmetry sector $\H^\vsigma$, the ground state space
    $[\G_k^q]^\vsigma$ of the constrained Hamiltonian $[H_k^q]^\vsigma$ is
    a subspace of $\H_k^q$.
    
    \item For $k < q$, the ground state $\ket{\psi_k^q}$ of the Hamiltonian $H_k^q$
    is non-degenerate and lies within the symmetric sector $\H^\s$, i.e.,
    $\ket{\psi_k^q} = \ket{[\psi_k^q]^\s}$ is also the unique ground state of
    $[H_k^q]^\s$. 
    
    In the natural basis, it takes the form $\ket{\psi_k^q} = \sum_{\vec{n} \in
    [\mathbb{Z}_2^{|\V|}]_k^q}(-1)^{|\vec{n}|}c_{\vec{n}}\ket{\vec{n}}$ with
    $c_{\vec{n}} > 0$, where $|\vec{n}| = \sum_{i \in \V} n_i$ denotes
    the total number of excitations.

    \end{itemize}

\end{lemma}
Remember that $\nO = 12$ for the particular family of blockade graphs $\GB_L$
we consider in this paper.
The proof of \cref{lem:GroundState_cpy} is based on the well-known theorem by
Perron and Frobenius~\cite{perronfrobenius1912} and the special structure of
blockade Hamiltonians \eqref{eq:HOmega}. That the gapped ground state space
$[\G_k^q]^\vsigma$ is a subspace of $\H_k^q$ is also a consequence of this
structure and follows via a theorem on spectral flow by
Nachtergaele~\etal~\cite{nachtergaele2019}.
The general form of \cref{lem:GroundState_cpy} is also required for the proof
of \cref{lem:LowerBound_cpy,lem:Induction_cpy}; \cref{pro:GlobalGS}
follows as a corollary for the special case $(k, q) = (0, \infty)$:
\begin{proof}[Proof of \cref{pro:GlobalGS}]
    Consider again the unconstrained special case $(k,q) = (0, \infty)$. With
    $[H_0^\infty]^\vsigma=H^\vsigma$, $[\G_0^\infty]^\vsigma= \G^\vsigma$ and
    $\ket{[\psi_0^\infty]^\vsigma}=\ket{\psi^\vsigma}$, the statement of
    \cref{lem:GroundState_cpy} implies \cref{pro:GlobalGS}.
\end{proof}

\subsection{Gap between symmetry sectors (\texorpdfstring{\cref{pro:GapBetweenSS}}{Proposition 3})}%
\label{subsec:GapBetweenSS}

The proof of \cref{pro:GapBetweenSS}, which asserts a finite gap between the
lowest-energy manifolds of the symmetric and the asymmetric sectors, follows
from a sequence of lemmas which we state here and prove later in
\cref{sec:TechnicalPart}.

The main insight is that every state in $\H_0^0$ (i.e., states without
excitations on a fixed plaquette $\p$) is necessarily symmetric under
$U_\p$, i.e., $\sigma_\p = +1$. This can be used to prove the following lemma:
\begin{lemma}[Asymmetric sectors]
    \label{lem:AsymmetricSectors_cpy}

    Fix an arbitrary asymmetric sector $\H^\vsigma$ to define $H^\vsigma$, and
    fix a plaquette $\p \in \P$ to define $H_1^\infty$. Then the ground state
    energy of $H^\vsigma$ is not smaller than the ground state energy of
    $H_1^\infty$: $E^\vsigma \ge E_1^\infty$.

\end{lemma}
This lemma is proven in \cref{subsec:AsymmetricSectors}. It allows us to ignore
the (possibly complicated) spectrum in the asymmetric sectors and focus on
deriving a lower bound on $E_1^\infty$ instead. (Note that we expect this bound
not to be sharp.) \cref{lem:AsymmetricSectors_cpy} motivates a posteriori the
definition of the constrained Hamiltonian in
\cref{subsec:RestrictedHilbertSpace}.
The remaining task is then to lower-bound the ground state energy $E_1^\infty$.
This is achieved by two additional lemmas. The first focuses on deriving a
suitable lower bound, but does not yet guarantee that this bound is positive:
\begin{lemma}[Lower bound]
    \label{lem:LowerBound_cpy}
    There exist $\OOLD > 0$, $\LOLD \in \mathbb{N}$, such that for all $\Omega
    \in (0, \OOLD]$, and for a positive integer $k$ with $k \le \nO$, the
    following holds:
    There exists a strictly increasing function $\DE_k:\;\mathbb{R}_{\ge 0}
    \rightarrow \mathbb{R}_{\ge 0}$ such that, for all $L \ge \LOLD$, one has
    $E_k^\infty - E_{k-1}^\infty \ge \DE_k(\varepsilon_k)$ with $\varepsilon_k
    := \lVert P_k^k\ket{\psi_k^\infty}\rVert$.
\end{lemma}
This lemma is proven in \cref{subsec:LowerBound}. Note that the lower bound
$\DE_1(\varepsilon_1)$ (the case for $k = 1$) \emph{would} be positive, if we
could show that $\varepsilon_1 = \lVert P_1^1\ket{\psi_1^\infty}\rVert$ is not
smaller than a positive constant independent of the system size.
This fact is asserted by the final \cref{lem:Induction_cpy}, which we prove in
\cref{subsec:Induction}. The idea is to use an induction in $k$ for values $k
\le \nO$ starting at $k = \nO$.  This requires again -- directly and
indirectly -- much of the toolbox developed so far, and this is the reason why
the previous \cref{lem:GapStability_cpy,lem:GroundState_cpy,lem:LowerBound_cpy} were
formulated for the general case of nonnegative, integer parameters $k$ and $q$.
The lemma asserts the following:
\begin{lemma}[Single-excitation weight]
    \label{lem:Induction_cpy}
    There exist $\OSLE > 0$, $\LOLE \in \mathbb{N}$, such that for all $\Omega
    \in (0, \OSLE)$, there exists a constant $\varepsilon_\lb > 0$ such that,
    for all $L \ge \LOLE$, one has $\lVert P_1^1\ket{\psi_1^\infty}\rVert =
    \varepsilon_1 \ge \varepsilon_\lb$.
\end{lemma}

\noindent We now have the tools to prove \cref{pro:GapBetweenSS}:
\begin{proof}[Proof of \cref{pro:GapBetweenSS}.]

    Let us fix an arbitrary asymmetric sector $\H^\vsigma$.  We want to use
    \cref{lem:LowerBound_cpy,lem:Induction_cpy}. Since both introduce their own
    upper and lower bounds on $\Omega$ and $L$, respectively, we set $\OSPB :=
    \min\{\OOLD, \OSLE\}$ and $\LOPB := \max\{\LOLD, \LOLE\}$.  We then have
    the following chain of lower bounds:
    \begin{align}
        E^\vsigma - E^\s 
        \stackrel{(1)}{\ge} E_1^\infty - E^\s 
        \stackrel{(2)}{=} E_1^\infty - E_0^\infty 
        \stackrel{(3)}{\ge} \DE_1(\varepsilon_1) 
        \stackrel{(4)}{\ge} \DE_1(\varepsilon_\lb) 
        =: \DElb
        \stackrel{(5)}{=}\cL
        \stackrel{(6)}{>} 0
        \,.
    \end{align}
    Let us go through each (in)equality step by step:
    \begin{enumerate}[label=(\arabic*)]

        \item Since $\H^\vsigma$ is asymmetric, there is at least one plaquette
            $\p \in \P$ with $\sigma_\p = -1$; fix this plaquette to define
            $E_1^\infty$. \cref{lem:AsymmetricSectors_cpy} then asserts that
            $E^\vsigma \ge E_1^\infty$.

        \item Remember that $E_0^\infty \equiv E$. Using the theorem of
            Perron--Frobenius, it is easy to show that $E^\s = E$ (this is a
            general property of the full Hamiltonian $H_0^\infty \equiv H$ and
            valid for all $\Omega > 0$ and $L \in \mathbb{N}$, see the proof of
            \cref{lem:GroundState_cpy} in \cref{subsec:GroundState}). Hence we
            obtain $E^\s = E_0^\infty$.

        \item This follows directly from \cref{lem:LowerBound_cpy} for $k=1\leq
            n_0=12$ with $\varepsilon_1 = \lVert
            P_1^1\ket{\psi_1^\infty}\rVert$ and $\DE_1:\;\mathbb{R}_{\ge 0}
            \rightarrow \mathbb{R}_{\ge 0}$ a strictly increasing function.

        \item \cref{lem:Induction_cpy} asserts the existence of a positive
            constant $\varepsilon_\lb = \cL$ with $\varepsilon_1 \ge
            \varepsilon_\lb > 0$ for $L\geq L_3$. Since $\DE_1$ is increasing,
            it follows $\DE_1(\varepsilon_1) \ge \DE_1(\varepsilon_\lb)$.

        \item \cref{lem:LowerBound_cpy} asserts that $\DE_1$ is independent of $L\geq L_3$.
            \cref{lem:Induction_cpy} asserts that $\varepsilon_\lb$ is independent of $L\geq L_3$.

        \item Since $\DE_1$ is a strictly increasing function on
            $\mathbb{R}_{\geq 0}$ with values in $\mathbb{R}_{\geq 0}$
            (\cref{lem:LowerBound_cpy}), and $\varepsilon_\lb>0$
            (\cref{lem:Induction_cpy}), it follows that
            $\DElb=\DE_1(\varepsilon_\lb)>0$.
        
    \end{enumerate}
    This concludes the proof.
\end{proof}

\section{Technical part}
\label{sec:TechnicalPart}

At this point we have reduced \cref{the:ExcitationGap} to the five
\cref{lem:GapStability_cpy,lem:GroundState_cpy,lem:AsymmetricSectors_cpy,lem:LowerBound_cpy,lem:Induction_cpy},
which are finally proven in their respective
\cref{subsec:GapStability,subsec:GroundState,subsec:AsymmetricSectors,subsec:LowerBound,subsec:Induction}
below. Note that \cref{lem:GapStability_cpy,lem:GroundState_cpy,lem:LowerBound_cpy} are
-- directly or indirectly -- needed to prove other lemmas (see
\cref{fig:structure}).

\subsection{Gap stability (Michalakis--Zwolak)}%
\label{subsec:GapStability}

\noindent In this subsection, we prove \cref{lem:GapStability_cpy}, here
reprinted for the reader's convenience:
\setcounter{lemma}{0}
\renewcommand{\theHlemma}{technical.\arabic{lemma}}
\begin{lemma}[Gap stability]
    \label{lem:GapStability}

    There exist $\OOLA > 0$, $\LOLA \in \mathbb{N}$, such that for all
    $\Omega \in [0,\OOLA]$, $L \ge \LOLA$, and for all nonnegative integers $k$, $q$ 
    with $k \le q$, the following holds: 
    Within each symmetry sector $\H^\vsigma$, the
    constrained Hamiltonian $[H_k^q]^\vsigma$ is gapped by at least
    $\Deltamin/2$ and the ground state space $[\mathcal{G}_k^q]^{\vsigma}$ is
    $4$-fold (topologically) degenerate.

\end{lemma}
\noindent Note that the parameter $\Deltamin = \cL$ is independent of the
system size. The proof generalizes the special case for $H^\s$ studied by some
of us in Ref.~\cite{maier2025}:

\begin{proof}[Proof of \cref{lem:GapStability}.]

    The proof splits into four steps. In step (1) we add a gapped auxiliary
    term $H^{[\vsigma]}(\omega)$ to the classical Hamiltonian which reorders
    the symmetry sectors in the spectrum such that a prescribed sector
    $\H^\vsigma$ contains the global ground state. In step (2) we apply a
    theorem on gap stability by Michalakis and Zwolak~\cite{michalakis2013} to
    this auxiliary Hamiltonian to show that the gap remains open for small but
    finite transverse fields $\Omega>0$. In step (3) we show that for
    $\Omega>0$ the ground state manifold remains in the prescribed sector
    $\H^\vsigma$. Finally, in step (4) we use this to show that the original
    Hamiltonian $[H_k^q]^\vsigma$ (without the auxiliary term) inherits all
    these properties (i.e., is gapped with ground state manifold in
    $\H^\vsigma$).
    \begin{enumerate}[label=(\arabic*)]

    \item
    Consider an arbitrary symmetry sector $\H^\vsigma$ with eigenvalues
    $\sigma_p$ of $U_p$ for $p \in \P$, and fix a reference plaquette $\p \in
    \P$. We define the auxiliary Hamiltonian 
    \begin{align}
        \tilde H_k^q(\Omega, \omega; \vsigma) 
        := H_k^q(\Omega)\,+\,&H^{[\vsigma]}(\omega)
        \nonumber\\
        \text{with}\qquad
        &H^{[\vsigma]}(\omega)
        := \omega \sum_{p \in \P} (\mathds{1} - \hP_p^{[\sigma_p]}) 
        = \omega \sum_{p \in \P}\frac{\mathds{1} - \sigma_pU_p}{2}
        \,,
        \label{eq:AnciHam}
    \end{align}
    where $\hP_p^{[\sigma_p]} := (\mathds{1} + \sigma_pU_p)/2$ is a
    projector; we always consider $\omega \ge 0$.
    For $\Omega = 0$, we denote the unperturbed Hamiltonian by $\tilde
    H_0^{[\vsigma]}(\omega) := \tilde H_k^q(0, \omega; \vsigma)$, in particular
    $\tilde H_0^{[\s]}(\omega) = \tilde H_k^q(0, \omega; \s)$ for the symmetric
    sector. Correspondingly, we define $V_k^q(\Omega) := \Omega \left[\sum_{i
    \notin \p} \sigma_i^x + \sum_{i \in \p} P_k^q \sigma_i^x P_k^q \right]$ (we
    will treat this as a perturbation below). With these definitions, we can
    write $\tilde H_0^{[\vsigma]}(\omega) = H_0 + H^{[\vsigma]}(\omega)$ and
    $\tilde H_k^q(\Omega, \omega; \vsigma) = \tilde H_0^{[\vsigma]}(\omega) +
    V_k^q(\Omega)$. Since $H_k^q$ and $U_p$ commute
    (\cref{lem:CommutationsRelations}), $H_k^q$ and $\hP_p^{[\sigma_p]}$
    also commute. The symmetries on different plaquettes commute as well
    [\cref{eq:symmetries}], thus $\tilde H_k^q$ and $\hP_p^{[\sigma_p]}$
    commute. Hence $\tilde H_k^q$ is block-diagonal in the symmetry sectors.
    
    Note that $\tilde H_0^{[\vsigma]}(\omega)$ can be diagonalized exactly
    (since $\Omega=0$); it is then easy to see that the Hamiltonian is gapped
    and the gap is given by $\min\{\Deltamin, 2\omega\}$. Furthermore, the
    ground state space is given by $\tilde{\mathcal{G}}_0^{[\vsigma]} =
    \mathcal{G}_0 \cap \H^\vsigma$, where $\mathcal{G}_0$ is the ground state
    space of $H_0$ (the classical part), and the sector $\H^\vsigma$ is the
    ground state space of $H^{[\vsigma]}(\omega)$. Since the blockade terms in
    $H_0$ enforce a loop constraint, $\tilde H_0^{[\vsigma]}(\omega)$ has an
    (exact) $4$-fold topological ground state degeneracy for periodic boundary
    conditions: $\dim\tilde{\mathcal{G}}_0^{[\vsigma]} = 2^2 = 4$
    (\cref{lem:Stabilizer}).

    \item 
    We now want to transition from the fixpoint Hamiltonian
    $\tilde H_0^{[\vsigma]}(\omega)$ to the full Hamiltonian $\tilde H_k^q(\Omega,
    \omega; \vsigma)$ by switching on $\Omega$. Since the latter is no longer
    exactly diagonalizable, we need more sophisticated tools to keep the ground
    state space under control.
    The gap stability theorem by Michalakis and Zwolak~\cite[Theorem
    1]{michalakis2013} asserts the stability of a spectral gap in the
    thermodynamic limit ($L\to\infty$) under arbitrary, weak, local
    perturbations, provided that the unperturbed Hamiltonian satisfies six conditions:
    The Hamiltonian must be spatially local (A1), with periodic boundaries (A2),
    frustration-free (A3), and gapped (A4). Furthermore, the Hamiltonian must
    satisfy the so-called \emph{local gap condition} (A5) and exhibit
    \emph{local topological quantum order} (A6). In
    \cref{lem:ConditionsForGapStability} we show that $\tilde
    H_0^{[\vsigma]}(\omega)$ satisfies all six conditions.
    
    Let us henceforth set $\omega = \Deltamin/2$ and consider $V_k^q(\Omega)$
    as a local perturbation of $\tilde H_0^{[\vsigma]}(\omega)$.  The gap
    stability theorem then asserts that $\tilde H_k^q(\Omega,\omega;\vsigma)$
    is gapped, with a gap that is lower-bounded by $\Deltamin/2$ for
    sufficiently small but finite fields $\Omega \le \Omega^*$ and sufficiently
    large system size $L \ge L^*$. Here, the constants $\Omega^*>0$ and $L^*$
    are determined by the constants $J_0$ and $L_0$ of Ref.~\cite{michalakis2013}; their
    exact values are not important for us. In
    \cref{lem:ConditionsForGapStability} we show that $\Omega^*$ and $L^*$ can
    be chosen uniform in $k \le q$ and $\H^\vsigma$. We then set $\OOLA =
    \Omega^*$ and $\LOLA = L^*$.  
    We denote the ground state space of $\tilde H_k^q(\Omega,\omega;\vsigma)$
    by $\tilde \G_k^q(\Omega,\omega;\vsigma)$. Since $\tilde H_k^q$ remains
    gapped for $\Omega \le \OOLA$, the ground state degeneracy remains
    $\dim[\tilde
    \G_k^q(\Omega,\omega;\vsigma)]=\dim\tilde{\mathcal{G}}_0^{[\vsigma]}=4$ for
    $\Omega \le \OOLA$.  

    \item
    Our goal is to show that these statements are also valid for the
    constrained Hamiltonian $[H_k^q]^\vsigma(\Omega)$. To do so, we must show
    that the ground state space $\tilde\G_k^q(\Omega,\omega;\vsigma)$ is a
    subspace of $\H^\vsigma$ for $0\leq\Omega\leq\Omega_1'$.
    Since $\tilde H_0^{[\vsigma]}(\omega)$ is frustration-free, and only states
    in $\H^\vsigma$ make the positive-semi-definite auxiliary term
    $H^{[\vsigma]}(\omega)$ vanish, we know that for $\Omega=0$ the ground
    state space $\tilde{\mathcal{G}}_0^{[\vsigma]} \equiv
    \tilde\G_k^q(0,\omega;\vsigma)$ indeed lies in the symmetry sector
    $\H^\vsigma$. 
    Since $\tilde H_k^q(\Omega,\omega;\vsigma)$ is block-diagonal in the
    symmetry sectors, a continuous function of $\Omega$, and gapped for
    $0<\Omega \le \OOLA$ according to step (2), it follows that the ground
    state space $\tilde\G_k^q(\Omega,\omega;\vsigma)$ remains in the symmetry
    sector $\H^\vsigma$ for all $0\leq \Omega \le \OOLA$. A more formal proof
    of this claim follows via spectral flow arguments applied on the 
    parametric Hamiltonian $\tilde H_k^q(\Omega) = \tilde H_0^{[\vsigma]} + V_k^q(\Omega)$ 
    with projector $P^\vsigma$ (\cref{lem:SpectralFlow}).

    \item
    From \cref{eq:AnciHam} it follows that $P^\vsigma
    H^{[\vsigma]}(\omega)P^\vsigma =H^{[\vsigma]}(\omega)P^\vsigma=0$ and
    therefore
    \begin{align}
        P^\vsigma\tilde H_k^q(\Omega,\omega;\vsigma)P^\vsigma
        =\tilde H_k^q(\Omega,\omega;\vsigma)P^\vsigma
        = H_k^q(\Omega)P^\vsigma
        \stackrel{\text{def}}{=} [H_k^q]^\vsigma(\Omega)\,.
    \end{align}
    Note that the ground state energy of $\tilde H^q_k$ is always negative (use
    a superposition of classical ground states in $\H^\vsigma$ to show that the
    expectation value of $\tilde H^q_k$ is negative).  Since $\tilde H^q_k$ is
    block-diagonal in the symmetry sectors, and
    $\tilde\G_k^q(\Omega,\omega;\vsigma)\le\H^\vsigma$ for $\Omega \le \OOLA$
    from step (3), it follows immediately that
    $\tilde\G_k^q(\Omega,\omega;\vsigma)$ is also the ground state space of the
    symmetry-projected Hamiltonian $[H_k^q]^\vsigma(\Omega)$, i.e.,
    $\tilde\G_k^q(\Omega,\omega;\vsigma) = [\G_k^q]^\vsigma(\Omega)$.
    Furthermore, $[H_k^q]^\vsigma$ is gapped by at least $\Deltamin/2$ (uniform
    in $k$, $q$ and $\vsigma$) per step (2).

    \end{enumerate}
    This concludes the proof.
\end{proof}

\subsection{Ground state (Perron--Frobenius)}%
\label{subsec:GroundState}

In this subsection, we use the well-known theorem by Perron and
Frobenius~\cite{perronfrobenius1912} to derive properties of the ground
state(s) $\ket{\psi_k^q}$ of the Hamiltonian $H_k^q$ (which cannot be
diagonalized exactly). We prove the following lemma:
\setcounter{lemma}{1}
\begin{lemma}[Ground state]
    \label{lem:GroundState}
    There exist $\OOLB > 0$, $\LOLB \in \mathbb{N}$, such that for all $\Omega
    \in (0, \OOLB]$, $L \ge \LOLB$, and for all nonnegative integers $k$, $q$
    with $k \le \nO \le q$, the following holds:
    \begin{itemize}
   
    \item For every symmetry sector $\H^\vsigma$,
    the ground state space $[\G_k^q]^\vsigma$ of the constrained
    Hamiltonian $[H_k^q]^\vsigma$ is a subspace of $\H_k^q$.
    
    \item For $k < q$, the ground state $\ket{\psi_k^q}$ of the Hamiltonian $H_k^q$
    is non-degenerate and lies within the symmetric sector $\H^\s$, i.e.,
    $\ket{\psi_k^q} = \ket{[\psi_k^q]^\s}$ is also the unique ground state of
    $[H_k^q]^\s$. 
    
    In the natural basis, it takes the form $\ket{\psi_k^q} = \sum_{\vec{n} \in
    [\mathbb{Z}_2^{|\V|}]_k^q}(-1)^{|\vec{n}|}c_{\vec{n}}\ket{\vec{n}}$ with
    $c_{\vec{n}} > 0$, where $|\vec{n}| = \sum_{i \in \V} n_i$ denotes
    the total number of excitations.

    \end{itemize}
\end{lemma}
\noindent The proof again generalizes a special case studied by some of us in
Ref.~\cite{maier2025}:
\begin{proof}[Proof of \cref{lem:GroundState}.]

    Consider an arbitrary symmetry sector $\H^\vsigma$, and let $k$ and $q$ be
    nonnegative integers such that $k \le \nO \le q$. From
    \cref{lem:GapStability}, we know that $[H_k^q]^\vsigma \equiv P^\vsigma
    H_k^q P^\vsigma$ is gapped for $\Omega \le \OOLA$ and $L \ge \LOLA$, with a
    ground state space $[\G_k^q]^\vsigma$ that is $4$-fold degenerate. Let us
    therefore set $\OOLB:=\OOLA$ and $\LOLB:=\LOLA$. The proof splits into four
    steps:
    \begin{enumerate}[label=(\arabic*)]

    \item
    In a first step, we want to show that the ground state space
    $[\G_k^q]^\vsigma$ lies within the subspace $\H_k^q$ for all $\Omega \le
    \OOLB$. [Note that this is trivial for the special case $(k, q) = (0,
    \infty)$.]
    We start with the classical case ($\Omega=0$) where $H^q_k(0)=H_0$ is
    exactly solvable and it is easy to show that $[\G_k^q]^\vsigma(0) = \G_0
    \cap \H^\vsigma$ (\cref{lem:Stabilizer}). These classical ground states
    have $\nO$ excitations on $\p$ (by definition of $n_0$), so that
    $[\G_k^q]^\vsigma(0) \le \H_{\nO}^{\nO} \le \H_k^q$ (since $k \le \nO \le
    q$) and therefore $[\G_k^q]^\vsigma(0) = P_k^q[\G_k^q]^\vsigma(0)$  with
    $P_k^q$ the projector onto $\H_k^q$. 
    Now note that the parametric Hamiltonian $[H_k^q]^\vsigma(\Omega) =
    H_0P^\vsigma + V_k^q(\Omega)P^\vsigma$ commutes with the projector $P_k^q$
    (\cref{lem:CommutationsRelations}) and remains gapped for all
    $0\leq\Omega\leq\Omega_2'$. Hence the same continuity argument as in step
    (3) of the proof of \cref{lem:GapStability} yields
    $[\G_k^q]^\vsigma(\Omega) = P_k^q[\G_k^q]^\vsigma(\Omega)$ and therefore
    $[\G_k^q]^\vsigma \le \H_k^q$ for all $\Omega \in (0, \OOLB]$
    (\cref{lem:SpectralFlow}).
    
    Note that the ground state energy of $[H_k^q]^\vsigma$ is always negative
    (use a superposition of classical ground states in $\H^\vsigma$ to show
    that $[H_k^q]^\vsigma$ has negative expectation values).  Hence
    $[\G_k^q]^\vsigma$ is also the ground state space of the Hamiltonian
    $P_k^q[H_k^q]^\vsigma P_k^q=[H_k^q]^\vsigma P_k^q$ projected onto $\H_k^q$.
    Since this holds for every sector $\H^\vsigma$ in which $H_k^q$ is block-diagonal, 
    the global ground state $\ket{\psi_k^q}$ of $H_k^q$ lives in
    $\H_k^q$ and is therefore also the ground state of $H_k^q P_k^q$.
    
    \item
    Consider in the following $0<\Omega \leq \OOLB$ and define $Z:=
    \prod_{i\in\V}\sigma_i^z$ so that $Z\ket{\vec{n}} =
    (-1)^{|\vec{n}|}\ket{\vec{n}}$ where $|\vec{n}|=\sum_{i\in\V}n_i$ denotes
    the total number of excitations. With this, we can write $H_k^q(-\Omega) =
    ZH_k^q(\Omega)Z$ and hence $-H_k^q(-\Omega)P_k^q = -ZH_k^q(\Omega)P_k^qZ$.
    Importantly, the matrix $M_{\vec n,\vec n'}:=\bra{\vec
    n}[-H_k^q(-\Omega)P_k^q]\ket{\vec n'}$ for $\vec n,\vec n'\in
    [\mathbb{Z}_2^{|\V|}]_k^q$ is \emph{essentially nonnegative} (it has only
    nonnegative off-diagonal elements $0$ and $\Omega>0$) and is
    \emph{irreducible} in $\H_k^q$ for $k <
    q$. 
    Indeed, we have $[M^h]_{\vec{n}, \vec{n}'} \neq 0$ for any two states
    $\ket{\vec{n}}, \ket{\vec{n}'} \in \H_k^q$ where $h\in \mathbb{N}_0$ is the
    Hamming distance of their excitation patterns. Intuitively, we can always
    flip qubits one by one, while staying within the interval constrained by
    $k$ and $q$ on plaquette $\p$. Note that this argument fails for $k = q$
    since then $H_k^k(\Omega)P_k^k = H_\p(\Omega)P_k^k$ is diagonal on
    plaquette $\p$.
    
    \item
    The irreducibility and the essential nonnegativity of $M$ allow us to
    invoke the theorem by Perron and Frobenius~\cite{perronfrobenius1912}.  It
    asserts that the largest eigenvalue $-E_k^q(\Omega)$ of
    $-H_k^q(-\Omega)P_k^q = -ZH_k^q(\Omega)P_k^qZ$ is \emph{non-degenerate},
    and that its eigenvector $Z\ket{\psi_k^q(\Omega)} \equiv \sum_{\vec{n} \in
    [\mathbb{Z}_2^{|\V|}]_k^q}c_{\vec{n}}(\Omega)\ket{\vec{n}}$ can be chosen
    such that $c_{\vec{n}}(\Omega) > 0$.
    In turn, the \emph{smallest} eigenvalue $E_k^q(\Omega)$ of
    $H_k^q(\Omega)P^q_k$ is non-degenerate, and its eigenvector can be chosen
    such that $\ket{\psi_k^q(\Omega)} = \sum_{\vec{n} \in
    [\mathbb{Z}_2^{|\V|}]_k^q}(-1)^{|\vec{n}|}c_{\vec n}(\Omega)\ket{\vec{n}}$
    with $c_{\vec{n}}(\Omega) > 0$. Following step (1) above,
    $\ket{\psi_k^q(\Omega)}$ is then also the ground state of $H^q_k$ with
    ground state energy $E_k^q(\Omega)$.

    \item
    We know that $U_p$ commutes with $H_k^q$ for all $p \in \P$
    (\cref{lem:CommutationsRelations}), thus leaving all eigenspaces invariant
    -- in particular the non-degenerate ground state $\ket{\psi_k^q}$.  Since
    $U_p^2 = \mathds{1}$, its eigenvalues are $\pm 1$, and we find
    $U_p\ket{\psi_k^q} = \pm\ket{\psi_k^q}$. On the other hand, we know that
    $U_p\ket{\vec{n}} = \ket{\phi_p \cdot \vec{n}}$ acts on the natural basis
    via a permutation $\phi_p$ of the vertices in $p$.  
    Comparing coefficients by their absolute value yields $c_{\phi_p^{-1} \cdot
    \vec{n}}=c_{\vec{n}}$; comparing coefficients by their phase yields
    $(-1)^{|\phi_p^{-1} \cdot \vec n|}=\pm(-1)^{|\vec{n}|}$. Since the number
    of excited qubits is invariant under permutations, we have $|\vec{n}| =
    |\phi_p^{-1} \cdot \vec{n}|$ and the eigenvalue must be $+1$. We therefore
    conclude that $U_p\ket{\psi_k^q} = +1\ket{\psi_k^q}$ for all $p \in \P$,
    i.e., $\ket{\psi_k^q}\in\H^\s$. Since $E^q_k<0$ (see arguments for
    $[H^q_k]^{\vsigma}$ above), and $H^q_k$ is block-diagonal in the symmetry
    sectors, it follows that $\ket{\psi_k^q}=\ket{[\psi_k^q]^\s}$ is also the
    unique ground state of the symmetry-projected Hamiltonian $[H^q_k]^\s$.

    \end{enumerate}
    This concludes the proof.
\end{proof}
Note that \emph{irreducibility} and \emph{essential nonnegativity} are not
features of our specific blockade graph $\GB_L$, but rather general properties
of blockade Hamiltonians $H$ of the form \eqref{eq:HOmega}. We also
note that the uniqueness of the ground state does not contradict the
$4$-fold ground state degeneracy of our particular model: For $\Omega>0$, the
model is not at a renormalization fixpoint and we expect a finite-size
splitting of the ground state space.

\subsection{Asymmetric sectors}%
\label{subsec:AsymmetricSectors}

In this brief subsection we prove \cref{lem:AsymmetricSectors}, which exploits
the simple but important insight that states in $\H_0^0$ (i.e., without any
excitations on plaquette $\p$) are necessarily \emph{symmetric} under $U_\p$;
this can be used to derive a collective lower bound $E_1^\infty$ on the ground
state energies of all asymmetric sectors:
\setcounter{lemma}{2}
\begin{lemma}[Asymmetric sectors]
    \label{lem:AsymmetricSectors}

    Fix an arbitrary asymmetric sector $\H^\vsigma$ to define $H^\vsigma$, and
    fix a plaquette $\p \in \P$ to define $H_1^\infty$. Then the ground state
    energy of $H^\vsigma$ is not smaller than the ground state energy of
    $H_1^\infty$: $E^\vsigma \ge E_1^\infty$.

\end{lemma}
\begin{proof}[Proof of \cref{lem:AsymmetricSectors}.]

    Consider any asymmetric sector $\H^\vsigma$ and its symmetry-projected
    Hamiltonian $H^\vsigma =P^\vsigma H P^\vsigma$. Let us fix a plaquette $\p
    \in \P$ with $\sigma_\p = -1$ to define the constrained Hamiltonian
    $H_1^\infty = H_\p + \Omega\,P_1^\infty \sigma_\p^x P_1^\infty$.  Since
    states in $\H_0^0$ have no excitations on $\p$, and $\phi_\p$ permutes
    vertices only within $\p$, we have $U_\p P_0^0 = P_0^0$, i.e., all
    states in $\H_0^0$ are invariant under $U_\p$. Therefore $\H_0^0 \perp
    \H^\vsigma$ as $\H^\vsigma$ is not symmetric under $U_\p$ by assumption.
    Since $\H=\H^0_0\oplus\H^\infty_1$, we have $\H^\vsigma \le \H_1^\infty$ or
    simply $P^\vsigma P_1^\infty = P^\vsigma = P_1^\infty P^\vsigma$.
    Therefore
    \begin{align}
        [H_1^\infty]^\vsigma 
        \stackrel{\text{def}}{=} P^\vsigma H_1^\infty P^\vsigma 
        = P^\vsigma P_1^\infty H_1^\infty P_1^\infty P^\vsigma 
        = P^\vsigma P_1^\infty H P_1^\infty P^\vsigma 
        = P^\vsigma H P^\vsigma 
        \stackrel{\text{def}}{=} H^\vsigma
        \,.
    \end{align}
    As $E_1^\infty$ is the ground state energy of $H_1^\infty$, we obtain
    \begin{align}
        E_1^\infty \le \bra{\psi^\vsigma}H_1^\infty\ket{\psi^\vsigma}
        = \bra{\psi^\vsigma}[H_1^\infty]^\vsigma\ket{\psi^\vsigma}
        = \bra{\psi^\vsigma}H^\vsigma\ket{\psi^\vsigma} 
        = E^\vsigma
        \,
    \end{align}
    for the ground state $\ket{\psi^\vsigma}$ of $H^\vsigma$ with energy
    $E^\vsigma$. Note that the spectrum of $H_1^\infty$, and in particular
    $E_1^\infty$, is independent of our initial choice of $\p \in \P$ due to
    translational invariance on the periodic lattice. This concludes the proof.
\end{proof}

\subsection{Lower bound}%
\label{subsec:LowerBound}

In this subsection, we derive a lower bound on $E_k^\infty - E_{k-1}^\infty$,
i.e., the energy that is gained by allowing the excitation number on a
plaquette to fluctuate down to $k-1$ instead of $k$ excitations:
\setcounter{lemma}{3}
\begin{lemma}[Lower bound]
    \label{lem:LowerBound}

    There exist $\OOLD > 0$, $\LOLD \in \mathbb{N}$, such that for all
    $\Omega \in (0, \OOLD]$, and for a positive integer $k$ with $k \le \nO$, 
    the following holds:
    There exists a strictly increasing function $\DE_k:\;\mathbb{R}_{\ge 0}
    \rightarrow \mathbb{R}_{\ge 0}$ such that, for all $L \ge \LOLD$, one has
    $E_k^\infty - E_{k-1}^\infty \ge \DE_k(\varepsilon_k)$ with $\varepsilon_k
    := \lVert P_k^k\ket{\psi_k^\infty}\rVert$.

\end{lemma}
Note that the function $\DE_k$ is independent of the system size, but can
explicitly depend on the parameter $k$. Before starting with the proof, we need
some definitions and abbreviations:
\begin{definition}[Variational wave function] 
    \label{def:LowerBound}

    Consider some fixed plaquette $\p \in \P$ and $0<k \le \nO$. We denote by
    $\sigma_i^\pm:=\tfrac{1}{2}(\sigma_i^x\mp \mathrm{i}\sigma_i^y)$ the ladder
    operators of qubit $i$, and write $\sigma_\p^\pm := \sum_{i \in \p}
    \sigma_i^\pm$ for the sum over the qubits of plaquette $\p$. 
    If $\varepsilon_k= \lVert P_k^k\ket{\psi_k^\infty}\rVert > 0$, we define
    the \emph{variational wave function}
    \begin{align}
        \ket{\psi^{(k)}} := \alpha\ket{\psi_k^\infty} + \beta\ket{\psi_\perp^{(k)}} 
        \quad\text{with}\quad
        \ket{\psi_\perp^{(k)}} := \frac{1}{c_\perp^{(k)}}\,\sigma_\p^-P_k^k\ket{\psi_k^\infty}
        \label{eq:varwf}
    \end{align}
    for coefficients $\alpha, \beta \in \mathbb{R}$ with $\alpha^2 + \beta^2 =
    1$. Remember that $\ket{\psi_k^\infty}$ denotes the (due to finite-size
    effects) non-degenerate ground state of $H_k^\infty$, which is orthogonal
    to $\ket{\psi_\perp^{(k)}}$ as the latter has exactly $k-1$ excitations on
    plaquette $\p$. The normalizing factor $c_\perp^{(k)} := \lVert
    \sigma_\p^-P_k^k\ket{\psi_k^\infty} \rVert$ is nonzero for
    $\varepsilon_k>0$ (this is a consequence of \cref{lem:GroundState} and
    shown in \cref{lem:Bounds}). In this case, we define the quantities
    \begin{align}
        M_k := \bra{\psi_k^\infty}H_{k-1}^\infty\ket{\psi_\perp^{(k)}}
        \quad\text{and}\quad
        D_k := \bra{\psi_\perp^{(k)}}H_{k-1}^\infty\ket{\psi_\perp^{(k)}} - E_k^\infty
        \,.
        \label{eq:EnergyTerms}
    \end{align}
    By contrast, for $\varepsilon_k=0$ we set $\ket{\psi^{(k)}} :=
    \ket{\psi_k^\infty}$.
\end{definition}
\noindent With these concepts, we can now prove \cref{lem:LowerBound}:
\begin{proof}[Proof of \cref{lem:LowerBound}.]

    Let $0<k \le \nO$. To construct a lower bound on $E_k^\infty -
    E_{k-1}^\infty$, we note that
    \begin{align}
        E_k^\infty - E_{k-1}^\infty \ge
        E_k^\infty - \bra{\psi^{(k)}}H_{k-1}^\infty\ket{\psi^{(k)}}
        \quad\text{for all $\alpha,\beta\in\mathbb{R}$ with $\alpha^2+\beta^2=1$}
        \label{eq:h1}
    \end{align}
    by the variational principle, where $\ket{\psi^{(k)}} \in \H_{k-1}^\infty$
    is the normalized variational wave function of \cref{def:LowerBound} and
    $H^\infty_{k-1}$ the constrained Hamiltonian.  The strategy is to
    lower-bound the right-hand side of \cref{eq:h1}.
    
    Consider the Hamiltonian $[H_k^\infty]^\s$. From \cref{lem:GapStability},
    we know that $[H_k^\infty]^\s$ is gapped for $\Omega \le \OOLA$ and $L \ge
    \LOLA$, with a $4$-fold degenerate ground state space $[\G_k^\infty]^\s$.
    Furthermore, \cref{lem:GroundState} asserts that $\ket{\psi_k^\infty} \in
    [\G_k^\infty]^\s \le \H_k^\infty$, given that $\Omega \le \OOLB$ and $L \ge
    \LOLB$. We therefore choose $\OOLD := \min\{\OOLA, \OOLB\}$ and $\LOLD :=
    \max\{\LOLA, \LOLB\}$. Since $\ket{\psi_k^\infty} \in\H_k^\infty \le
    \H_{k-1}^\infty$, we have
    $\bra{\psi_k^\infty}H_{k-1}^\infty\ket{\psi_k^\infty}=\bra{\psi_k^\infty}H_{k}^\infty\ket{\psi_k^\infty}
    = E_k^\infty$.  
    For the special case $\varepsilon_k = 0$ with $\ket{\psi^{(k)}} =
    \ket{\psi_k^\infty}$, we then obtain $E_k^\infty - E_{k-1}^\infty\ge 0$
    from \cref{eq:h1}; we may define $\DE_k(0) := 0$. 
    In the following, we consider the nontrivial case $\varepsilon_k > 0$.

    From \cref{lem:GroundState} we know that $\ket{\psi_k^\infty} =
    \sum_{\vec{n} \in
    [\mathbb{Z}_2^{|\V|}]_k^\infty}(-1)^{|\vec{n}|}c_{\vec{n}}(\Omega)\ket{\vec{n}}$
    with $c_{\vec{n}} > 0$. Then, the amplitudes of both
    $\ket{\psi_k^\infty}$ and $\ket{\psi_\perp^{(k)}}$ [\cref{eq:varwf}] are real in
    the natural basis, so that $M_k\in\mathbb{R}$ [$D_k\in\mathbb{R}$ in any
    case, \cref{eq:EnergyTerms}]. We can then evaluate \cref{eq:h1} using
    \cref{eq:EnergyTerms,eq:varwf}:
    \begin{align}
        E_k^\infty - \bra{\psi^{(k)}}H_{k-1}^\infty\ket{\psi^{(k)}} 
        = -\left[2\alpha\beta M_k + \beta^2 D_k\right]
        \,.
    \end{align}
    In \cref{lem:Bounds} we show that $M_k \ge 0$ [see also \cref{eq:EnBounds}
    below]; in this case, the right-hand side is maximized for the real
    coefficients
    \begin{align}
        \text{\small$\displaystyle 
	    \left.
	        {
	        \begin{array}{c}
	            \alpha\,(+)\\
	            \beta\,(-)
	        \end{array}
	        }
	    \right\}
        $}
	    = \pm\frac{1}{\sqrt{2}}\sqrt{1 \pm \text{sgn}(D_k)(1 + 4M_k^2/D_k^2)^{-1/2}}
	\end{align}
    of the variational wave function. With this, we find the optimal
    variational lower bound
    \begin{align}
        E_k^\infty - E_{k-1}^\infty \ge
        \max\limits_{\substack{\alpha,\beta \in \mathbb{R}\\\alpha^2+\beta^2=1}}
        \left\{
        E_k^\infty - \bra{\psi^{(k)}}H_{k-1}^\infty\ket{\psi^{(k)}} 
        \right\}
        = \sqrt{M_k^2 + \left(\frac{D_k}{2}\right)^2} - \frac{D_k}{2}
        \,.\label{eq:ExactLowerBound}
    \end{align}
    The term on the right-hand side is nonnegative, strictly increasing in
    $|M_k|$, and weakly [for $M_k \neq 0$ strictly] decreasing in $D_k$; for
    $M_k \neq 0$ [given $D_k < \infty$] this term is positive. In
    \cref{lem:Bounds} [see \cref{eq:lem_M_k,eq:lem_Delta_k}], we derive the
    bounds
    \begin{align}
        M_k \ge \Omega\sqrt{k}\,\varepsilon_k =: M_\lb^{(k)}(\varepsilon_k)
        \quad\text{and}\quad
        D_k \le \Deltamax + \frac{\Omega(|\p| + 1)^2}{4\sqrt{k}\,\varepsilon_k} 
        =: D_\lb^{(k)}(\varepsilon_k)
        \,;
        \label{eq:EnBounds}
    \end{align}
    with $\varepsilon_k=\lVert P_k^k\ket{\psi_k^\infty}\rVert > 0$ in this branch of the proof,
    and $\Omega > 0$ by assumption. Here we defined the functions
    $M_\lb^{(k)}, D_\lb^{(k)}:\, \mathbb{R}_{> 0} \rightarrow \mathbb{R}_{> 0}$
    on the positive domain.
    We can use the inequalities \eqref{eq:EnBounds} to obtain $E_k^\infty-
    E_{k-1}^\infty \ge \DElb^{(k)}(\varepsilon_k)$ with the function
    \begin{align} 
        \label{eq:LBEnergy}
        \DElb^{(k)}\,:\,\mathbb{R}_{\ge 0} \rightarrow \mathbb{R}_{\ge 0}\,, 
        \qquad
        \varepsilon \,\mapsto\, 
        \text{\small$\displaystyle 
        \sqrt{M_\lb^{(k)}(\varepsilon)^2 + 
        \left(\frac{D_\lb^{(k)}(\varepsilon)}{2}\right)^2} - \frac{D_\lb^{(k)}(\varepsilon)}{2}
        $}
        \,.
    \end{align}
    We define $\DElb^{(k)}(0) := 0$ via the continuous extension to
    $\varepsilon = 0$.  Importantly, this is consistent with our definition
    $\DE_k(0) = 0$ from the special case $\varepsilon_k = 0$, so $E_k^\infty-
    E_{k-1}^\infty \ge \DElb^{(k)}(\varepsilon_k)$ remains valid in this limit.
    The lower bound $\DElb^{(k)}$ parametrically depends only on $\Omega$, $k$,
    $\Deltamax$ and $|\p|$, which are positive constants.  Crucially, the
    function $\DElb^{(k)}$ is independent of the system size. Also, since
    $M_\lb^{(k)}$ is strictly increasing and $D_\lb^{(k)}$ is strictly
    decreasing (in $\varepsilon$), the function $\DElb^{(k)}$ is strictly
    increasing. Hence we can identify $\DE_k \equiv \DElb^{(k)}$, which
    concludes the proof.
\end{proof}

\subsection{Single-excitation weight}%
\label{subsec:Induction}

The lower bound in \cref{lem:LowerBound} depends on the weight $\varepsilon_k
=\lVert P_k^k\ket{\psi_k^\infty}\rVert$ of the constrained ground state
$\ket{\psi_k^\infty}$ in the sector with $k$ excitations on a plaquette $\p$.
Here we show for $k=1$ that $\varepsilon_1$ can be lower-bounded by a positive
constant that is independent of the system size:
\setcounter{lemma}{4}
\begin{lemma}[Single-excitation weight] 
    \label{lem:Induction}

    There exist $\OSLE > 0$, $\LOLE \in \mathbb{N}$, such that for all
    $\Omega \in (0, \OSLE)$, there exists a constant $\varepsilon_\lb > 0$ such
    that, for all $L \ge \LOLE$, one has $\lVert P_1^1\ket{\psi_1^\infty}\rVert =
    \varepsilon_1 \ge \varepsilon_\lb$.

\end{lemma}
\begin{proof}[Proof.]

    The strategy is to prove the claim via induction on $k$, starting at $k =
    \nO$ and decrementing towards $k=1$. We begin with the induction step and
    prove the induction start below.

    \begin{itemize}

        \item \textbf{Induction step:}
        Let us assume that for some $0<k \le \nO$ there exists a bound
        $\vep^{(k)}_\lb>0$ such that $\vep_k\geq\vep^{(k)}_\lb$ for all $L\geq
        L_5'$ where $L_5'$ is a constant. We want to show that then there
        exists another bound $\vep^{(k-1)}_\lb>0$ such that
        $\vep_{k-1}\geq\vep^{(k-1)}_\lb$ for all $L\geq L_5'$ as well.
        \begin{enumerate}[label=(\arabic*)]

        \item 
        Let $0<k \le \nO$. From \cref{lem:LowerBound} we know that, for all
        $\Omega \le \OOLD$, there exists a strictly increasing function
        $\DE_k:\;\mathbb{R}_{\ge 0} \rightarrow \mathbb{R}_{\ge 0}$ such that
        for all $L \ge \LOLD$ we have
        \begin{align}
            E_k^\infty - E_{k-1}^\infty
            \ge \DE_k(\varepsilon_k) 
            \quad\text{with}\quad
            \varepsilon_k = \lVert P_k^k\ket{\psi_k^\infty}\rVert\,.
            \label{eq:LowerBound}
        \end{align}
        In the following, we therefore choose $\OSLE \le \OOLD$ and $\LOLE \ge \LOLD$.
        
        \item 
        In addition to \cref{eq:LowerBound}, we can also construct an \emph{upper} bound:
        \begin{subequations}
        \begin{align}
            E_k^\infty - E_{k-1}^\infty
            &\le \bra{\psi_{k-1}^\infty}H_k^\infty - H_{k-1}^\infty\ket{\psi_{k-1}^\infty}
            \\
            &= \Omega\bra{\psi_{k-1}^\infty}P_k^\infty\sigma_\p^xP_k^\infty
            -P_{k-1}^\infty\sigma_\p^xP_{k-1}^\infty\ket{\psi_{k-1}^\infty}
            \\
            &= -\Omega\bra{\psi_{k-1}^\infty}\sigma_\p^+P_{k-1}^{k-1} 
            + P_{k-1}^{k-1}\sigma_\p^-\ket{\psi_{k-1}^\infty}
            \\
            &\le 2\Omega\lVert \sigma_\p^+\rVert\,\lVert P_{k-1}^{k-1}\ket{\psi_{k-1}^\infty}\rVert 
            \\
            &\le (|\p| + 1)\Omega\,\varepsilon_{k-1}
            \,.
            \label{eq:UpperBound}
        \end{align}
        \end{subequations}
        From the second to the third line we used $P_k^\infty = P_{k-1}^\infty
        - P_{k-1}^{k-1}$. Note that $\sigma_\p^\pm$ increases (decreases) the
        number of excitations on $\p$ by one, so the form of the states
        $\ket{\psi_{k-1}^\infty}$, given by \cref{lem:GroundState} for $\Omega
        \in (0, \OOLB]$ and $L \ge \LOLB$, makes the matrix element negative.
        The second to last inequality follows by Cauchy--Schwarz and the
        operator norm (see \cref{subsec:OperatorNorms}), the last inequality is
        derived in \cref{lem:OperatorNorms}.
        
        \item
        Combining both inequalities \eqref{eq:LowerBound} and
        \eqref{eq:UpperBound} yields
        \begin{align}
            \varepsilon_{k-1} 
            \ge \frac{\DE_k(\varepsilon_k)}{(|\p| + 1)\Omega} 
            \ge \frac{\DE_k(\varepsilon^{(k)}_\lb)}{(|\p| + 1)\Omega} 
            \equiv \varepsilon^{(k-1)}_\lb
            \,.
            \label{eq:LowerBoundEpsk}
        \end{align}
        Remember that $\DE_k$ is strictly increasing. The second
        inequality exploits this fact, together with the assumption
        $\vep_k\geq\vep^{(k)}_\lb$ for all $L\geq L_5'$. Crucially, the new
        lower bound $\varepsilon^{(k-1)}_\lb$ is again independent of the
        system size $L$.
        Since $\vep^{(k)}_\lb>0$ it follows that $\DE_k(\varepsilon^{(k)}_\lb)
        > 0$ [recall \cref{eq:LBEnergy}] and therefore
        $\varepsilon^{(k-1)}_\lb>0$ for $0 < \Omega < \Omega_5'$.
        This concludes the induction step.

        \end{enumerate}

        \item \textbf{Induction start:}
        It remains to be shown that the assumption is justified for some
        $0<k\leq n_0$. We set $k=n_0$ and prove the existence of
        $\vep^{(n_0)}_\lb>0$ such that $\vep_{n_0}\geq \vep^{(n_0)}_\lb$ for
        $L\geq L_5'$.
        
        \emph{Notation.} In the following, we often use quantities constrained
        to exactly $n_0$ excitations on plaquette $\p$ in the symmetric sector
        $\H^\s$. To streamline our notation, we introduce $\square_\circ\equiv
        [\square^{n_0}_{n_0}]^\s$; so for example $[H_{\nO}^{\nO}]^\s \equiv
        H_\circ$ etc.
        \begin{enumerate}[label=(\arabic*)]
        \item
        We start with the Hamiltonian $H_\circ$. From \cref{lem:GapStability},
        we know that $H_\circ$ is gapped with gap $\Delta_\circ \ge
        \Deltamin/2$ for $\Omega \le \OOLA$ and $L \ge \LOLA$, with a ground
        state space $\G_\circ$ that is $4$-fold degenerate. 
        In \cref{lem:GroundState} we showed that $\G_\circ \le \H_{\nO}^{\nO}$,
        given that $\Omega \le \OOLB$ and $L \ge \LOLB$. In the following, we
        therefore assume $\OSLE \le \min\{\OOLA, \OOLB\}$ and $\LOLE \ge
        \max\{\LOLA, \LOLB\}$. 
        Let us write $P_{\G_\circ}$ for the projector onto $\G_\circ$. Since
        $\G_\circ \le \H_{\nO}^{\nO}$, we can use $P_{\G_\circ} = P_{\G_\circ}
        P_{\nO}^{\nO}$ and show that
        \begin{align}
            \varepsilon_{\nO} 
            =\lVert P_{\nO}^{\nO} \ket{\psi_{\nO}^\infty} \rVert 
            = \lVert P_{\G_\circ}\rVert \lVert P_{\nO}^{\nO} \ket{\psi_{\nO}^\infty} \rVert 
            \ge \lVert P_{\G_\circ} P_{\nO}^{\nO} \ket{\psi_{\nO}^\infty} \rVert 
            = \lVert P_{\G_\circ} \ket{\psi_{\nO}^\infty} \rVert
            \,.
            \label{eq:H1}
        \end{align}
        Here we used that the operator norm of a nonzero projector is unity
        (\cref{lem:OperatorNorms}). The remainder of the proof is dedicated to
        lower-bounding $\lVert P_{\G_\circ} \ket{\psi_{\nO}^\infty} \rVert$.
        
    \item For $\Omega$ and $L$ as above, we note the following facts:
        \begin{itemize}

        \item We have $\ket{\psi_{\nO}^\infty}\in\H_{n_0}^\infty$ and
            $\ket{\psi_{\nO}^\infty}\in\H^\s$ (\cref{lem:GroundState}).

        \item We have $\G_\circ\leq\H_\nO^\nO\leq\H_\nO^\infty$ (\cref{lem:GroundState})
        and $\G_\circ\leq \H^\s$ (since $E_\circ < 0$).

        \end{itemize}
        Hence we can decompose $\ket{\psi_{\nO}^\infty} =
        \alpha\ket{\psi_\parallel} + \beta\ket{\psi_\perp}\in \H^\s \cap
        \H_\nO^\infty$ into a state $\ket{\psi_\parallel} \in \G_\circ$ and some
        orthogonal state $\ket{\psi_\perp} \in \H^\s \cap \H_\nO^\infty$ (both normalized) with
        $P_{\G_\circ}\ket{\psi_\perp} = 0$ such that $|\alpha|^2 + |\beta|^2 =
        1$. 
        
        Note that by definition, $\ket{\psi_\perp}$ only has contributions from
        vectors perpendicular to $\G_\circ$ so that
        $\bra{\psi_\perp}H_\circ\ket{\psi_\perp} \ge E_\circ + \Delta_\circ$.
        This allows us to lower-bound the energy difference
        \begin{align}
            \Delta_{\psi_{\nO}^\infty} 
            := \bra{\psi_{\nO}^\infty}H_\circ\ket{\psi_{\nO}^\infty} - E_\circ
            \ge |\beta|^2\left(\bra{\psi_\perp}H_\circ\ket{\psi_\perp} - E_\circ\right) 
            \ge |\beta|^2\Delta_\circ
            \,.
            \label{eq:H2}
        \end{align}
        In the first inequality, we use that the cross-terms vanish and that
        $\bra{\psi_\parallel}H_\circ\ket{\psi_\parallel}-E_\circ \geq 0$.

        \item
        This allows us to rewrite Ineq.~\eqref{eq:H1} as
        \begin{align}
            \varepsilon_{\nO}^2 
            \ge \lVert P_{\G_\circ}\ket{\psi_{\nO}^\infty} \rVert^2 
            = 1 - |\beta|^2 
            \ge 1 - \Delta_{\psi_{\nO}^\infty}/\Delta_\circ
            \label{eq:varepsilonN}
            \,.
        \end{align}
        Note that $\Delta_{\psi_{\nO}^\infty}\geq 0$ since $E_\circ$ is the
        ground state energy of $H_\circ$.

        It remains to be shown that there exists an $\hat\Omega >  0$ [determined below] 
        so that $\Delta_{\psi_{\nO}^\infty} < \Delta_\circ$ for $\Omega \in (0, \hat\Omega)$.  

        \item
        In analogy to the energy difference \eqref{eq:H2}, we can define
        $\Delta_{\psi_\circ} := \bra{\psi_\circ}H_{\nO}^\infty\ket{\psi_\circ}
        - E_{\nO}^\infty$, where $\ket{\psi_\circ}$ is a ground state of
        $H_\circ$ with ground state energy $E_\circ$. The difference
        $\Delta_{\psi_\circ}$ is again nonnegative since $E_{\nO}^\infty$ is
        the ground state energy of $H_{\nO}^\infty$.  
        We find for the sum of the two energy differences:
        \begin{align}
            \begin{aligned}[b]
            \Delta_{\psi_\circ}+ \Delta_{\psi_{\nO}^\infty} 
            &\stackrel{\text{def}}{=}
            \bra{\psi_\circ}H_{\nO}^\infty\ket{\psi_\circ} - E_{\nO}^\infty
            + \bra{\psi_{\nO}^\infty}H_\circ\ket{\psi_{\nO}^\infty} - E_\circ
            \\
            &= \bra{\psi_\circ}H_\circ\ket{\psi_\circ} - E_{\nO}^\infty
            + \bra{\psi_{\nO}^\infty}H_{\nO}^\infty - \Omega \sigma_\p^x\ket{\psi_{\nO}^\infty} - E_\circ
            \\
            &= -\Omega\bra{\psi_{\nO}^\infty}\sigma_\p^x\ket{\psi_{\nO}^\infty} 
            \le \Omega\lVert \sigma_\p^x \rVert 
            = |\p|\Omega
            \,.
            \end{aligned}
            \label{eq:psiineq}
        \end{align}
        In the second row, we used $\ket{\psi_\circ} = P_{\nO}^{\nO}
        \ket{\psi_\circ} = P^\s \ket{\psi_\circ}$ and $\ket{\psi_{\nO}^\infty}
        = P_\nO^\infty \ket{\psi_{\nO}^\infty} = P^\s \ket{\psi_{\nO}^\infty}$
        to absorb and inject the projectors from the Hamiltonians [defined in
        \cref{eq:GeneralizedHamiltonian}~ff.]. In the last equality, we used
        the operator norm from \cref{lem:OperatorNorms}.

        \item
        Ineq.~\eqref{eq:psiineq} implies $\Delta_{\psi_{\nO}^\infty} \le
        |\p|\Omega$ since $\Delta_{\psi_\circ} \ge 0$. Plugging this into
        Ineq.~\eqref{eq:varepsilonN}, and using $\Delta_\circ \ge \Deltamin/2$,
        we finally obtain
        \begin{align}
            \varepsilon_{\nO}^2 
            \ge 1 - \Delta_{\psi_{\nO}^\infty}/\Delta_\circ
            \ge 1 - 2|\p|\Omega/\Deltamin \equiv \hat\varepsilon^2
            \,.
        \end{align}
        For $\hat\vep^2\geq 0$ we require $\Omega \le \hat\Omega :=
        \Deltamin/(2|\p|)$; in particular $\hat\varepsilon>0$ for $\Omega <
        \hat\Omega$. Crucially, $\hat\Omega$ is a constant independent of the
        system size.
        Since we require
        \cref{lem:GapStability,lem:GroundState,lem:LowerBound}, we set $\OSLE
        := \min\{\OOLA, \OOLB, \OOLD, \hat\Omega\}$ and $\LOLE := \max\{\LOLA,
        \LOLB, \LOLD\}$, so that for $\Omega \in (0, \OSLE)$ and $L \ge \LOLE$
        all used bounds are satisfied and we obtain $\varepsilon_{\nO} \ge
        \hat\varepsilon$ with a constant $\hat\varepsilon > 0$. This concludes
        the induction start.
        \end{enumerate}%
    \end{itemize}
    In conclusion, we showed that for $\Omega \in (0, \OSLE)$ and $L \ge \LOLE$
    we have $\varepsilon_1 \ge \vep^{(1)}_\lb$ with $\vep^{(1)}_\lb=\cL >
    0$, so that we can identify $\varepsilon_\lb \equiv \vep^{(1)}_\lb$. This
    concludes the proof.
\end{proof}

\section{Generalization to other groups}
\label{sec:nonabelian}

In this section, we sketch how \cref{the:ExcitationGap} can be extended to the
family of blockade Hamiltonians presented in Ref.~\cite{buchler2026quantum} --
which realizes the topologically ordered phases of all \emph{quantum double
models} that were introduced by Kitaev in Ref.~\cite{kitaev2003fault}. Quantum
doubles are defined for every finite group $G$, and their anyon types
correspond to the irreducible representations of the Drinfeld double
$\mathcal{D}(G)$~\cite{kitaev2003fault}. These models generalize the toric
code, which is the quantum double $\mathcal{D}(\mathbb{Z}_2)$, and exhibit
non-Abelian anyons if the group $G$ is non-Abelian. These phases are
conceptually interesting and have potential applications in quantum
computing~\cite{kitaev2003fault}. 

For any given finite group $G$, the construction from
Ref.~\cite{buchler2026quantum} yields a translation-invariant, two-dimensional
blockade Hamiltonian $H_G(\Omega)$ of the form \eqref{eq:HOmega}. For
$G=\mathbb{Z}_2$ one ends up with the Hamiltonian that we studied in this paper
and for which \cref{the:ExcitationGap} was proven [recall
\cref{fig:CoarseGraining}~(a)].  The Hamiltonians $H_G(\Omega)$ again have
local unitary plaquette symmetries $U_p(g)$ for $g \in G$ and $p\in\P$, which
are realized by permutations of qubits on the perimeter of the plaquette $p$.
While we still have $\comm{H_G(\Omega)}{U_p(g)}= 0$ and $\comm{U_p(g)}{U_q(h)} = 0$ for
$p\neq q$, we now have [cf.~\cref{eq:symmetries}]
\begin{align}
    U_p(g)U_p(h) &= U_p(gh),
    \qquad U_p(e) = \Id
    \label{eq:GroupSymmetries}
\end{align}
for all $g,h\in G$ on a single plaquette $p\in\P$, where $e$ is the identity of
$G$. In particular, plaquette symmetries on the \emph{same} plaquette need not
commute if $G$ is non-Abelian. This means that the full Hilbert space
$\H=\bigoplus_\vec\rho\H^\vec{\rho}$ decomposes into symmetry sectors
$\H^\vec{\rho}$, where each plaquette $p$ is labeled by an irreducible
representation $\rho_p$ of $G$. Crucially, for non-Abelian groups $G$, some of
these will be higher-dimensional.

\subparagraph{Abelian groups.}

Let us start with the simple case where $G$ is Abelian. The only notational
difference (compared to the main text) is that now irreducible representations
$\rho_p$ are labeled by complex phases (instead of signs
$\sigma_p$)~\cite{CharactersBookRef1}, so that the analog of
\cref{eq:SymmetrySectors} reads
\begin{align}
    \H^{\vec\rho} 
    := \spn{\ket{\psi} \in \H_L \; \vert \; \forall_{p \in \P}:\forall_{g \in G}: 
    \; U_p(g)\ket{\psi} = \rho_p(g)\ket{\psi}}\,.
    \label{eq:AbelianSymmetrySectors}
\end{align}
The symmetric sector $\H^\s$ is then given by the trivial representation on
every plaquette ($\vec\rho = \s$). Note that even the constraint $\prod_{p \in
\P} U_p(g) = \mathds{1}$ remains valid for all $g \in G$ if $G$ is Abelian
[cf.~\cref{eq:prodUp}].

Carefully retracing the proof, one finds that -- up to notational modifications
-- every step remains valid and \cref{the:ExcitationGap} holds: the crucial
ingredient needed for \cref{lem:GapStability} is \emph{local topological
quantum order} (LTQO)~\cite{michalakis2013}, which can be established along the same
lines as for the toric code [step (A6) in the proof of
\cref{lem:ConditionsForGapStability}]. Crucially, LTQO is still present within
\emph{every} symmetry sector $\H^\vec{\rho}$ (this fails for non-Abelian
groups, see below). Hence we can apply the gap stability theorem by Michalakis
and Zwolak~\cite[Theorem 1]{michalakis2013} within each symmetry sector, so
that the argument of \cref{lem:GapStability} gives a gap of $\Deltamin/2$ in
each of these sectors. The ground state space $[\G_k^q]^{\vec\rho}$ is now
$|G|^2$-fold topologically degenerate (which is an algebraic property of
quantum doubles for Abelian groups on the torus~\cite[Sec.~31.5]{Simon}). From
the argument of \cref{lem:GroundState} we can then again infer
$[\G_k^q]^{\vec\rho} \le \H_k^q$ and show that the Perron--Frobenius ground
state $\ket{\psi_k^q}$ of $H_k^q$ is non-degenerate and lives in the symmetric
sector $\H^\s$. The rest of the proof can be translated with the substitution
$\vec\sigma\mapsto\vec\rho$.

\subparagraph{Non-Abelian groups.}

For non-Abelian groups $G$, the symmetry sectors $\H^{\vec\rho}$ generically carry 
higher-dimensional irreducible representations $\rho_p$~\cite{CharactersBookRef1} 
as labels on some plaquettes $p \in \mathbb{P}$. These cause degenerate states in
the auxiliary Hamiltonian [the analog of $\tilde H_0^{[\vec\sigma]}(\omega)$
in~\cref{eq:AnciHam}], some of which correspond to the topological fusion channels of the
non-Abelian anyons (here: charges).  Since these fusion channels can be
distinguished \emph{locally} (pairs of charges can be close together), the
argument for LTQO [cf.~\cref{eq:LTQOCondition} and
step (A6) in the proof of \cref{lem:ConditionsForGapStability}] breaks down and
the prerequisites for the application of the gap stability
theorem~\cite[Theorem 1]{michalakis2013} cannot be established. This is not an
issue within the symmetric sector $\H^\s$ where all $U_p(g)$ act as the
identity, so that \cref{lem:GapStability} remains valid for this sector
($\vec\sigma\mapsto\s$). 
The missing gap assertion within asymmetric sectors is not a fatal flaw as we
never used these gaps \emph{directly} to establish the global gap
(cf.~\cref{fig:Spectrum}) -- we only needed the gap within the symmetric sector
and the energy separation between the global ground state space and the
lowest-energy states in all asymmetric sectors. However, we made use of the
gaps in the asymmetric sectors in the proof of \cref{lem:GroundState} to infer the
properties of the global ground state $\ket{\psi^q_k}$ of the constrained
Hamiltonian $H^q_k$. 

It turns out that a weakened version of \cref{lem:GroundState} can be salvaged
without using the gaps in the asymmetric sectors (\cref{lem:GroundState_prime}
in \cref{subsec:GapNonAbelian}). This version of the lemma only asserts that
$[\G^q_k]^\s\leq\H^q_k$, and that there exists a ground state (which is no
longer guaranteed to be unique) of the form
$\ket{\psi^q_k}=\sum_{\vec{n}\in\mathbb{Z}_2^{|\mathbb{V}|}} (-1)^{|\vec n|}
c_\vec{n}\ket{\vec n}$ where $c_\vec{n}\geq 0$ (and not $c_\vec{n}>0$);
$\ket{\psi^q_k}$ still lives in the symmetric sector $\H^\s$. It is then only a
matter of carefully validating the rest of the proof to show that we really
only need a ground state $\ket{\psi_k^q} \in \H_k^q \cap \H^\s$ with
nonnegative coefficients $c_{\vec{n}} \ge 0$ and alternating signs $(-1)^{|\vec
n|}$ (we can also recover $c_\vec{n}>0$ and uniqueness in $\H_k^q$ for $k<q$).  
Finally, note that the proofs of
\cref{lem:AsymmetricSectors,lem:LowerBound,lem:Induction} (which lead to
\cref{pro:GapBetweenSS}) are agnostic to the gaps (and ground state spaces) of
asymmetric sectors and therefore generalize straightforwardly.
In summary, one can show that \cref{the:ExcitationGap} is actually valid for a
much larger family of blockade Hamiltonians~\cite{buchler2026quantum}, which
realize all gapped, topologically ordered quantum phases of the quantum doubles
$\mathcal{D}(G)$ for arbitrary finite groups $G$, using only two-body blockade
interactions.

\section{Discussion} 
\label{sec:conclusion}

We have shown that the symmetric blockade Hamiltonian proposed by some of us in
Ref.~\cite{maier2025} has a stable excitation gap above the topologically
ordered toric code ground state for small transverse fields. We also sketched
how this proof can be extended to cover the more general construction proposed
in the follow-up work~\cite{buchler2026quantum}. In combination, this
rigorously establishes that the gapped ground states of these symmetric
blockade structures realize the quantum double phases. What makes this result
remarkable is that these models use only two-body interactions, are not exactly
solvable, and are not at the renormalization fixpoint of their respective
quantum phase.

The proof leverages the local symmetries of these models and exploits the
fact that the symmetric sector necessarily includes the ``all-empty'' state without excitations,
which makes the symmetric sector strictly larger (in dimension) than all other symmetry sectors. Quantum
fluctuations then convert this difference in Hilbert space dimension into a
finite energy gap. We emphasize that the proof is agnostic to the specific
structure of the blockade graphs, as long as the local symmetries exist. Hence, the
strategy may be useful for proving finite excitation gaps in other
topologically ordered models based on blockade interactions. In particular, the
proof straightforwardly applies to related constructions on different lattices,
with different lengths of links and different encodings on the links.
Conversely, the proof cannot be extended to models that lack a distinguished
symmetric sector, such as the pure $\mathbb{Z}_2$ lattice gauge theory.
Finally, we point out that the proof also applies in the limit $U \rightarrow
\infty$ of strict blockades, namely on the blockade-respecting Hilbert space
$\H_U = \ker{H_U}$ [recall \cref{eq:H0}]. This requires only the minor
modifications proposed in \cref{lem:GroundState_prime}
(\cref{subsec:GapNonAbelian}), which rely on a corollary of the
Perron--Frobenius theorem for general essentially nonnegative matrices on the
full Hilbert space.


\noindent\textbf{Data availability:}
 No data were generated or analyzed in this study.



\bibliographystyle{bibstyle.bst}
\bibliography{bibliography}

\clearpage
\appendix
\crefalias{section}{appendix}
\crefalias{subsection}{appendix}
\crefalias{subsubsection}{appendix}
\crefname{appendix}{Appendix}{Appendices}
\Crefname{appendix}{Appendix}{Appendices}
\titleformat{\section}{\bfseries\centering}{}{0pt}{}
\renewcommand{\thesection}{Appendix}
\renewcommand{\thesubsection}{\Alph{subsection}}
\makeatletter\renewcommand{\p@subsection}{}\makeatother 
\numberwithin{equation}{subsection}
\renewcommand{\theequation}{\thesubsection\arabic{equation}}

\section{Appendix} 
\label{sec:Appendix}

In these appendices we collect mathematical identities, repetitive
calculations, and general results that are disconnected from the main line of
reasoning of the proof. These technical lemmas are not shown in
\cref{fig:structure} and can be standalone results or based on other lemmas (in
which case they are responsible for some of the cross-references shown in
\cref{fig:structure}).

\subsection{Commutation relations}%
\label{subsec:CommutationsRelations}

Here we list and prove several commutation relations that are used
throughout the paper:
\begin{applem}[Commutation relations]\label{lem:CommutationsRelations}

    Let $k \le q$ be nonnegative integers. Fix an arbitrary plaquette $\p \in
    \P$ that defines $H_k^q$ [via \cref{eq:GeneralizedHamiltonian}], and let
    $\H^\vsigma$ be an arbitrary symmetry sector that defines the projector
    $P^\vsigma$.
    The following commutators vanish for all $p \in \P$:
    \begin{align}
        0 
        = \comm{H_k^q}{P_k^q} 
        = \comm{P_k^q}{U_p} 
        = \comm{H_k^q}{U_p} 
        = \comm{P_k^q}{P^\vsigma} 
        = \comm{H_k^q}{P^\vsigma} 
        \,.
        \label{eq:CommRelations}
    \end{align}
    This means that $H_k^q$ decomposes into blocks $[H_k^q]^\vsigma :=
    P^\vsigma H_k^q P^\vsigma = H_k^q P^\vsigma$ that are characterized by the
    eigenvalues $\vsigma$ of the $U_p$ for all $p \in \P$. Furthermore, $H_k^q$
    has a decoupled block $H_k^qP_k^q = P_k^qHP_k^q$ within $\H_k^q$.

\end{applem}
\begin{proof}
We prove that each commutator vanishes step-by-step; we start with the leftmost
commutator:
\begin{enumerate}[label=(\arabic*)]

    \item The projector $P_k^q$ commutes with $H_0$ as both are diagonal in the natural
        basis. It commutes trivially with 
        $\Omega\sigma_i^x$ for qubits $i \notin \p$, and it commutes with
        $\Omega P_k^q\sigma_\p^xP_k^q$ because $(P_k^q)^2 = P_k^q$.
        Therefore, $P_k^q$ commutes with every term of $H_k^q$ in
        \cref{eq:GeneralizedHamiltonian}.

    \item The graph automorphism $\phi_p$ acts locally in that it only permutes
        vertices on the same link or site. In particular, $\phi_p(\p) = \p$,
        i.e., $\phi_p$ permutes vertices within $\p$, and therefore the representation
        $U_p$ conserves the number of excitations $n_\p$ on $\p$. This implies
        that $P_k^q$ and $U_p$ commute.

    \item The classical Hamiltonian $H_0$ commutes with $U_p$ by construction.
        Additionally, we already know that $P_k^q$ and $U_p$ commute from (2). Since
        $\phi_p(\p) = \p$, $\phi_p$ conserves the set of vertices that do (not)
        belong to $\p$. Thus, $U_p$ commutes with $\Omega\sum_{i \in\p}
        \sigma_i^x$ and $\Omega\sum_{i \notin \p} \sigma_i^x$. This means that
        $U_p$ commutes with every term of $H_k^q$ in
        \cref{eq:GeneralizedHamiltonian}.

    \item We can write $P^\vsigma = \prod_{p \in \P} (\mathds{1} +
        \sigma_pU_p)/2$ for the projector onto $\H^\vsigma$. As $P_k^q$ and
        $U_p$ commute, it immediately follows that $P_k^q$ and $P^\vsigma$
        commute.

    \item As $H_k^q$ and $U_p$ commute, it immediately follows that $H_k^q$ and
        $P^\vsigma$ commute.

\end{enumerate}
This concludes the proof.
\end{proof}

\subsection{Operator norms}%
\label{subsec:OperatorNorms}

In this appendix we focus on bounds and identities of various operator norms
that are used throughout the paper.
Let us start by recapitulating the induced operator norm of a bounded linear
operator $O: \H \rightarrow \H$ on a complex Hilbert space $\H$, which is
defined via
\begin{align}
    \lVert O \rVert 
    := \sup_{\lVert \ket{\psi} \rVert = 1} \lVert O\ket{\psi} \rVert
    \,.
\end{align}
Here, $\lVert\ket{\psi}\rVert$ for $\ket{\psi} \in \H$ is the Hilbert space
norm. Besides the properties associated with every norm, the operator norm
fulfills the adjoint property $\lVert O \rVert = \lVert O^\dag \rVert =
\sqrt{\lVert O^\dag O \rVert}$, the compatibility property $\lVert O\ket{\psi}
\rVert \le \lVert O \rVert \lVert \ket{\psi} \rVert$, and the
sub-multiplicativity property $\lVert O_1 O_2 \rVert \le \lVert O_1 \rVert
\lVert O_2 \rVert$. For a normal matrix $N$ one has $\lVert N \rVert = \rho(N)$
with the spectral radius $\rho$. 
In our case, the Hilbert space norm is the Euclidean 2-norm and
its induced operator norm is called the spectral 2-norm; since $O^\dag O$ is
Hermitian and hence normal, the operator norm is
\begin{align}
    \lVert O \rVert 
    = \sqrt{\rho(O^\dag O)}
    \,.
\end{align}
With these preliminaries, we can derive some useful identities and bounds for
operator norms:
\begin{applem}[Operator norms]
    \label{lem:OperatorNorms}

    Let $\p \in \P$; we assume that the number of vertices $|\p|$ is even. Let
    $k \le |\p|$ be a nonnegative integer and let $P$ be an arbitrary nonzero
    (orthogonal) projector (e.g., $P = P_k^q$ for $k \le q \le |\p|$).  
    We can then evaluate the following operator norms exactly:
    \begin{subequations}
        \begin{align}
            \lVert P \rVert &= 1
            \,,\label{eq:L1a}\\
            \lVert \sigma_\p^x \rVert &= |\p|
            \,,\label{eq:L1b}\\
            \lVert \sigma_\p^+P_k^k \rVert 
            &= \lVert \sigma_\p^-P_{k+1}^{k+1} \rVert 
            = \sqrt{(k+1)(|\p|-k)}
            \,.
            \label{eq:L1c}
        \end{align}
        Furthermore, we can estimate the following operator norms via lower and
        upper bounds:
        \begin{alignat}{4}
            \tfrac{1}{2}|\p| \le\;& \lVert \sigma_\p^\pm \rVert 
            &&= \tfrac{1}{2}\sqrt{|\p|(|\p| + 2)} &&\le \tfrac{1}{2}(|\p| + 1)
            \,,\label{eq:L1d}\\
            \tfrac{1}{4}|\p|^2 \le\;& \lVert (\sigma_\p^\pm)^2 \rVert 
            &&= \tfrac{1}{4}|\p|(|\p| + 2) &&\le \tfrac{1}{4}(|\p| + 1)^2
            \,.\label{eq:L1e}
        \end{alignat}
    \end{subequations}
\end{applem}
\begin{proof}
    We start by defining the spin operators $S_i^\alpha := \sigma_i^\alpha/2$
    for $\alpha \in \{x, y, z\}$, and write $S^\alpha := \sigma_\p^\alpha/2
    \equiv \sum_{i \in \p} S_i^\alpha$ for their sum. Note that the $S^\alpha$
    fulfill the angular momentum algebra. Additionally, we use the ladder
    operators $S_i^\pm := \sigma_i^\mp = S_i^x \pm \mathrm{i} S_i^y$ and
    $S^\pm := \sigma_\p^\mp = S^x \pm \mathrm{i} S^y$ for notational
    convenience (note the reversed signs, cf. \cref{def:LowerBound}).  
    
    \noindent In the following, we go through each (in)equality step by step: 
\begin{itemize}

\item \cref{eq:L1a}: $P$ is a Hermitian projector with
    eigenvalues $0$ and $1$. Hence its norm is given by its spectral radius
    \begin{align}
        \lVert P \rVert 
        = \rho(P)
        = 1
        \,.
    \end{align}

\item \cref{eq:L1b}: Consider $\sigma_\p^x = \sum_{i\in\p}\sigma_i^x$, which is
    Hermitian. Hence its norm is given by its spectral radius
    \begin{align}
        \lVert \sigma_\p^x \rVert 
        = \rho(\sigma_\p^x) 
        = |\p|
        \,.
    \end{align}
    Note that the other Pauli sums $\sigma_\p^y$ and $\sigma_\p^z$ are unitarily equivalent and therefore
    have the same operator norm.
    
\item \cref{eq:L1d}: We set $O = \sigma_\p^+=\sum_{i\in\p}\sigma_i^+$ with
    the adjoint $O^\dag = \sigma_\p^-$; then we can write the product $O^\dag
    O$ as
    \begin{align}
        \sigma_\p^-\sigma_\p^+ 
        = S^+S^- 
        = (S^x)^2 + (S^y)^2 - \mathrm{i} \comm{S^x}{S^y} 
        = \vec S^2 - S^z(S^z - 1)
        \,.
    \end{align}
    This operator has eigenvalues $S(S+1) - M(M-1)$ with $S \in
    \{0,\ldots,\tfrac{1}{2}|\hat p|\}$ integer for even $|\hat p|$ and $M \in
    \{-S,\ldots,+S\}$. The maximal eigenvalue (by absolute value) is obtained
    for $S=\tfrac{1}{2}|\p|$ and $M\in\{0,1\}$, so that $\lVert
    \sigma_\p^\pm \rVert = \tfrac{1}{2}\sqrt{|\p|(|\p| + 2)}$. The inequalities
    in \eqref{eq:L1d} follow straightforwardly.

\item \cref{eq:L1e}: Set $O = (\sigma_\p^+)^2$ with the adjoint $O^\dag =
    (\sigma_\p^-)^2$; then we can write $O^\dag O$ as
    \begin{align}
        (\sigma_\p^-)^2(\sigma_\p^+)^2 
        = (S^+)^2(S^-)^2 
        &= S^+[\vec S^2 - S^z(S^z - 1)]S^-
        \nonumber\\
        &= \left[\vec S^2 - (S^z - 1)(S^z - 2)\right]
        \left[\vec S^2 - S^z(S^z - 1)\right]\,.
    \end{align}
    This operator has eigenvalues $[S(S+1) - (M-1)(M-2)][S(S+1) - M(M-1)]$ with
    $S \in \{0,\ldots,\tfrac{1}{2}|\hat p|\}$ integer for even $|\hat p|$ and
    $M \in \{-S,\ldots,+S\}$. The maximal eigenvalue (by absolute value) is
    obtained for $(S,M) = (\tfrac{1}{2}|\p|, 1)$, so that $\lVert
    (\sigma_\p^\pm)^2 \rVert = \tfrac{1}{4}|\p|(|\p| + 2)$. With \cref{eq:L1d},
    this implies that the upper bound $\lVert (\sigma_\p^\pm)^2 \rVert \le
    \lVert \sigma_\p^\pm \rVert^2$ (sub-multiplicativity) is sharp. The
    inequalities in \eqref{eq:L1e} follow straightforwardly.

\item \cref{eq:L1c}: Note that $\lVert \sigma_\p^+P_k^k \rVert = \lVert
    P_{k+1}^{k+1}\sigma_\p^+ \rVert = \lVert \sigma_\p^-P_{k+1}^{k+1} \rVert$
    by operator identity and the adjoint property, respectively.
    Hence, it is sufficient to consider $O =
    \sigma_\p^+ P_k^k$ with the adjoint $O^\dag = P_k^k\sigma_\p^-$. For
    $P_k^k$ we can write symbolically
    \begin{align}
        P_k^k
        = \delta_{k, \hat{n}_\p} 
        = \delta_{k, \tfrac{1}{2}|\p|-S^z}
        \,,
        \label{eq:Pkq}
    \end{align}
    which enforces $M^{(k)} := \tfrac{1}{2}|\p| - k$ for the quantum number of 
    $S^z$. Furthermore, \cref{eq:Pkq} shows that $P_k^k$ commutes with both
    $\vec S^2$ and $S^z$; thus we can rewrite the product $O^\dag O$ as
    \begin{align}
        P_k^k\sigma_\p^-\sigma_\p^+ P_k^k
        = P_k^k S^+ S^- P_k^k
        = [\vec S^2 - S^z(S^z - 1)] \; \delta_{k, \tfrac{1}{2}|\p|-S^z}
        \,.
    \end{align}
    This operator has the eigenvalues $[S(S + 1) - M^{(k)}(M^{(k)} - 1)]$ with $S \in
    \{|M^{(k)}|,\ldots,\hlf|\p|\}$, and eigenvalue $0$ on the kernel of $P_k^k$
    (note that $M^{(k)}$ is an integer for even $|\p|$).  
    The maximal eigenvalue (by absolute value) is obtained for $S =
    \tfrac{1}{2}|\p|$, so that $\lVert\sigma_\p^+ P_k^k\rVert =
    \sqrt{(k+1)(|\p|-k)}$.  With \cref{eq:L1a,eq:L1d}, this implies that the
    upper bound $\lVert \sigma_\p^+P_k^k \rVert \le \lVert \sigma_\p^+ \rVert$
    (sub-multiplicativity) is sharp for $k \in \tfrac{1}{2}|\p| - \{0, 1\}$ (in
    this case, we have $M^{(k)} \in \{0,1\}$, respectively).
    %
    %
\end{itemize}
This concludes the proof.
\end{proof}

\subsection{Bounds and matrix elements}%
\label{subsec:BoundsAndMatrixElements}

In \cref{def:LowerBound} we introduced the energies $M_k$ and $D_k$ that we
used in the proof of \cref{lem:LowerBound}. Here we estimate these quantities
by deriving lower and upper bounds. To this end, we first evaluate some matrix
elements needed for the derivation. This step is based on the coefficients
$(-1)^{|\vec n|}c_{\vec n}$ of the ground state $\ket{\psi_k^\infty}$ 
alternating in sign with the excitation number $|\vec n|$ (\cref{lem:GroundState}).
\begin{applem}[Bounds]
    \label{lem:Bounds}
    There exist $\OORC > 0$, $\LORC \in \mathbb{N}$, such that for $\Omega
    \in (0, \OORC]$, $L \ge \LORC$,
    and for every positive integer $k$ with $k \le \nO$, 
    the following holds: 
    Consider $\varepsilon_k = \lVert P_k^k\ket{\psi_k^\infty}\rVert$ for a fixed
    plaquette $\p \in \P$ and assume $\varepsilon_k > 0$. Then,
    \begin{subequations}
        \begin{align}
            \Omega\sqrt{k}\varepsilon_k \le M_k 
            &\le \Omega\sqrt{k}\sqrt{|\p| - k + 1}
            \,,\label{eq:lem_M_k}\\
            D_k &\le \Deltamax + \frac{\Omega(|\p| + 1)^2}{4\sqrt{k}\varepsilon_k}
            \,.\label{eq:lem_Delta_k} 
        \end{align}
    \end{subequations}
\end{applem}
\begin{proof}
    We assume $\varepsilon_k = \lVert P_k^k\ket{\psi_k^\infty}\rVert > 0$ so
    that all expressions are well-defined. When working with the constrained
    Hamiltonian $H_k^q$ [defined in \cref{eq:GeneralizedHamiltonian}], we
    implicitly fix a plaquette $\p \in \P$. Let $i, j \in \p$ and $l \in \mathbb{V}$ denote three
    (not necessarily distinct) qubits until stated otherwise.  
    We start by showing that the matrix elements 
    \begin{subequations}
        \begin{align}
        \bra{\psi_k^\infty}P_k^k\sigma_i^+\sigma_j^-P_k^k\ket{\psi_k^\infty} 
        &\ge 0 
        \label{eq:LBMEOhnen}
        \,,\\
        \bra{\psi_k^\infty}P_k^k\sigma_i^+\hat{n}_l\sigma_j^-P_k^k\ket{\psi_k^\infty} 
        &\ge 0 
        \label{eq:LBMEMitn}
        \,,
    \end{align}
    \end{subequations}
    are nonnegative.
    Let us combine both cases by writing $\hat{n}_l^s \in \{\mathds{1},
    \hat{n}_l\}$ for $s \in \{0, 1\}$, respectively.
    For $i = j \in \p$, the claim of nonnegativity is trivial, since
    \begin{align}
        \bra{\psi_k^\infty}P_k^k\sigma_i^+\hat{n}_l^s\sigma_i^-P_k^k\ket{\psi_k^\infty} 
        = \lVert \hat{n}_l^s\sigma_i^-P_k^k\ket{\psi_k^\infty}\rVert^2 
        \ge 0\,.
    \end{align}
    Now let $i \neq j$ with $i,j \in \p$. \cref{lem:GroundState} asserts that
    for $\Omega \in (0, \OOLB]$ and $L \ge \LOLB$ we can write
    $\ket{\psi_k^\infty} = \sum_{\vec{n} \in
    [\mathbb{Z}_2^{|\V|}]_k^\infty}(-1)^{|\vec{n}|}c_{\vec{n}}\ket{\vec{n}}$
    with $c_{\vec{n}} > 0$ in the natural basis. Let us therefore choose $\OORC
    = \OOLB$ and $\LORC = \LOLB$ henceforth.
    Then the matrix elements read
    \begin{align}
        \bra{\psi_k^\infty}P_k^k\sigma_i^+\hat{n}_l^s\sigma_j^-P_k^k\ket{\psi_k^\infty} 
        = \sum_{\vec{n}, \vec{m} \in [\mathbb{Z}_2^{|\V|}]_k^\infty}(-1)^{|\vec{n}| + |\vec m|}
        c_{\vec{n}}c_{\vec m}\bra{\vec{n}}P_k^k\sigma_i^+\hat{n}_l^s\sigma_j^-P_k^k\ket{\vec m}
        \,.
    \end{align}
    Here, $\bra{\vec{n}}P_k^k\sigma_i^+\hat{n}_l^s\sigma_j^-P_k^k\ket{\vec m}$
    is only nonzero if $n_h = m_h$ for qubits $h \notin \p$, and if $n_\p = k =
    m_\p$ on $\p$. (Remember that $m_\p$ is the $\ket{\vec m}$-eigenvalue of
    $\hat{n}_\p=\sum_{h \in \p}\hat n_h$.) This means in particular that for the
    total number of excitations $|\vec{n}| = |\vec m|$.
    This allows us to rewrite the matrix elements
    \begin{align}
        \bra{\psi_k^\infty}P_k^k\sigma_i^+\hat{n}_l^s\sigma_j^-P_k^k\ket{\psi_k^\infty} 
        = \sum_{\vec{n}, \vec{m} \in [\mathbb{Z}_2^{|\V|}]_k^\infty} 
        c_{\vec{n}}c_{\vec m}\bra{\vec{n}}P_k^k\sigma_i^+\hat{n}_l^s\sigma_j^-P_k^k\ket{\vec m} 
        \ge 0\,.
    \end{align}
    For the last inequality, note that $c_{\vec{n}}, c_{\vec m} > 0$, and that
    all matrix elements of $\hat{n}_l^s$, $\sigma_j^\pm$, and $P_k^k$ are
    nonnegative in the natural basis $\ket{\vec m}$. 
    We can now use these matrix elements to evaluate the quantities introduced
    in \cref{def:LowerBound}. We begin with the norm $c_\perp^{(k)} = \lVert
    \sigma_\p^-P_k^k\ket{\psi_k^\infty} \rVert$ that we use to normalize
    $\ket{\psi_\perp^{(k)}} =
    \sigma_\p^-P_k^k\ket{\psi_k^\infty}/c_\perp^{(k)}$:
    \begin{enumerate}[label=(\arabic*)]

        \item We start by upper-bounding $c_\perp^{(k)}$ using the
        compatibility of the operator norm:
        \begin{align}
            c_\perp^{(k)}
            &=
            \lVert \sigma_\p^-P_k^k\ket{\psi_k^\infty}\rVert
            \le 
            \lVert \sigma_\p^-P_k^k \rVert \lVert\ket{\psi_k^\infty}\rVert
            = \sqrt{k}\sqrt{|\p| - k + 1}
            \,.
            \label{eq:incperp1}
        \end{align}
        In the last step, we used $\lVert\ket{\psi_k^\infty}\rVert = 1$, and 
        \cref{eq:L1c} from \cref{lem:OperatorNorms} (with $k$ replaced by $k-1$)
        to evaluate the operator norm.
        We continue by lower-bounding $c_\perp^{(k)}$ using the nonnegativity
        of the matrix element \eqref{eq:LBMEOhnen}:
        \begin{align}
            [c_\perp^{(k)}]^2
            =\lVert \sigma_\p^-P_k^k\ket{\psi_k^\infty}\rVert^2
            = \sum_{i,j \in \p}\bra{\psi_k^\infty}P_k^k\sigma_i^+\sigma_j^-P_k^k\ket{\psi_k^\infty} 
            \ge \bra{\psi_k^\infty}P_k^k\hat{n}_\p P_k^k\ket{\psi_k^\infty} 
            = k\varepsilon_k^2 
            \label{eq:incperp2}
        \end{align}
        with $\varepsilon_k=\lVert P_k^k\ket{\psi_k^\infty}\rVert$. For the
        lower bound, we can drop all terms with $i \neq j$ for $i,j \in \p$
        as they are nonnegative by \cref{eq:LBMEOhnen}, and we use $\sigma_i^+\sigma_i^- = \hat n_i$.
        
        In summary, we obtain the bounds $\sqrt{k}\varepsilon_k \le
        c_\perp^{(k)} \le \sqrt{k}\sqrt{|\p| - k + 1}$.

    \item We can now estimate the energies $M_k$ and $D_k$.
        Consider $M_k$ introduced in \cref{eq:EnergyTerms} in
        \cref{def:LowerBound}; it can be evaluated straightforwardly:
        \begin{align}
            \begin{split}
                M_k 
                &= \bra{\psi_k^\infty}H_{k-1}^\infty\ket{\psi_\perp^{(k)}} 
                \\
                &= \bra{\psi_k^\infty}P_k^\infty H_{k-1}^\infty P_{k-1}^{k-1}\sigma_\p^-\ket{\psi_k^\infty}/c_\perp^{(k)}
                \\
                &= \Omega\bra{\psi_k^\infty} P_k^\infty \sigma_\p^+P_{k-1}^{k-1}\sigma_\p^-\ket{\psi_k^\infty}/c_\perp^{(k)} 
                \\
                &= \Omega\braket{\psi_\perp^{(k)} | \psi_\perp^{(k)}}c_\perp^{(k)} 
                = \Omega\,c_\perp^{(k)}
                \,.
            \end{split}
        \end{align}
        In the second line, we used that
        $\ket{\psi_k^\infty}\in\H_k^\infty$ from \cref{lem:GroundState}; in
        the third line, we used that $P_k^\infty [H_\p + \Omega\sigma_\p^-]
        P_{k-1}^{k-1} = 0$.
        Together with the Ineqs.~\eqref{eq:incperp1} and \eqref{eq:incperp2}
        for $c_\perp^{(k)}$, we directly obtain the first result
        $\Omega\sqrt{k}\varepsilon_k \le M_k \le \Omega\sqrt{k}\sqrt{|\p| - k +
        1}$.

        \item Now consider $D_k$ introduced in \cref{eq:EnergyTerms} in
        \cref{def:LowerBound}; we find
        \begin{align}
            \begin{split}
                D_k 
                &= \bra{\psi_\perp^{(k)}}H_{k-1}^\infty\ket{\psi_\perp^{(k)}} - E_k^\infty
                \\
                &= \bra{\psi_k^\infty}\sigma_\p^+
                P_{k-1}^{k-1}HP_{k-1}^{k-1}\sigma_\p^-\ket{\psi_k^\infty}/[c_\perp^{(k)}]^2 - E_k^\infty 
                \\
                &= \bra{\psi_k^\infty}P_k^\infty \comm{P_k^k\sigma_\p^+}{H} P_{k-1}^{k-1}\sigma_\p^-\ket{\psi_k^\infty}/[c_\perp^{(k)}]^2
                \,.
            \end{split}
            \label{eq:proof-dk-1}
        \end{align}
        In the third line, we use again that $\ket{\psi_k^\infty} \in
        \H_k^\infty$ from \cref{lem:GroundState} to obtain $P_k^\infty H P_k^k
        = P_k^\infty H_k^\infty P_k^k$.
        In the commutator, only three contributions survive that have support
        in $\p$:
        \begin{align}
            P_k^\infty\comm{P_k^k\sigma_\p^+}{H}P_{k-1}^{k-1} 
            = 
            \begin{aligned}[t]
            &\sum_{i \in \p}\Delta_i\sigma_i^+P_{k-1}^{k-1} 
            - \Omega (\sigma_\p^+)^2P_{k-1}^{k-1}
            \\
            &- U \sum_{\{i,j\} \in \E}
            \left(\delta_{i \in\p}\sigma_i^+\hat{n}_j + \delta_{j \in \p}\sigma_j^+\hat{n}_i\right)
            P_{k-1}^{k-1} 
            \,.
            \end{aligned}
            \label{eq:proof-dk-2}
        \end{align}
        Terms with support outside of $\p$ commute with $P_k^k\sigma_\p^+$, and
        for the detuning term we use the commutator
        $\comm{\sigma_\p^+}{\hat{n}_i} = -\sigma_i^+\delta_{i \in \p}$.
        %
        %
        We use the notation $\delta_{i \in \p} = 1$ if $i \in \p$
        and $0$ otherwise. Note that the commutator
        $P_k^\infty\comm{P_k^k\sigma_\p^+}{\sigma_\p^-}P_{k-1}^{k-1}$ vanishes
        due to the projectors.
        
        Plugging \cref{eq:proof-dk-2} into \cref{eq:proof-dk-1} yields the
        expression
        \begin{align}
            \begin{split}
                D_k 
                = 
                &\sum_{i,j \in \p}\Delta_i\frac{\bra{\psi_k^\infty} P_k^k\sigma_i^+ \sigma_j^-P_k^k \ket{\psi_k^\infty}}{[c_\perp^{(k)}]^2} 
                - \Omega \frac{\bra{\psi_k^\infty} (\sigma_\p^+)^2 \ket{\psi_\perp^{(k)}}}{c_\perp^{(k)}}
                \\
                &- U\sum_{\{i,l\} \in \E}\sum_{j \in \p}
                \frac{\bra{\psi_k^\infty} P_k^k
                \left(\delta_{i \in\p}\sigma_i^+\hat{n}_l + \delta_{l \in \p}\sigma_l^+\hat{n}_i\right)
                \sigma_j^-P_k^k\ket{\psi_k^\infty}}{[c_\perp^{(k)}]^2}
                \,, 
            \end{split}
            \label{eq:LongDkterm}
        \end{align}
        where we rediscover the matrix elements from the Ineqs.
        \eqref{eq:LBMEOhnen} and \eqref{eq:LBMEMitn} in the first and third
        term, respectively.
        
        We know already that both matrix elements are nonnegative, so we can
        upper-bound
        \begin{align}
            D_k 
            \le \Deltamax
            + \Omega \frac{|\bra{\psi_k^\infty} (\sigma_\p^+)^2 \ket{\psi_\perp^{(k)}}|}{c_\perp^{(k)}}
            \le \Deltamax
            + \Omega \frac{\lVert(\sigma_\p^+)^2\rVert}{c_\perp^{(k)}}
            \le \Deltamax + \Omega\frac{(|\p| + 1)^2}{4\sqrt{k}\varepsilon_k}
            \,. \label{eq:DkFinal}
        \end{align}
        In the first term we use the nonnegativity from \cref{eq:LBMEOhnen} to
        upper-bound $\Delta_i$ by $\Deltamax$, and we identify $\lVert \ket{\psi_\perp^{(k)}} \rVert^2 = 1$
        from the double sum. The third term is nonpositive by $U > 0$ and by \cref{eq:LBMEMitn},
        and can be upper-bounded by zero. 
        For the second inequality in \eqref{eq:DkFinal} we use Cauchy--Schwarz
        and the compatibility of the operator norm. Finally, we upper-bound the
        operator norm using \cref{eq:L1e} from \cref{lem:OperatorNorms} and
        lower-bound $c_\perp^{(k)}$ using Ineq.~\eqref{eq:incperp2}.
    \end{enumerate}
    This concludes the proof.
\end{proof}

\subsection{Ground state degeneracy}%
\label{subsec:GSD}

In this appendix, we derive the dimension of the classical ground state space
$\G_0$ and its intersection with a symmetry sector $\H^\vsigma$.
By construction of the blockade graph $\GB_L$, its maximum-weight independent
sets are in one-to-one correspondence with the simultaneous ground states of
the vertex operators of the toric code Hamiltonian, i.e., the set of all
\emph{closed-loop patterns}. For the toric code, the dimension of their span 
(in symmetry sector $\H^\vsigma$) is well known~\cite{nielsen2010}. Here we re-derive this
result in the stabilizer formalism for the classical blockade Hamiltonian $H_0$
that is induced by $\GB_L$~\cite{maier2025}.

\begin{applem}[Ground state degeneracy]
    \label{lem:Stabilizer}
    Let $\H^\vsigma$ be an arbitrary symmetry sector. The Hamiltonian $H_0$ is
    gapped by $\Deltamin$, with a degenerate ground state space $\G_0$ that has
    the dimensions
    \begin{subequations}
        \begin{align}
            \dim{(\G_0)} &= 2^{|\P| + 1}
            \,, 
            \label{eq:dimG0}\\
            \dim{(\H^\vsigma\cap\G_0)} &= 4
            \,. 
            \label{eq:dimG0sigma}
        \end{align}
    \end{subequations}
\end{applem}
\begin{proof}
    To show \cref{eq:dimG0}, note that the blockade Hamiltonian $H_0$ has been
    designed such that its ground states map one-to-one to patterns of
    \emph{closed loops} on the honeycomb graph (for details see
    Ref.~\cite{maier2025}). These patterns are generated from the empty pattern
    by application of plaquette operators $U_p$ and nontrivial loop operators
    around the torus. Due to $\prod_{p\in \P}U_p=\Id$, only $2^{|\P|-1}$
    distinct patterns can be generated by the $|\P|$ plaquette operators. Since
    plaquette operators can only change the number of non-contractible loops
    around the two periodic directions of the torus by an even number, there
    are four homologically distinct classes of loop patterns, characterized by
    the parity of the number of non-contractible loops in each direction. Hence, the total
    number of loop patterns is $4\times 2^{|\P|-1}=2^{|\P|+1}$.

    To show \cref{eq:dimG0sigma}, we can view $\H^\vsigma\cap\G_0$ as the
    subspace of $\G_0$ stabilized by the group generated by $K=|\P|-1$ pairwise
    commuting, independent plaquette operators with appropriate signs (due to $\prod_{p\in \P} \sigma_pU_p=\Id$, 
    only $|\P|-1$ stabilizer generators are independent [\cref{subsec:LocalSymmetries}]). This group does not
    contain $-\Id$. With the loop patterns as computational basis,
    we can interpret $\G_0$ as the space of $N=|\P|+1$ ``qubits''
    on which the stabilizers act as $\sigma^x$-type Pauli strings.
    It is then well known~\cite{nielsen2010} that the dimension
    of the stabilized subspace is $\dim{(\H^\vsigma\cap\G_0)}=2^{N-K}=4$.

    The Hamiltonian $H_0$ has only two relevant energy scales: $\Deltamin$ and
    $\Deltamax \equiv 2\Deltamin$ (remember that $U \gg \Deltamax$); hence the
    blockade-free eigenstates have energies that are integer multiples of
    $\Deltamin$, and the ground states are gapped by $\Deltamin$.
\end{proof}

\subsection{Conditions for gap stability}%
\label{subsec:ConditionsForGapStability}

This appendix is a (very lengthy but mostly straightforward)
prerequisite for the proof of \cref{lem:GapStability}.

\subsubsection{Review of the gap stability theorem}
We are concerned with the stability of the excitation gap of a Hamiltonian
$H_\text{ff}$ defined on a square lattice $\barN$ (see \cref{fig:CoarseGraining2} for
our specific lattice), with a tensor product Hilbert space
$\H$ of local subsystems, which \ldots
\begin{enumerate}[label=(A\arabic*)]

    \item \ldots is spatially local (\textsf{Local}). This means $H_\text{ff}=
        \sum_{\barn \in \barN} Q_\barn$ can be written as a sum of local
        operators $Q_\barn$ whose supports have a range bounded independently of $L$.

    \item \ldots satisfies periodic boundary conditions (\textsf{PBC}).

    \item \ldots is frustration-free (\textsf{Frust-Free}). This means that the ground
        states $\ket{\psi_\text{ff}}$ of $H_\text{ff}$ satisfy
        $Q_\barn\ket{\psi_\text{ff}} = q_{0,\barn}\ket{\psi_\text{ff}}$, where
        $q_{0,\barn}$ denotes the minimal eigenvalue of $Q_\barn$.

    \item \ldots has an excitation gap above the ground state space that is
        lower-bounded by $\gamma > 0$ (independent of the system size) for all
        $L \ge 2$ (\textsf{Gap}). 

        The excitation gap allows us to write $\G_\text{ff}$ for the ground
        state space of $H_\text{ff}$. 

\end{enumerate}
For gap stability, \emph{two} additional conditions must be satisfied by the
Hamiltonian $H_\text{ff}$; to formulate these, we first need further concepts:
\begin{itemize}

\item We define the ball $B_\barn(r)$ of radius $r \in \mathbb{N}_0$ around
    $\barn \in \barN$ in the usual $\ell^1$ graph metric on $\barN$. This means
    $B_\barn(r) \subseteq \V$ is just the subset of vertices that lie in nodes from
    $\barN$ within $r$ units of distance from $\barn$.

\item Consider some ``area'' $A \subseteq \V$ given by a subset of vertices. We
    define the spatially restricted Hamiltonian $H_\text{ff}\vert_A :=
    \sum_{\text{supp}(Q_\barn) \subseteq A} Q_\barn$ in the natural way, with
    contributions only from operators $Q_\barn$ with support in $A$. Here, the
    support denotes the subset of vertices on which the operator acts
    nontrivially. In particular, for $A = \V$ we recover $H_\text{ff}\vert_\V =
    H_\text{ff}$.

\end{itemize}
We call $H_\text{ff}$ \emph{locally gapped} by $\gamma_\text{loc}: \mathbb{N}_0
\rightarrow \mathbb{R}_+$ if, for all $\barn \in \barN$ and for all $r \in
\mathbb{N}_0$, the Hamiltonian $H_\text{ff}\vert_{B_\barn(r)}$ is gapped by at
least $\gamma_\text{loc}(r)$ and $\gamma_\text{loc} \in \Omega(r^{-m})$ for
some $m \in \mathbb{N}$ (i.e., decays at most polynomially). With this, we can
formulate the next condition:
\begin{enumerate}[label=(A\arabic*)]\setcounter{enumi}{4}

    \item The Hamiltonian $H_\text{ff}$ is \emph{locally gapped}
        (\textsf{Local-Gap}). 

        This allows us to write $\G_\text{ff}\vert_A$ for the ground state
        space of $H_\text{ff}\vert_A$.

\end{enumerate}

To formulate the last condition, we consider nonnegative integers $r \le L^* <
L$ and $\d \in \{1,\,\ldots\,, L - r\}$, where $L^* \sim L$ is some cutoff
parameter that scales with the system size. We consider the balls $B_\barn(r)$
and $B_\barn(r+\d)$ around some $\barn \in \barN$ of radius $r$ and $r + \d$,
respectively. 
Consider a state
$\ket{\psi_\text{ff}\vert_{B_\barn(r+\d)}}\in\G_\text{ff}\vert_{B_\barn(r+\d)}$
from the ground state space of the spatially restricted Hamiltonian
$H_\text{ff}\vert_{B_\barn(r+\d)}$ and evaluate its partially traced density
matrix
\begin{align}
    \rho_{B_\barn(r)}(\d) 
    := \Tr{\ket{\psi_\text{ff}\vert_{B_\barn(r+\d)}} \bra{\psi_\text{ff}\vert_{B_\barn(r+\d)}}}
    {B_\barn(r+\d)\backslash B_\barn(r)}
    \,. 
    \label{eq:PartTrDenM}
\end{align}
A system exhibits \emph{local topological quantum order} (LTQO) if, for any
two such states $\ket{\psi_\text{ff}\vert_{B_\barn(r+\d)}}$ and
$\ket{\psi_\text{ff}'\vert_{B_\barn(r+\d)}}$ in
$\G_\text{ff}\vert_{B_\barn(r+\d)}$, their density matrices
$\rho_{B_\barn(r)}(\d)$ and $\rho_{B_\barn(r)}'(\d)$ satisfy
\begin{align}
    \lVert \rho_{B_\barn(r)}(\d) - \rho_{B_\barn(r)}'(\d) \rVert_1 
    \le 2F(\d)
    \,,
    \label{eq:LTQOCondition}
\end{align}
with some decaying function $F: \mathbb{N} \rightarrow \mathbb{R}_+$ 
(e.g., for the toric code, $F$ is a step function)\footnote{%
In Ref.~\cite{michalakis2013} this function is denoted $\Delta_0$
(we stick to the notation $F$ of Ref.~\cite{maier2025});
the length $\d$ is denoted $l$ in Refs.~\cite{michalakis2013, maier2025}.
}. Here, $\lVert\cdot\rVert_1$ denotes the Schatten-$1$ matrix
norm. LTQO can be interpreted as the condition that \emph{local} observables
cannot distinguish between different ground states, i.e., that the ground state
degeneracy arises from our choice of topology (here: the PBCs).
%
%
With this, we can formulate the last condition:
\begin{enumerate}[label=(A\arabic*)]\setcounter{enumi}{5}

    \item The Hamiltonian $H_\text{ff}$ exhibits \emph{local topological
        quantum order} (\textsf{LTQO}).

\end{enumerate}

Finally, we need to specify a class of perturbations $V_\text{pert}$ under
which the stability of the gap of $H_\text{ff}$ is asserted. In the following,
we call $V_\text{pert}$ a $(J,f)$-perturbation of strength $J > 0$ and decay $f
\in \mathcal{O}([1+r]^{-4})$ with $f(r) \le 1$ on the domain $r \ge 0$, iff one
can write $V_\text{pert} = \sum_{\barn \in \barN} \sum_{r = 0}^L V_\barn(r)$
such that $\text{supp}(V_\barn(r)) \subseteq B_\barn(r)$ and $\lVert V_\barn(r)
\rVert \le Jf(r)$.
With all these preliminaries, the \emph{theorem on gap stability} by Michalakis
and Zwolak~\cite[Theorem 1]{michalakis2013} can be formulated as follows:
\begin{theorem}[Michalakis and Zwolak]
    \label{the:Michalakis}

    Consider a sequence of Hamiltonians $(H_\mathrm{ff})_L$ that satisfies the
    conditions $\mathrm{(A1)\text{--}(A6)}$, with an excitation gap that is
    lower-bounded by $\gamma > 0$, and a $(J,f)$-perturbation
    $V_\mathrm{pert}$.
    Then there exist constants $J_0 > 0$ and $L_0 \ge 2$ such that, for $J \le
    J_0$ and $L \ge L_0$, the excitation gap $\gamma_L - \delta_L$ of
    $H_\mathrm{ff} + V_\mathrm{pert}$ is bounded from below by
    $\gamma/2$, and the ground state splitting $\delta_L$ is
    bounded by a rapidly decaying function of the system size.

\end{theorem}
For the definitions related to the excitation gap, cf. \cref{def:ExcitationGap}.
In the remainder of this appendix, we show that the six conditions (A1)--(A6)
are met by the unperturbed auxiliary Hamiltonian \eqref{eq:AnciHam} introduced in
\cref{subsec:GapStability} so that \cref{the:Michalakis} applies, which is
crucial for the proof of \cref{lem:GapStability}.

\subsubsection{Application of the gap stability theorem}

We consider the auxiliary Hamiltonian $\tilde H_0^{[\vsigma]}(\omega) = H_0 +
H^{[\vsigma]}(\omega)$ [recall \cref{eq:AnciHam} for $\Omega=0$] with 
\begin{align}
    H^{[\vsigma]}(\omega)
    =\omega \sum_{p \in \P} (\mathds{1} - \hP_p^{[\sigma_p]})
    =\omega \sum_{p \in \P} \frac{\mathds{1} - \sigma_pU_p}{2}
\end{align}
as the unperturbed Hamiltonian $H_\mathrm{ff}$. Note that
$H^{[\vsigma]}(\omega)$ energetically penalizes symmetry sectors $\H^{\vsigma'}
\neq \H^\vsigma$. As the perturbation $V_\text{pert}$ we consider $V_k^q(\Omega) =
\Omega \left[\sum_{i \notin \p} \sigma_i^x + \sum_{i \in \p} P_k^q \sigma_i^x
P_k^q \right]$. We can then formulate and prove the following technical lemma:
\begin{applem}[Conditions for gap stability]
    \label{lem:ConditionsForGapStability}
    Let $k \le q$ be nonnegative integers and $\H^\vsigma$ be an arbitrary
    symmetry sector. Then the unperturbed Hamiltonian $H_\mathrm{ff} = \tilde
    H_0^{[\vsigma]}(\omega)$ satisfies the conditions $\mathrm{(A1)\text{--}(A6)}$ of
    \cref{the:Michalakis} and $V_\text{pert} = V_k^q(\Omega)$ is a
    $(J,f)$-perturbation.
\end{applem}

\begin{figure}[tbp]
    \includegraphics[width=0.75\linewidth]{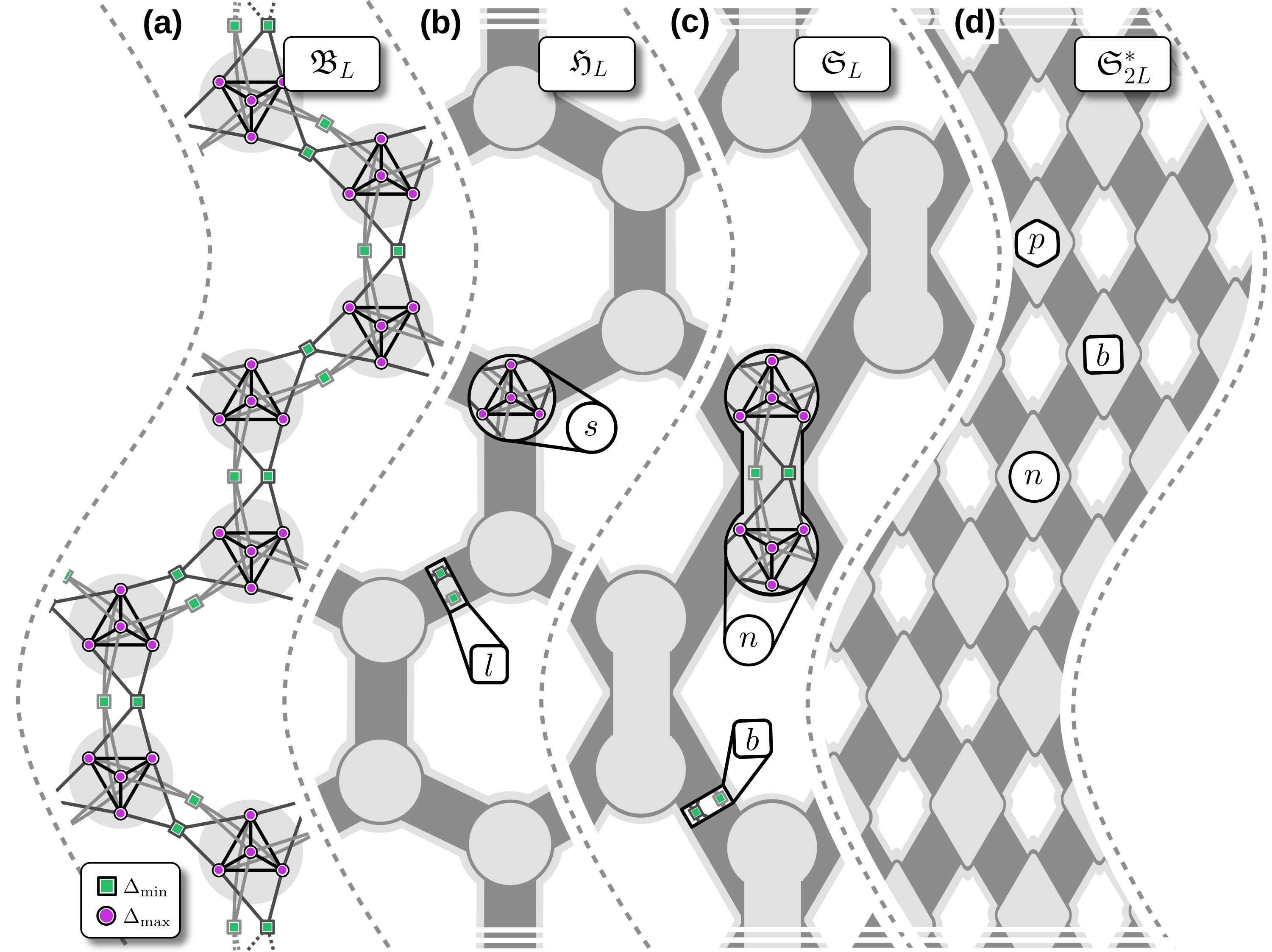}
    \caption{%
        \textbf{Coarse graining the blockade graph.}
        \textbf{(a)} In the \emph{vertex-weighted blockade graph} $\GB_L = (\V,
        \E, W)$ we identify the vertices $i \in \V$ with the qubits (squares
        and circles), the edges $\{i, j\} \in \E$ with the blockades (gray and
        black lines), and the vertex weights $\Delta_i = W(i)$ with the
        detunings (squares $\Deltamin$ and circles $\Deltamax \equiv
        2\Deltamin$).
        \textbf{(b)} In the first coarse-graining step, we identify the sites
        $s \in \S$ with four qubits from $\V$ each (gray circles), and the
        links $l \in \L$ with two qubits (thick gray lines). The constructed
        graph $\GH_L = (\S, \L)$ is a \emph{honeycomb tiling graph}.
        \textbf{(c)} In the next coarse-graining step, we identify the nodes $n
        \in \N$ with the qubits from two adjacent sites $s_n, s_n' \in \S$ in a
        fixed direction and the link $l_n$ between them (gray dumbbells). The
        bonds $b \in \B$ are identified with the two qubits of the remaining
        links (thick gray lines). The constructed graph $\NL =
        (\N, \B)$ is a \emph{square lattice graph}.
        \textbf{(d)} In the last step, we identify the nodes, plaquettes and
        bonds as the new vertex set $\barN = \N\cup \P \cup \B$ (gray
        rhomboids), such that the \emph{extended square lattice graph}
        $\barNL\equiv(\barN, \barB)$ is twice as dense.
    }
    \label{fig:CoarseGraining2}
\end{figure}

Before proving \cref{lem:ConditionsForGapStability}, it is convenient to
introduce a more precise notation for coarse-graining the honeycomb grid (see
\cref{fig:CoarseGraining2}).
We start with the blockade graph $\GB_L$ that is visualized in
\cref{fig:CoarseGraining2}~(a). We already introduced in
\cref{subsec:SpatialStructure} the \emph{honeycomb tiling graph} that emerges
by identifying sites and links with sets of vertices, see
\cref{fig:CoarseGraining2}~(b).
The honeycomb graph can be coarse-grained into a \emph{square lattice graph}
$\NL \equiv (\N, \B)$, see \cref{fig:CoarseGraining2}~(c). We refer to the
vertices $n \in \N \simeq \mathbb{Z}_L^2$ of the square graph as \emph{nodes}
and to the edges $b \in \B$ of the square graph as \emph{bonds}. As for the
honeycomb graph, we identify nodes and bonds with subsets of vertices of the
blockade graph $\GB_L$; each node $n = s_n \cup s_n' \cup l_n$ corresponds to
two adjacent sites $s_n,s_n'\in\S$ from the honeycomb graph and the link $l_n
\in \L$ between them. Furthermore, each bond $b = l_b$ corresponds to a unique
link $l_b\in\L$ of the honeycomb graph. We also introduce the shorthand
notation $b(n)$ for the subset of qubits from bonds that emanate from node $n$.
For technical reasons, it is convenient to define an \emph{extended square
lattice graph} $\barNL\equiv(\barN, \barB)$ by identifying each plaquette
$p \in \P$ and each bond $b \in \B$ as additional nodes in $\barN := \N\cup \P
\cup \B \simeq \mathbb{Z}_{2L}^2$. This makes $\barNL$ twice as dense as $\NL$, 
see \cref{fig:CoarseGraining2}~(d). To
differentiate the extended from the original square graph, we denote nodes from
the extended set by $\barn \in \barN$.

We are now prepared for the proof of \cref{lem:ConditionsForGapStability}. In
summary, we will find the gap $\gamma = \min\{\Deltamin, 2\omega\}$, the local
gap $\gamma_\text{loc}(r) \equiv \min\{\Deltamin, \omega\}$, and a
$(J,f)$-perturbation with $J = \Omega \max_{\barn \in \barN} |\barn|$ and $f(r)
= \delta_{r, 0}$. 
\begin{proof}[Proof of \cref{lem:ConditionsForGapStability}]

    Let us first establish that $V_\text{pert} = V_k^q(\Omega)$ is a
    $(J,f)$-perturbation. To this end, we define for $A \equiv B_\barn(r)$ the
    spatially restricted perturbation
    \begin{align}
        V_k^q(\Omega)|_A 
        := \Omega \left[
        \sum_{i \in A\backslash \p} \sigma_i^x 
        + \sum_{i \in A \cap \p} P_k^q \sigma_i^x P_k^q 
        \right]
    \end{align}
    in the natural way, with contributions only from vertices $i \in A$; for $A
    = \V$ we recover $V_k^q(\Omega)|_\V = V_k^q(\Omega)$.  
    We can now identify $V_\barn(r) := \delta_{r,0}\delta_{\barn \notin \P}V_k^q(\Omega)|_{B_\barn(r)}$,
    which is only nonzero for radius $r = 0$ and $\barn$ a node or a bond (not a plaquette). 
    With this, we can write
    $V_k^q(\Omega) = \sum_{\barn \in \barN} \sum_{r = 0}^L V_\barn(r)$. By
    construction, $V_\barn(r)$ has $\text{supp}(V_\barn(r)) = \barn =
    B_\barn(r)$ for $r = 0$ and $\barn \in \N \cup \B$, and 
    $\text{supp}(V_\barn(r)) = \emptyset \subset B_\barn(r)$ else.
    Then, we have $\text{supp}(V_\barn(r)) \subseteq
    B_\barn(r)$ for all $r \in \mathbb{N}_0$ and $\barn \in \barN$ as required. Additionally, since $V_\barn(r) \neq
    0$ only for $r = 0$, the perturbation has fixed range zero, so that $f(r)
    := \delta_{r, 0}$ can be chosen trivially as a step function (which
    decays sufficiently fast). We can now estimate naively $\lVert
    V_\barn(r) \rVert \le \Omega |\barn|\delta_{r, 0}$ using the triangle
    inequality of the operator norm and \cref{lem:OperatorNorms}. This allows
    us to identify $J := \Omega \max_{\barn \in \barN} |\barn|$ as the strength
    of the perturbation. Hence $V_\text{pert} = V_k^q(\Omega)$ is a
    $(J,f)$-perturbation.
    
    Let us now fix an arbitrary symmetry sector $\H^\vsigma$, which defines the
    unperturbed Hamiltonian $H_\text{ff} = \tilde H_0^{[\vsigma]}$, and show
    that the conditions (A1)--(A6) are met:
    \begin{enumerate}[label=(A\arabic*)]
    
    \item Note that $\tilde H_0^{[\vsigma]}$ is of the same structure as
    $\tilde H_0^{[\s]}$ -- up to flipped signs $\sigma_p$ in
    $H^{[\vsigma]}(\omega)$. The idea is to include the plaquettes $\P$ as
    additional nodes in the \emph{extended lattice} $\barN$ (see
    \cref{fig:CoarseGraining2}). For a plaquette $p \in \P$, we can then
    identify 
    \begin{align}
        Q_p := \omega\left(\mathds{1} - \hP_p^{[\sigma_p]}\right) 
        = \omega\frac{\mathds{1} - \sigma_pU_p}{2}
    \end{align}
    as $\omega$ times the orthogonal projector, with $\text{supp}(Q_p) = p \subseteq B_p(2)$.
    For the bonds we set $Q_b := 0$ so $\text{supp}(Q_b) = \emptyset$.  
    For $H_0 = \sum_{n \in \N} Q_n$ we can make the decomposition
    \begin{align}
        Q_n := 
        -\sum_{i \in n} \Delta_i\hat n_i 
        - \sum_{i \in b(n)} \frac{\Delta_i\hat n_i}{2} 
        + U\hspace{-2ex}\sum_{\substack{\{i,j\}\,\in\,\E \text{ for } \\ i,j\,\in\,n \cup b(n)}}\hspace{-2ex} \hat n_i\hat n_j
    \end{align} 
    with $\text{supp}(Q_n) = n \cup b(n) \subseteq B_n(1)$ by inverting the
    process of \emph{amalgamation} from Ref.~\cite{stastny2023} (note that
    there exist no in-bond blockades; these would require a factor of $1/2$). 
    This yields the spatially local decomposition 
    \begin{align}
        \tilde H_0^{[\vsigma]}(\omega) = \sum_{\barn \in \barN} Q_\barn
        \qquad
        \text{with $\text{supp}(Q_\barn) \subseteq B_\barn(2)$.}
    \end{align}
    All operators are local with constant range of support, so the condition
    \textsf{Local} is satisfied. Note that $Q_p$ (and generally $Q_\barn$)
    depends on $\sigma_p$; we usually suppress this dependence on $\vsigma$ to
    streamline the notation.  

    \item By assumption, we consider blockade Hamiltonians on the torus, i.e.,
    with periodic boundary conditions. Hence condition \textsf{PBC} is
    satisfied.

    \item It is easy to check that the unperturbed Hamiltonian $\tilde
        H_0^{[\vsigma]}$ is \emph{frustration-free}:
    
    The first term $H_0$ is the classical blockade Hamiltonian which we
    know to be frustration-free with respect to the above decomposition. The
    frustration-free ground states of the Hamiltonian $H_0 = \sum_{n \in \N} Q_n$ are given
    by the concatenation of compatible ground states of the Hamiltonians $Q_n$
    at each node $n \in \N$. This yields an extensively degenerate ground state
    space $\G_0$ spanned by the closed-loop patterns that conserve the
    $\mathbb{Z}_2$-valued flux at each site.

    The second term $H^{[\vsigma]}(\omega)$ in the full Hamiltonian introduces
    additional local operators $Q_p = \omega(\mathds{1} - \hP_p^{[\sigma_p]})$
    on each plaquette $p \in \P$. Remember that $P^\vsigma \equiv \prod_{p \in
    \P} \hP_{p}^{[\sigma_{p}]}$, so the frustration-free ground state space 
    of $H^{[\vsigma]}(\omega)$ is simply $\H^\vsigma$. 
    We calculated the dimension of the intersection space
    $\H^\vsigma \cap \G_0$ in \cref{lem:Stabilizer}; in particular, we showed
    that the space is nonzero. Therefore, the combined Hamiltonian $\tilde
    H_0^{[\vsigma]}$ is still frustration-free, and the condition
    \textsf{Frust-Free} is satisfied.

    \item Leveraging the frustration-freeness, it is now easy to show that
        $\tilde H_0^{[\vsigma]}$ is gapped:

    The ground states $\ket{\tilde \psi_0^{[\vsigma]}} \in \G_0 \cap
    \H^\vsigma$ are the states that minimize the energy of both Hamiltonians
    $H_0$ and $H^{[\vsigma]}(\omega)$ individually. Moreover, both Hamiltonians
    are gapped: The ground state space $\G_0$ of $H_0$ is gapped by $\Deltamin$
    (\cref{lem:Stabilizer}), and the ground state space $\H^\vsigma$ of
    $H^{[\vsigma]}(\omega)$ is gapped by $2\omega$ (as symmetry violations
    occur pairwise). Both parts $H_0$ and $H^{[\vsigma]}(\omega) = \omega
    \sum_{p \in \P} (\mathds{1} - \sigma_pU_p)/2$ commute due to
    \cref{eq:symmetries}, so the composite Hamiltonian $\tilde H_0^{[\vsigma]}
    = H_0 + H^{[\vsigma]}(\omega)$ is gapped by $\gamma :=
    \min\{\Deltamin, 2\omega\}$ (independent of the system size) for all $L$.
    We conclude that the condition \textsf{Gap} is fulfilled. In the following,
    we denote the degenerate ground state space of $\tilde H_0^{[\vsigma]}$ by
    $\tilde \G_0^{[\vsigma]} := \G_0 \cap \H^\vsigma$.

    \item The strategy to show the local gap condition is similar to that used
    for frustration-freeness (A3) and the excitation gap (A4). Let us fix
    an arbitrary $\barn \in \barN$ and $r \in \mathbb{N}_0$, and abbreviate
    $A \equiv {B_\barn(r)}$.
    We consider $\tilde H_0^{[\vsigma]}\vert_A = H_0\vert_A +
    H^{[\vsigma]}(\omega)\vert_A$ with the local decomposition
    \begin{align}
        H_0\vert_A = \sum_{\text{supp}(Q_n) \subseteq A} Q_n
        \qquad\text{and}\qquad
        H^{[\vsigma]}(\omega)\vert_A = \sum_{\text{supp}(Q_p) \subseteq A} Q_p
        \,.
        \label{eq:H0A}
    \end{align}
    The first term $H_0\vert_A$ is again a classical blockade Hamiltonian, and
    we know that its decomposition in \eqref{eq:H0A} is frustration-free (since
    the full Hamiltonian $H_0 = H_0\vert_\V$ is frustration-free). The
    Hamiltonian $H_0\vert_A$ has an extensively degenerate ground state space
    $\G_0|_A$, which includes the patterns of closed loops that conserve the
    $\mathbb{Z}_2$-valued flux at sites $s \subset n$ with $\text{supp}(Q_n)
    \subseteq A$. The excited eigenstates are characterized by de-excited
    qubits and are gapped by $\Deltamin$ from $\G_0|_A$
    (recall \cref{lem:Stabilizer}).

    Let us now consider a state $\ket{\vec{n}|_{A}} \in \G_0|_A$ with a
    classical closed-loop pattern on $A$ and define 
    \begin{align}
        \ket{\tilde \psi_0^{[\vsigma]}|_A} 
        := \left[\,
        \prod_{\text{supp}(Q_p) \subseteq A} \sqrt{2}\hP_{p}^{[\sigma_{p}]}
        \,\right]\ket{\vec{n}|_{A}}
        \,. 
        \label{eq:nloopA}
    \end{align}
    For $A \neq \V$ [remember $A = \V$ is covered in (A4)], the plaquette operators $U_p$ with $p = \text{supp}(Q_p)
    \subseteq A$ are functionally independent. Thus, the loop states obtained from
    applying the $\hP_{p}^{[\sigma_{p}]}$ are mutually orthogonal 
    and $\ket{\tilde \psi_0^{[\vsigma]}|_A}$ is normalized. 
    By construction, $\ket{\tilde \psi_0^{[\vsigma]}|_A}$ minimizes the energy
    of every term $Q_p$ in $H^{[\vsigma]}(\omega)\vert_A$ individually (with
    the lowest eigenvalue $q_{0,p} = 0$). The excited eigenstates are
    characterized by violations of the symmetry condition and are gapped by
    $\omega$ from the ground state space $\H^\vsigma\vert_A$ of
    $H^{[\vsigma]}(\omega)\vert_A$.

    Additionally, since $H_0$ and $\hP_{p}^{[\sigma_{p}]}$ commute by
    construction, the state $\ket{\tilde \psi_0^{[\vsigma]}|_A}$ is still a
    frustration-free ground state of $H_0\vert_A$ in $\mathcal{G}_0|_A$. This
    means $\ket{\tilde \psi_0^{[\vsigma]}|_A} \in \mathcal{G}_0|_A \cap
    \H^\vsigma\vert_A$, i.e., $\ket{\tilde \psi_0^{[\vsigma]}|_A}$ is a
    frustration-free ground state of the combined Hamiltonian $\tilde
    H_0^{[\vsigma]}|_A$. Both parts $H_0\vert_A$ and
    $H^{[\vsigma]}(\omega)\vert_A$ commute, so the excited states are separated
    from the ground state space by at least $\min\{\Deltamin,\omega\} > 0$ (independent of $r$); that is,
    $\tilde H_0^{[\vsigma]}$ is locally gapped with
    $\gamma_\text{loc}(r):= \min\{\Deltamin,\omega\}$. We conclude that the
    condition \textsf{Local-Gap} is satisfied. In the following, we write
    $\tilde{\mathcal{G}}_0^{[\vsigma]}|_A := \mathcal{G}_0|_A \cap
    \H^\vsigma\vert_A$ for the degenerate ground state space.

    \item Finally, we show that the Hamiltonian $\tilde H_0^{[\vsigma]}$
    exhibits LTQO for arbitrary $\vsigma \in \{\pm1\}^{|\P|}$. The idea is
    to reduce the problem to the symmetric case $\vsigma = \s$ -- which has
    already been treated in detail in Ref.~\cite[Appendix~D.4]{maier2025}.
    In the following, we often restrict expressions to the ground state space
    $\G_0$ of the classical part $H_0$ of the Hamiltonian. For equalities that
    are valid only within this subspace we write $\stackrel{\G_0}{=}$.
    
    Remember that we always have $\prod_{p \in \P} \sigma_p = 1$ (otherwise
    $\H^\vsigma$ is trivial), that is, the number of plaquettes with $\sigma_p
    = -1$ is \emph{even}. This allows us to pair up the plaquettes $p \in \P$ with
    $\sigma_p = -1$ by connecting them via a path on the \emph{dual}
    lattice. Without loss of generality, we can assume that each link
    contributes to at most one path, i.e., the paths are disjoint (we can
    perform ``surgery'' to drop links that are traversed twice in the two paths --
    this conserves the parity condition at every plaquette). We associate each
    path with the set of vertices of its links, and collect the contributing
    link vertices of all paths in a unified vertex set $\bar\Gamma$, see
    \cref{fig:ZLoop}.

    \begin{figure*}[tb]
        \includegraphics[scale=0.8]{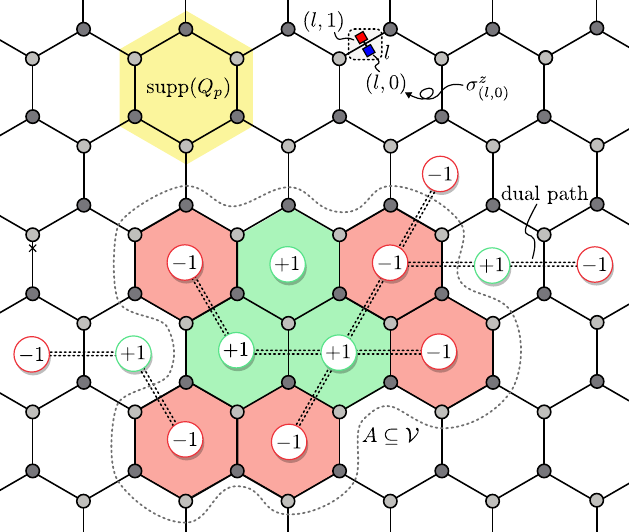}
        \caption{%
            \textbf{Symmetries and plaquettes.} 
            We consider periodic boundary conditions. One finds algebraically
            that there always exists an even number of plaquettes $p \in \P$
            with eigenvalues $\sigma_p = -1$ (red circles). One can connect
            these plaquettes pairwise via disjoint paths on the dual lattice
            (double-dashed lines). The vertices $i \in l$ of links $l$
            traversed by the dual paths define the vertex set $\bar\Gamma$.
            Then, each plaquette with $\sigma_p = -1$ has an \emph{odd} number
            of bounding links $l \subseteq \bar\Gamma$. The remaining
            plaquettes have $\sigma_p = +1$ (green or no circles) and an
            \emph{even} number of bounding links $l \subseteq \bar\Gamma$. As
            an example, the yellow area marks the support of $Q_p$ for a plaquette
            $p$, which fully incorporates the adjacent sites $s \subset p$ and
            the bounding links $l \subset p$.  The dashed area $A \subseteq \V$
            incorporates $8$ plaquettes $p \in \P$ whose support
            $\text{supp}(Q_p) = p$ fully lies in $A$. Each such plaquette is
            colored red (green) if $\sigma_p = -1$ ($\sigma_p = +1$). To each
            link $l \in \L$ there are associated two qubits $(l,0)$ [blue] and
            $(l,1)$ [red]. For each link $l$ that is both in $\bar\Gamma$ and
            $A$, the Pauli $z$-matrix $\sigma_{(l,0)}^z$ is included in the
            path operator $Z_A^\text{path}$.
        }
        \label{fig:ZLoop}
    \end{figure*}

    By construction of the blockade graph $\mathfrak{B}_L$, each link $l$ of
    the honeycomb graph is associated with two vertices, which we label by
    $(l,0)$ and $(l,1)$.  Consider now an ``area'' $A \subseteq \V$ and define
    the path operator
    \begin{align}
        Z_A^\text{path} 
        := \prod_{l \subseteq \bar\Gamma \cap A} \sigma_{(l,0)}^z
        \,,
    \end{align}
    which includes one Pauli $z$-matrix for each link along the paths in $A$.
    Note that $Z_A^\text{path}$ is a product of Pauli $z$-matrices, and hence
    is Hermitian, unitary, and squares to the identity.

    The path operator $Z_A^\text{path}$ commutes with $Q_n$ (since both operators are diagonal in
    the natural basis) so that $Z_A^\text{path}$ preserves the $Q_n$-eigenspace
    $\G_0\vert_A$, as do $\tilde H_0^{[\vsigma]}\vert_A$ and in particular $\tilde H_0^{[\s]}\vert_A$. It
    therefore makes sense to restrict the Hilbert space to the ground state
    space $\G_0\vert_A$. By construction, for each plaquette $p \in \P$ with
    $\text{supp}(Q_p)=p \subseteq A$ and $\sigma_p = +1$ ($\sigma_p = -1$),
    there exists an even (odd) number of outgoing dual paths, i.e., an even
    (odd) number of links $l \subseteq \bar\Gamma \cap p$. Within the ground
    state space $\G_0\vert_A$, on every link $l \in \L$ with 
    $l \subseteq \text{supp}(Q_n) \subseteq A$ for some $n \in \N$, 
    exactly one qubit is excited,
    so we can rewrite $Z_A^\text{path} \stackrel{\G_0\vert_A}{=} \prod_{l
    \subseteq \bar\Gamma \cap A} \left(\sigma_{(l,0)}^z -
    \sigma_{(l,1)}^z\right)/2$. Thus, in $\G_0\vert_A$, the product operator
    $Z_A^\text{path}$ (anti)commutes with the operators $U_p$ of plaquettes $p
    \subseteq A$ for which $\sigma_p = +1$ ($\sigma_p = -1$).
    We can therefore perform a unitary transformation
    \begin{align}
        Z_A^\text{path}Q_n(Z_A^\text{path})^\dag = Q_n
        \qquad\text{and}\qquad
        Z_A^\text{path}Q_p^{[\sigma_p]}(Z_A^\text{path})^\dag 
        \stackrel{\G_0\vert_A}{=} Q_p^{[1]}
        \,.
        \label{eq:Similarity}
    \end{align}
    Here and in the following, we use the notations $Q_p^{[\vsigma]} \equiv
    Q_p^{[\sigma_p]} \equiv \omega(\mathds{1} - \hP_p^{[\sigma_p]})$ and
    $Q_n^{[\vsigma]} \equiv Q_n$ interchangeably.
    \Cref{eq:Similarity} implies that the Hamiltonians $\tilde
    H_0^{[\vsigma]}|_A = \sum_{\text{supp}(Q_\barn^{[\vsigma]}) \subseteq A}
    Q_\barn^{[\vsigma]}$ and $\tilde H_0^{[\s]}|_A$ are unitarily equivalent in
    $\G_0\vert_A$. In particular, both Hamiltonians have the same spectrum in
    $\G_0\vert_A$.
    Indeed, $\ket{\tilde \psi_0^{[\vsigma]}\vert_A} \in
    \G_0\vert_A$ is an eigenvector of $\tilde H_0^{[\vsigma]}|_A$ for some
    eigenvalue, if and only if
    \begin{align}
        \ket{\tilde \psi_0^{[\s]}\vert_A} 
        := Z_A^\text{path}\ket{\tilde \psi_0^{[\vsigma]}\vert_A}
        \label{eq:ZLoopStates}
    \end{align}
    is again in $\G_0\vert_A$ and an eigenvector of $\tilde H_0^{[\s]}|_A
    \stackrel{\G_0\vert_A}{=} Z_A^\text{path}\tilde
    H_0^{[\vsigma]}|_A(Z_A^\text{path})^\dag$ with the same eigenvalue.

    To show LTQO, we consider the two balls $B_\barn(r)$ and $B_\barn(r+\d)$ and
    naturally split
    \begin{align}
        Z_{B_\barn(r+\d)}^\text{path} 
        = Z^\text{path}_{B_\barn(r)} 
        \otimes Z^\text{path}_{B_\barn(r+\d)\backslash B_\barn(r)}
    \end{align}
    into a tensor product of unitaries acting on the different regions.
    We are interested in the (frustration-free) ground states $\ket{\tilde
    \psi_0^{[\vsigma]}\vert_{B_\barn(r+\d)}}$ in
    $\tilde{\mathcal{G}}_0^{[\vsigma]}|_{B_\barn(r+\d)}$ of the Hamiltonian
    $\tilde H_0^{[\vsigma]}|_{B_\barn(r+\d)}$ in region $A = B_\barn(r+\d)$
    [e.g., $\ket{\tilde \psi_0^{[\vsigma]}\vert_A}$ defined in
    \cref{eq:nloopA}].
    To evaluate the subsystem density matrix
    $\tilde{\rho}_{B_\barn(r)}^{[\vsigma]}(\d)$ [defined in
    \cref{eq:PartTrDenM}], we rewrite the density matrix in terms of the
    subsystem density matrix $\tilde{\rho}_{B_\barn(r)}^{[\s]}(\d)$ of a ground
    state $\ket{\tilde \psi_0^{[\s]}\vert_{B_\barn(r+\d)}} \in
    \tilde{\mathcal{G}}_0^{[\s]}|_{B_\barn(r+\d)}$ in the symmetric sector
    [given by \cref{eq:ZLoopStates}]:
    \begin{align}
        \tilde{\rho}_{B_\barn(r)}^{[\vsigma]}(\d) 
        = Z_{B_\barn(r)}^\text{path}\tilde{\rho}_{B_\barn(r)}^{[\s]}(\d)Z_{B_\barn(r)}^\text{path}
        \,.
    \end{align}
    Here we used that $Z^\text{path}_{B_\barn(r+\d)\backslash B_\barn(r)}$
    cancels with itself due to the cyclicity of the (partial) trace.
    Crucially, the Schatten-$1$ norm in the definition \eqref{eq:LTQOCondition}
    of LTQO is unitarily invariant. This implies that the unitaries
    $Z_{B_\barn(r)}^\text{path}$ drop out of the norm:
    \begin{align}
        \lVert \tilde\rho_{B_\barn(r)}^{[\vsigma]}(\d) - \tilde{\rho}'\-_{B_\barn(r)}^{[\vsigma]}(\d) \rVert_1 
        \begin{aligned}[t]
        = \lVert\tilde{\rho}_{B_\barn(r)}^{[\s]}(\d) - \tilde{\rho}'\-_{B_\barn(r)}^{[\s]}(\d)\rVert_1 
        \le 2F(\d)
        \,.
        \end{aligned}
    \end{align}
    The last inequality for the symmetric case follows from Ref.~\cite[Appendix
    D.4]{maier2025}, where it was shown that $F$ can be chosen as a step
    function. Hence we conclude that the condition \textsf{LTQO} is satisfied
    for $\tilde H_0^{[\vsigma]}$ in every symmetry sector $\H^\vsigma$.
        
    \end{enumerate}
    In summary, we showed that all conditions (A1)--(A6) needed for
    \cref{the:Michalakis} are satisfied by the unperturbed Hamiltonian $\tilde
    H_0^{[\vsigma]}$, and that $V_k^q(\Omega)$ is a $(J,f)$-perturbation.
\end{proof}

\subsection{Spectral flow}%
\label{subsec:SpectralFlow}

In the proofs of \cref{lem:GapStability,lem:GroundState} we argued that the
ground state space of a smoothly parameter-dependent Hamiltonian remains within a fixed
decoupled subspace as long as the gap does not close. For example, we used that the
ground state space $\tilde\G_k^q(\Omega,\omega;\vsigma)$ of the auxiliary
Hamiltonian \eqref{eq:AnciHam} remains a subspace of the symmetry sector
$\H^\vsigma$ for all $0\leq\Omega\leq\Omega_1'$. Here we show this rigorously
using spectral flow arguments in a generic setting:
\begin{applem}[Spectral flow]
    \label{lem:SpectralFlow}
    
    Let $H(\Omega) := H_0 + V(\Omega)$ for $\Omega \in [0,\Omega_{\max}]$ be a
    smooth family of Hermitian operators on a finite-dimensional Hilbert space $\H$
    with $V(0) = 0$. Let
    $\G(\Omega)$ denote the ground state space of $H(\Omega)$, and assume that
    there is a $\gamma > 0$ such that $H(\Omega)$ is gapped by at least
    $\gamma$ above $\G(\Omega)$ for all $\Omega \in [0,\Omega_{\max}]$. Let $P$
    be a projector with $\com{H(\Omega)}{P} = 0$ for all $\Omega \in
    [0,\Omega_{\max}]$ such that $\G(0) \leq P\H$. Then $\G(\Omega) \leq
    P\H$ for all $\Omega \in [0,\Omega_{\max}]$.

\end{applem}
\begin{proof}

    Note that $\com{H(\Omega)}{P} = 0$ for all $\Omega$ implies $\com{H_0}{P} =
    0$, $\com{V(\Omega)}{P} = 0$ (since $V(0)=0$), and therefore also
    $\com{V'(\Omega)}{P} = 0$.
    The existence of a gap allows us to apply a theorem on spectral flow by
    Nachtergaele~\etal~\cite[Theorem 6.3]{nachtergaele2019} to the
    parameter-dependent Hamiltonian $H(\Omega) = H_0 + V(\Omega)$ with
    parameter $\Omega$. 
    The theorem asserts that $\G(\Omega) = U(\Omega)\G(0)$ (for elementwise
    operation), with the unitary spectral flow $U(\Omega)$ uniquely determined
    by
    \begin{align}
        \frac{d}{d\Omega} U(\Omega) 
        = -\mathrm{i}D(\Omega)U(\Omega) 
        \qquad\text{and}\qquad 
        U(0) = \mathds{1}
        \,,\label{eq:Nachtergaele1}
    \end{align}
    where $D(\Omega) := \int_\mathbb{R}e^{\mathrm{i}tH(\Omega)}\,V'(\Omega)\,e^{-\mathrm{i}tH(\Omega)}W(t)\, dt$
    for some weight function $W \in L^1(\mathbb{R})$ that defines the spectral
    flow. It is guaranteed that a unique solution of \cref{eq:Nachtergaele1}
    exists.
    
    Let us assume that $\com{U(\Omega)}{P}=0$; then, with elementwise
    operation on $\G(\Omega)$, we can write
    \begin{align}
        P\G(\Omega) 
        = P U(\Omega)\G(0)
        = U(\Omega)P\G(0)
        = U(\Omega)\G(0)
        = \G(\Omega)
        \,,
    \end{align}
    and therefore $\G(\Omega) \le P\H$. Thus it remains to be shown that
    $\com{U(\Omega)}{P}=0$.
    To this end, define the Hermitian unitary $U_P := \mathds{1} - 2P$, which
    commutes with $H(\Omega)$ and $V'(\Omega)$ since these operators commute
    with $P$. Then, $U_P$ commutes with $D(\Omega)$, and for
    \cref{eq:Nachtergaele1} we find
    \begin{align}
        \frac{d}{d\Omega} [U_P^\dag U(\Omega)U_P] 
        = U_P^\dag[-\mathrm{i}D(\Omega)U(\Omega)]U_P 
        = -\mathrm{i}D(\Omega)[U_P^\dag U(\Omega)U_P]\,,
        \qquad
        U_P^\dag U(0)U_P = \mathds{1}
        \,;\label{eq:h2}
    \end{align}
    hence by uniqueness of the solution
    \begin{align}
        U(\Omega) 
        = U_P^\dag U(\Omega)U_P 
        = U(\Omega) 
        - 2P U(\Omega) 
        - 2U(\Omega)P  
        + 4P U(\Omega)P 
        \,.\label{eq:h3}
    \end{align}
    Left and right multiplication by $P$ yields $P U(\Omega) = P U(\Omega)P =
    U(\Omega)P$ and therefore $\com{U(\Omega)}{P}=0$ as required. 
\end{proof}

\subsection{Weakened gap stability for non-Abelian groups}%
\label{subsec:GapNonAbelian}

Here we show the weakened version of \cref{lem:GroundState}, which we used in
\cref{sec:nonabelian} to argue that \cref{the:ExcitationGap} is actually valid
for the much larger class of blockade Hamiltonians introduced in
Ref.~\cite{buchler2026quantum}. To this end, we assume that the Hilbert space
splits into symmetry sectors $\H^\vec{\rho}$, labeled by (potentially
higher-dimensional) irreducible representations $\rho_p$ of some finite group $G$ on
each plaquette $p\in\P$. Crucially, we only assume that \cref{lem:GapStability} 
holds for the symmetric sector $\H^\s$, where the plaquette symmetries act
with the trivial representation on every plaquette. Then one can show:

\setcounter{lemma}{1}
\renewcommand{\thelemma}{\arabic{lemma}$\,'$}
\renewcommand{\theHlemma}{nonabelian.\arabic{lemma}}

\begin{lemma}[Ground state]
    \label{lem:GroundState_prime}

    There exist $\OOLB > 0$, $\LOLB \in \mathbb{N}$, such that for all $\Omega
    \in (0, \OOLB]$, $L \ge \LOLB$, and for all nonnegative integers $k$, $q$
    with $k \le \nO \le q$, the following holds:
    \begin{itemize}
   
    \item The ground state space $[\G_k^q]^\s$ of the constrained
    Hamiltonian $[H_k^q]^\s$ is a subspace of $\H_k^q$.
    
    \item There exists a ground state $\ket{\psi_k^q} = \sum_{\vec{n} \in
        [\mathbb{Z}_2^{|\V|}]_k^q}(-1)^{|\vec{n}|}c_{\vec{n}}\ket{\vec{n}}$ of
        the Hamiltonian $H_k^q$ with coefficients $c_{\vec{n}} \ge 0$ in the
        natural basis.
        This ground state $\ket{\psi_k^q}$ lies in the symmetric sector $\H^\s$, i.e., $\ket{\psi_k^q}
        = \ket{[\psi_k^q]^\s}$ is also a ground state of $[H_k^q]^\s$.

    \end{itemize}

\end{lemma}
\begin{proof}[Proof.]

    Consider the symmetric sector $\H^\s$, and let $k$ and $q$ be nonnegative
    integers such that $k \le \nO \le q$. From \cref{lem:GapStability}, we know
    that $[H_k^q]^\s \equiv P^\s H_k^q P^\s$ is gapped for $\Omega \le \OOLA$
    and $L \ge \LOLA$, with a ground state space $[\G_k^q]^\s$ the degeneracy
    of which depends on the group $G$. Let us therefore set $\OOLB:=\OOLA$ and
    $\LOLB:=\LOLA$. The proof splits into five steps:
    \begin{enumerate}[label=(\arabic*)]

    \item
    In a first step, we want to show that the ground state space $[\G_k^q]^\s$
    lies within the subspace $\H_k^q$ for all $\Omega \le \OOLB$. [Note that
    this is trivial for the special case $(k, q) = (0, \infty)$.]
    We start with the classical case ($\Omega=0$) where $H^q_k(0)=H_0$ is
    exactly solvable and it is easy to show that $[\G_k^q]^\s(0) = \G_0 \cap
    \H^\s$ (cf.~\cref{lem:Stabilizer}). These classical ground states have
    $\nO$ excitations on $\p$ (by definition of $n_0$), so that $[\G_k^q]^\s(0)
    \le \H_{\nO}^{\nO} \le \H_k^q$ (since $k \le \nO \le q$) and therefore
    $[\G_k^q]^\s(0) = P_k^q[\G_k^q]^\s(0)$  with $P_k^q$ the projector onto
    $\H_k^q$. 
    Now note that the parametric Hamiltonian $[H_k^q]^\s(\Omega) = H_0P^\s +
    V_k^q(\Omega)P^\s$ commutes with the projector $P_k^q$
    (cf.~\cref{lem:CommutationsRelations}) and remains gapped for all
    $0\leq\Omega\leq\Omega_2'$. Hence the same continuity argument as in step
    (3) of the proof of \cref{lem:GapStability} yields $[\G_k^q]^\s(\Omega) =
    P_k^q[\G_k^q]^\s(\Omega)$ and therefore $[\G_k^q]^\s \le \H_k^q$ for all
    $\Omega \in (0, \OOLB]$ (\cref{lem:SpectralFlow}).
    Note that the ground state energy of $[H_k^q]^\s$ is always negative (use a
    superposition of classical ground states in $\H^\s$ to show that
    $[H_k^q]^\s$ has negative expectation values). Hence $[\G_k^q]^\s$ is also
    the ground state space of the Hamiltonian $P_k^q[H_k^q]^\s P_k^q=[H_k^q]^\s
    P_k^q$ projected onto $\H_k^q$.

    In contrast to the proof of \cref{lem:GroundState} in \cref{subsec:GroundState}, we now
    \emph{cannot} conclude that the \emph{global} ground state $\ket{\psi_k^q}$
    of $H_k^q$ lives in $\H_k^q$ since we lack knowledge about the ground
    states of $[H_k^q]^\vec{\rho}$ in the nontrivial symmetry sectors.

    \item
    Consider in the following $0<\Omega \leq \OOLB$ and define $Z:=
    \prod_{i\in\V}\sigma_i^z$ so that $Z\ket{\vec{n}} =
    (-1)^{|\vec{n}|}\ket{\vec{n}}$ where $|\vec{n}|=\sum_{i\in\V}n_i$ denotes
    the total number of excitations. With this, we can write $-H_k^q(-\Omega) =
    -ZH_k^q(\Omega)Z$ (note that we do not project onto $\H^q_k$ as in
    \cref{subsec:GroundState}). Note that the matrix $\bra{\vec
    n}[-H_k^q(-\Omega)]\ket{\vec n'}$ is \emph{not} irreducible for $\vec
    n,\vec n'\in \mathbb{Z}_2^{|\V|}$ since $H_k^q(-\Omega)$ does not couple
    excitation numbers outside of $\{k,\ldots,q\}$ on the plaquette $\p$.
    To fix this, we consider the auxiliary matrix $M_{\vec n,\vec n'}(\epsilon) :=
    \bra{\vec n}[-H_k^q(-\Omega)]\ket{\vec{n}'} + \epsilon$ for some $\epsilon >
    0$. This matrix is clearly \emph{essentially positive} (as it has only
    positive off-diagonal elements $\epsilon$ and $\epsilon + \Omega$).
    
    \item
    The essential positivity of $M(\epsilon)$ allows us to invoke the theorem by Perron
    and Frobenius~\cite{perronfrobenius1912}. It asserts that the largest
    eigenvalue $-E_k^q(\Omega;\epsilon)$ of $-H_k^q(-\Omega)
    +\epsilon\sum_{\vec n,\vec n'}\ket{\vec n}\bra{\vec n'}$ is
    \emph{non-degenerate}, and that its eigenvector
    $Z\ket{\psi_k^q(\Omega;\epsilon)} \equiv \sum_{\vec{n} \in
    \mathbb{Z}_2^{|\V|}}c_{\vec{n}}(\Omega;\epsilon)\ket{\vec{n}}$ can be
    chosen such that $c_{\vec{n}}(\Omega;\epsilon) > 0$.
    In turn, the \emph{smallest} eigenvalue $E_k^q(\Omega;\epsilon)$ of
    $H_k^q(\Omega)-\epsilon\sum_{\vec n,\vec n'}(-1)^{|\vec n|+|\vec n'|}\ket{\vec
    n}\bra{\vec n'}$ is non-degenerate, and its eigenvector can be chosen
    such that $\ket{\psi_k^q(\Omega;\epsilon)} = \sum_{\vec{n} \in
    \mathbb{Z}_2^{|\V|}}(-1)^{|\vec{n}|}c_{\vec
    n}(\Omega;\epsilon)\ket{\vec{n}}$ with $c_{\vec{n}}(\Omega;\epsilon) > 0$. 

    In contrast to the proof in \cref{subsec:GroundState}, we cannot conclude
    that $\ket{\psi_k^q(\Omega;\epsilon)}$ is the ground state of $H^q_k$ at
    this point.

    \item
    We know that $U_p(g)$ commutes with $H_k^q(\Omega)$ for all $p \in \P$ and
    $g\in G$. Since the $U_p(g)$ act by permutations on qubits, they also
    commute with the auxiliary operator $H_k^q(\Omega)-\epsilon\sum_{\vec n,\vec
    n'}(-1)^{|\vec n|+|\vec n'|}\ket{\vec n}\bra{\vec n'}$. Hence they leave
    all eigenspaces invariant -- in particular the non-degenerate ground state,
    and we find $U_p(g) \ket{\psi_k^q(\Omega;\epsilon)} =
    e^{\mathrm{i}\varphi_p(g)}\ket{\psi_k^q(\Omega;\epsilon)}$ for some phase
    $\varphi_p(g)$. On the other hand, we know that $U_p(g)\ket{\vec{n}} =
    \ket{\phi_p(g) \cdot \vec{n}}$ acts on the natural basis via a permutation
    $\phi_p(g)$ of the vertices in $p$.  
    Comparing coefficients by their absolute value yields $c_{\phi_p(g)^{-1}
    \cdot \vec{n}}=c_{\vec{n}}$; comparing coefficients by their phase yields
    $(-1)^{|\phi_p(g)^{-1} \cdot \vec n|}=e^{\mathrm{i}\varphi_p(g)}(-1)^{|\vec{n}|}$. 
    Since the number of excited qubits is invariant under permutations, we have
    $|\vec{n}| = |\phi_p(g)^{-1} \cdot \vec{n}|$ and the eigenvalue must be
    $e^{\mathrm{i}\varphi_p(g)}=+1$. We therefore conclude that
    $U_p(g) \ket{\psi_k^q(\Omega;\epsilon)} = +1\ket{\psi_k^q(\Omega;\epsilon)}$
    for all $p \in \P$ and $g\in G$, i.e., the ground state transforms under
    the trivial representation on all plaquettes:
    $\ket{\psi_k^q(\Omega;\epsilon)}\in\H^\s$. 

    \item
    In the limit $\epsilon \rightarrow 0$ we have $M_{\vec n,\vec n'}(\epsilon)
    \rightarrow \bra{\vec n}[-H_k^q(-\Omega)]\ket{\vec{n}'}$, and we recover a
    (now possibly degenerate) normalized ground state of the form
    $\ket{\psi_k^q} = \sum_{\vec{n} \in
    \mathbb{Z}_2^{|\V|}}(-1)^{|\vec{n}|}c_{\vec n}\ket{\vec{n}}$ with
    eigenenergy $E_k^q$ and coefficients $c_{\vec{n}} \ge 0$ in the natural
    basis (note $\geq$ and $\mathbb{Z}_2^{|\V|}$, instead of $>$ and
    $[\mathbb{Z}_2^{|\V|}]_k^q$ as in \cref{subsec:GroundState}). From
    $P^\s\ket{\psi_k^q(\Omega;\epsilon)} = \ket{\psi_k^q(\Omega;\epsilon)}$
    (with $P^\s$ continuous), it follows that
    $P^\s\ket{\psi_k^q}=\ket{\psi_k^q}$ in the limit $\epsilon\rightarrow 0$.
    Since $\ket{\psi_k^q}$ is a global ground state of $H^q_k$, it is in
    particular a ground state of $[H^q_k]^\s$ so that $\ket{\psi_k^q} \in
    [\G_k^q]^\s$. Finally, using $[\G_k^q]^\s \le \H_k^q$, we have shown that
    $\ket{\psi_k^q} \in \H_k^q$ so that the ground state takes the form
    $\ket{\psi_k^q} = \sum_{\vec{n} \in
    [\mathbb{Z}_2^{|\V|}]^q_k}(-1)^{|\vec{n}|}c_{\vec n}\ket{\vec{n}}$ with
    $c_{\vec n}\geq 0$.

    \end{enumerate}
    This concludes the proof.
\end{proof}

\end{document}